\documentclass[lettersize,onecolumn]{IEEEtran}
\usepackage{amsmath,amsfonts}
\usepackage{algorithmic}
\usepackage{algorithm}
\usepackage{array}
\usepackage{caption}
\usepackage{textcomp}

\usepackage[colorlinks,linkcolor=blue,anchorcolor=blue,citecolor=blue]{hyperref}
\usepackage{stfloats}
\usepackage{url}
\usepackage{verbatim}
\usepackage{graphicx}
\usepackage{cite}
\usepackage{makecell}
\usepackage{amssymb}
\usepackage[table]{xcolor}
\usepackage{threeparttable}
\usepackage{diagbox}
\usepackage{multirow}
\usepackage{makecell}
\usepackage{amsthm}
\usepackage{bm}
\theoremstyle{plain}

\newtheorem{theorem}{Theorem}
\newtheorem{lemma}[theorem]{Lemma}

\newtheorem{definition}[theorem]{Definition}
\newtheorem{corollary}[theorem]{Corollary}

\theoremstyle{definition}
\newtheorem{example}{Example}
\newtheorem{construction}{Construction}
\newtheorem{remark}{Remark}

\begin{document}

\title{Optimal-Access Cooperative MSR Codes: Parity-Check Matrix Construction And a Unified Transformation}

\author{ Yaqian Zhang, Jingke Xu, Ya-Feng Liu
\thanks{Yaqian Zhang and Ya-Feng Liu are with Ministry of Education Key Laboratory of Mathematics and Information Networks, School of Mathematical Sciences,
Beijing University of Posts and Telecommunications, Beijing 100876, China (e-mail: zhangyq@bupt.edu.cn, yafengliu@bupt.edu.cn). }
\thanks{Jingke Xu is with School of Information Science and Engineering, Shandong Agricultural University, Tai'an 271018, China (e-mail: xujingke@sdau.edu.cn). }
}

\maketitle

\begin{abstract}
 Cooperative MSR codes are a kind of storage codes which enable optimal-bandwidth repair of any $h\geq2$ node erasures in a cooperative way, while retaining the minimum storage as an $[n,k]$ MDS code. Each code coordinate (node) is assumed to store an array of $\ell$ symbols, where $\ell$ is termed as sub-packetization.
To address the disk IO (input/output) capability, a cooperative MSR code is said to have optimal-access property, if during node repair, the amount of data accessed at each helper node meets a lower bound on this quantity.
Large sub-packetization tends to induce high complexity and large IO cost in practice.

In this paper, we focus on reducing the sub-packetization level of optimal-access cooperative MSR codes.
We propose new constructions of optimal-access cooperative MSR codes through two methods.
At first, we propose a direct explicit construction by designing its parity-check matrix. Such parity-check matrix is built by repeatedly employing two crucial parity-check matrices of two basic MDS array codes as building blocks.
Secondly, we propose a generic transformation framework. Starting from an arbitrary $[n+d-k,d]$ MDS scalar code, one can derive a final cooperative MSR code by systematically applying two basic transformations in this framework.
Both approaches yield $(n,k,\ell=\delta^m)$ optimal-access cooperative MSR codes with $\delta=d-k+h$ and $m=\binom{n}{h}-\lfloor\frac{n}{\delta}\rfloor(\binom{\delta}{h}-1)$.
The derived codes can repair any $h~(2\le h\le n-k)$ erasures using $d~(k\le d \le n-h)$ helper nodes.
Compared with the state of the art (with $\ell=\delta^{\binom{n}{h}}$), the derived codes can reduce the sub-packetization $\ell$ by a fraction of $1/\delta^{\lfloor\frac{n}{\delta}\rfloor(\binom{\delta}{h}-1)}$, where $\delta=d-k+h$.
Moreover, we also show that some previous code structures of optimal-access cooperative MSR codes (Zhang, Zhang \& Wang CL'2020) and optimal-access MSR codes with $h=1$ (Ye \& Barg TIT'2017, and Li, Tang \& Tian TIT'2018) are included as special cases of our transformation construction.
At last, we note that all of the constructions are built over a finite field of linear size $\geq n+d-k$.

\end{abstract}

\begin{IEEEkeywords}
Distributed storage, cooperative MSR codes, repair bandwidth, optimal-access, sub-packetization.
\end{IEEEkeywords}

\section{Introduction}\label{introduction}
In large-scale distributed storage systems (DSS), data is stored across many storage nodes where node failures may occur frequently.
To protect data from node failures, erasure codes are extensively used in DSS.
Typically, A file containing $k$ data blocks is encoded to $n$ blocks using an erasure code with each encoded block stored in one storage node.
The system requires that any $k$ nodes can reconstruct the original file.
Besides, if a node fails, that is, the data stored at that node is erased. Then, the sustaining system requires that the failed node should be repaired by downloading data from $d~(\leq n-1)$ surviving nodes ({\it called helper nodes}).
Two important metrics for the node repair efficiency are the total amount of data downloaded ({\it called repair bandwidth}) and the volume of data accessed at the helper nodes, where the former indicates the network usage and the latter characterizes the disk I/O cost.
A central issue in DSS is how to repair failed nodes with high repair efficiency.
In \cite{Dimakis2011}, Dimakis et al. gave a tradeoff between the storage overhead and repair bandwidth, where codes with parameters lying on this tradeoff curve are called {\it regenerating codes}.
One extreme point of regenerating codes are minimum storage regenerating (MSR) codes, which achieve minimum storage overhead and have been extensively studied in the literature \cite{Ramchandran2011,Kumar2011,Vardy2016,Sasidharan2015,Ye2016,Ye2016sub-,JLi-2018}.

MSR codes can only deal with single node failures.  In some scenarios, multiple node failures are quite common in DSS. For example, in Total Recall \cite{total-recall} a repair process is triggered only after the total number of failed
nodes has reached a predefined threshold.
To repair multiple node failures simultaneously, cooperative regenerating codes are defined in \cite{Hu2010} for repairing $h\geq2$ nodes through a cooperative repair model, wherein cooperative MSR codes have  attracted much attention due to their minimum storage cost.

In the cooperative repair model, when $h$ nodes fail, $h$ newcomers for repairing the $h$ failed nodes independently download data from $d$ helper nodes and then exchange data among themselves.
Specifically, suppose node $i$ stores a vector $\bm{c}_i\in F^{\ell}$ for $i\in[n]$, where $F$ is a finite field. Let $\mathcal{H}\subseteq[n]$ with $|\mathcal{H}|=h$ be the set of failed nodes.
For each $i\in\mathcal{H}$, let $\mathcal{R}_i\subseteq[n]\setminus\mathcal{H}$ with $|\mathcal{R}_i|=d$ be the set of helper nodes connected by node $i$. Denote by $\mathcal{R}=(\mathcal{R}_i)_{i\in\mathcal{H}}$ the $h$ helper node sets. Then the cooperative repair process includes the following two phases:
\begin{itemize}
\item \textbf{Download phase.} For each $i\in\mathcal{H}$ and $j\in\mathcal{R}_i$, node $i$ downloads $\beta_1$ symbols from helper node $j$ by accessing $\alpha_{i,j}^{(\mathcal{H},\mathcal{R})}$ coordinates of node $j$'s storage data $\bm{c}_j=(c_{j,0},\ldots,c_{j,\ell-1})$.
\item \textbf{Collaboration phase.} For each $i\in\mathcal{H}$ and $i'\in\mathcal{H}\setminus\{i\}$, node $i$ downloads $\beta_2$ symbols from node $i'$.
\end{itemize}
Then each node $i\in\mathcal{H}$ should be able to recover its erased data $\bm{c}_i$ using the data downloaded and exchanged in both phases.
In particular, 
if the cooperative repair is linear, then the repair process only involves in  linear operations.
That is, in the download phase, there exist $hd$ repair matrices $S_{i,j}^{(\mathcal{H},\mathcal{R})}\in F^{\beta_1\times \ell}$, $i\in\mathcal{H}$, $j\in\mathcal{R}_i$ with ${\rm rank}(S_{i,j}^{(\mathcal{H},\mathcal{R})})=\beta_1$.
For each $i\in\mathcal{H}$ and $j\in\mathcal{R}_i$, node $i$ downloads the $\beta_1$ symbols $S_{i,j}^{(\mathcal{H},\mathcal{R})}\bm{c}_j^{\top}$ from each helper node $j\in\mathcal{R}_i$.
And $S_{i,j}^{(\mathcal{H},\mathcal{R})}\bm{c}_j^{\top}$ only depends on $\alpha_{i,j}^{(\mathcal{H},\mathcal{R})}$ symbols of $\bm{c}_j$.
In the collaboration phase, there exist $h(h-1)$ repair matrices $T_{i,i'}^{(\mathcal{H},\mathcal{R})}\in F^{\beta_2\times(d\beta_1)}$, $i,i'\in\mathcal{H}$, $i\neq i'$ with ${\rm rank}(T_{i,i'}^{(\mathcal{H},\mathcal{R})})=\beta_2$.
For each $i\in\mathcal{H}$ and $i'\in\mathcal{H}\setminus\{i\}$, node $i$ downloads the $\beta_2$ symbols $T_{i,i'}^{(\mathcal{H},\mathcal{R})}(\bm{c}_jS_{i',j}^{(\mathcal{H},\mathcal{R})\top}:j\in\mathcal{R}_{i'})^{\top}$ from node $i'$.
In order to recover $\bm{c}_i$, $i\in\mathcal{H}$, there exist $h$ reconstruction matrices $W_i^{(\mathcal{H},\mathcal{R})}\in F^{\ell\times(d\beta_1+(h-1)\beta_2)}$, $i\in\mathcal{H}$, s.t. $W_i^{(\mathcal{H},\mathcal{R})}(\bm{c}_jS_{i,j}^{(\mathcal{H},\mathcal{R})\top}: j\in\mathcal{R}_i,~ 
(\bm{c}_jS_{i',j}^{(\mathcal{H},\mathcal{R})\top})_{j\in\mathcal{R}_{i'}}T_{i,i'}^{(\mathcal{H},\mathcal{R})\top}: i'\in\mathcal{H}\setminus\{i\})^{\top}=\bm{c}_i^{\top}$ for all $i\in\mathcal{H}$.
During the cooperative repair model, the total repair bandwidth is defined as $\gamma=h(d\beta_1+(h-1)\beta_2)$ symbols, and the total amount of data accessed is $\gamma_a=\sum_{i\in\mathcal{H}}\sum_{j\in\mathcal{R}_i}\alpha_{i,j}^{(\mathcal{H},\mathcal{R})}$ symbols.

Indeed, cooperative MSR (also MSR) codes belong to a subclass of MDS codes, known as {MDS array codes} \cite{Blaum1998}.
An $(n,k,\ell)$ MDS array code over a finite field $F$ is formed by a set of vectors $(\bm{c}_1,\ldots,\bm{c}_n)$, where each $\bm{c}_i\in F^{\ell}$ $i\in[n]$ is a row vector of length $\ell$.
$\ell$ is called the sub-packetization level.
It satisfies that any $k$ coordinates $\bm{c}_i$ can be seen as information coordinates and can reconstruct the whole codeword (termed as MDS property).
Each coordinate $\bm{c}_i\in F^{\ell}$ is stored in one storage node $i$ for $i\in[n]$.
For MDS array codes, it is shown in \cite{Hu2010,Ye2018} that the repair bandwidth and the amount of accessed data for cooperatively repairing $h$ nodes using $d$ helper nodes are respectively lower bounded by
\begin{equation}\label{bw-bound}
\gamma\geq\frac{(d+h-1)hl}{d-k+h}, \quad \gamma_a\geq\frac{dhl}{d-k+h}.
\end{equation}
If an $(n,k,\ell)$ MDS array code can cooperatively repair any $h$ of the $n$ nodes with $\gamma$ achieving (\ref{bw-bound}) with equality,
then the MDS array code is exactly a cooperative MSR code. Moreover, when both $\gamma$ and $\gamma_a$ meet (\ref{bw-bound}) with equality, the code is said to have optimal-access property and called an optimal-access cooperative MSR code.
In this paper, we focus on the construction of optimal-access cooperative MSR codes.

\subsection{Previous works}
In the literature, constructions of cooperative MSR codes used to restrict to limited parameters\cite{Shum2013,Scouarnec2012,Shum2016}.
Untill the work \cite{Ye2018}, Ye and Barg gave a construction with general parameters, while the sub-packetization level is extraordinarily large.
Then Zhang et al. \cite{Zhang2020} presented the first nontrivial construction of cooperative MSR codes possessing optimal-access property and relatively small sub-packetization.
Lately, scalar cooperative MSR codes and binary cooperative MSR codes are respectively derived in \cite{Zhang2020-scalar} and \cite{Li2025}.
The works in \cite{Ye2020,Liu2023,Zhang2025} are devoted to reduce the sub-packetization level of cooperative MSR codes.
However, except the work \cite{Shum2013,Zhang2020,Li2025}, all the previous constructions do not have the optimal-access property. Here the construction of \cite{Shum2013} is trivial due to $d=k$.
In the optimal-access cooperative MSR codes of \cite{Zhang2020} and \cite{Li2025}, the sub-packetization $\ell=(d-k+h)^{\binom{n}{h}}$ which is still large due to practical consideration.
Thus, how to further reduce the sub-packetization of optimal-access cooperative MSR codes remains an open problem.
We summarize the previous results in Table \ref{tab-comparison}.

\subsection{Our contribution}
In this paper, we focus on reducing the sub-packetization level of optimal-access cooperative MSR codes.
We present explicit constructions of optimal-access cooperative MSR codes with sub-packetization $\ell=\delta^{\binom{n}{h}-\lfloor\frac{n}{\delta}\rfloor(\binom{\delta}{h}-1)}$, which reduces $\ell$ by a fraction of $1/\delta^{\lfloor\frac{n}{\delta}\rfloor(\binom{\delta}{h}-1)}$ compared with that of \cite{Zhang2020,Li2025}, where $\delta=d-k+h$.
The comparison on parameters is illustrated in Table \ref{tab-comparison}.

\begin{table}[!ht]
		\begin{center}
			\caption{Comparison with previous results of cooperative MSR code constructions.}
			\label{tab-comparison}
\begin{tabular}{| c | c | c | c |c|c|}
\hline
Ref. & $n,k$ & repair parameters $h, d$ & sub-packetization $\ell$ & field size & optimal-access  \\ \hline
\cite{Shum2013} & any $(n,k)$ & $d=k$ & $d-k+h$ & $\geq n$ & Yes \\ \hline
\cite{Scouarnec2012} & $(n,k=2)$ & $d=n-h$ & $d-k+h$ & $\geq n$ & No\\ \hline
\cite{Shum2016} & $(n=2k,k)$ & $d=n-h$ & $d-k+h$ & $\geq (d-k+1)n$ & No\\ \hline
\cite{Ye2018} & any $(n,k)$ & all possible $h,d$ & $((d-k)^{h-1}(d-k+h))^{\binom{n}{h}}$ &  $\geq (d-k+1)n$ & No \\
\hline
\cite{Zhang2020} & any $(n,k)$ & all possible $h,d$ & $(d-k+h)^{\binom{n}{h}}$ & $\geq n+d-k$ & Yes \\
\hline
\cite{Zhang2020-scalar} & $(n\geq 2k-1,k)$ & $d\geq\max\{2k-1-h,k\}$ & $d-k+h$ & $\geq (d-k+1)n$ & No\\ \hline
\cite{Li2025} & any $(n,k)$ & all possible $h,d$ & $(d-k+h)^{\binom{n}{h}}$ & $\geq 2$ & Yes \\ \hline
\cite{Ye2020} &  any $(n,k)$ & all possible $h,d$ & $(d-k+h)(d-k+1)^n$ & $\geq(d-k+1)n$ & No \\
\hline
\multirow{2}{*}{\cite{Liu2023}} &\multirow{2}{*}{any $(n,k)$} & $(h+1)\mid 2^n$, $d=k+1$ & $2^n$ & \multirow{2}{*}{$\geq(d-k+1)n$} & \multirow{2}{*}{No} \\
\cline{3-4} & & $(h+1)=(2t+1)2^m$, $d=k+1$ & $(2t+1)2^n$ & &  \\
\hline
\cite{Zhang2025} & any $(n,k)$ & all possible $h,d$ & $(d-k+h)(d-k+1)^{\lceil\frac{n}{2}\rceil}$ & $\geq(d-k+1)n+1$ & No \\
\hline
This work& any $(n,k)$ & all possible $h,d$ & \makecell[c]{$\delta^{\binom{n}{h}-\lfloor\frac{n}{\delta}\rfloor(\binom{\delta}{h}-1)}$, \\ where $\delta=d-k+h$ }& $\geq n+d-k$ & Yes \\
\hline
\end{tabular}
		\end{center}
Note: For the repair parameters $h, d$, all possible $h,d$ means $2\leq h\leq n-k$ and $k\leq d\leq n-h$.
	\end{table}

In our techniques, we construct codes through parity-check matrix construction, as well as a unified transformation framework.
Both methods give cooperative MSR codes with same parameters.
More specifically,

\begin{itemize}
\item[(1)]
In the parity-check matrix perspective,
we give an explicit construction of optimal-access cooperative MSR codes by designing its parity-check matrix.
Such parity-check matrix is built from two crucial basic code structures $\mathcal{C}_{\mathrm{I}}$ and $\mathcal{C}_{\mathrm{II}}$ (see Section \ref{sec0-building-block}).
$\mathcal{C}_{\mathrm{I}}$ is designed for cooperatively repairing one specific erasure pattern of $h$ erasures, and $\mathcal{C}_{\mathrm{II}}$ is for repairing $\binom{d-k+h}{h}$ specific erasure patterns. Both $\mathcal{C}_{\mathrm{I}}$ and $\mathcal{C}_{\mathrm{II}}$ are MDS array codes with sub-packetization $\ell=d-k+h$.
Then, by extending the code structure $\mathcal{C}_{\mathrm{I}}$ for $\binom{n}{h}-\lfloor\frac{n}{\delta}\rfloor\binom{\delta}{h}$  times and $\mathcal{C}_{\mathrm{II}}$ for $\lfloor\frac{n}{\delta}\rfloor$ times, where $\delta=d-k+h$, we finally derive the optimal-access cooperative MSR code for repairing any $h$ node erasures (all $\binom{n}{h}$ erasure patterns).
It is worth noting that this construction generalizes our previous  work \cite{Zhang2026} at ISIT'2026 for the case $h=2$ and $d=n-2$.

\item[(2)]
In the transformation perspective,
we present a generic transformation framework for building optimal-access cooperative MSR codes.
That is, from an arbitrary $[n+d-k,d]$ MDS scalar code, one can directly obtain an cooperative MSR code by applying the generic transformation.
To this end, we firstly design two basic transformations $\mathcal{T}_1$ and $\mathcal{T}_2$.
From an MDS scalar code, each time we apply the transformation $\mathcal{T}_1$ (w.r.t. $\mathcal{T}_2$), the newly obtained code can repair one (w.r.t. $\binom{d-k+h}{h}$) more erasure patterns with the cost of extending the sub-packetization to $d-k+h$ times.
Then, by compositing $\mathcal{T}_1$ for $\binom{n}{h}-\lfloor\frac{n}{\delta}\rfloor\binom{\delta}{h}$ times and $\mathcal{T}_2$ for $\lfloor\frac{n}{\delta}\rfloor$ times where $\delta=d-k+h$, the generic transformation framework is derived.

\item[(3)]
We also establish a connection between the transformation framework and the parity-check matrix construction.
We show that the parity-check matrix code can be obtained by the generic transformation from a Reed-Solomon-type MDS scalar code.
Moreover, we also give connections between our construction and some previous code structures of cooperative MSR codes \cite{Zhang2020} and MSR codes with $h=1$ \cite{Ye2016sub-,JLi-2018}.
We show that our construction includes those previous code structures as special cases.
The details are given in Subsection \ref{comparison}.
\end{itemize}

Indeed, in our construction, the key to reducing the sub-packetization is to partition all the $\binom{n}{h}$ $h$-node erasure patterns into $\lfloor\frac{n}{\delta}\rfloor$ intra-groups of  erasure patterns $\mathcal{P}_u$, $1\leq u\leq \lfloor\frac{n}{\delta}\rfloor$ (i.e., each intra-group $\mathcal{P}_u$ contains $\binom{\delta}{h}$ erasure patterns) and $\binom{n}{h}-\lfloor\frac{n}{\delta}\rfloor\binom{\delta}{h}$ inter-group patterns $\mathcal{P}_0$ (i.e., the remaining erasure patterns).
Then, based on the basic codes $\mathcal{C}_{\mathrm{I}}$, $\mathcal{C}_{\mathrm{II}}$ and basic transformations $\mathcal{T}_1$, $\mathcal{T}_2$,
every time we extend the code dimension to $\delta$ times using  the structure $\mathcal{C}_{\mathrm{I}}$ (w.r.t. $\mathcal{T}_1$), one more erasure pattern in $\mathcal{P}_0$ can be repaired.
And every time we extend the code dimension  to $\delta$ times  using the structure $\mathcal{C}_{\mathrm{II}}$ (w.r.t. $\mathcal{T}_2$), another $\binom{\delta}{h}$ erasure patterns in $\mathcal{P}_u$ can be repaired.
Thus, the final code with $\ell=\delta^m$ is derived through sub-packetization extension of $m$ times where $m=\binom{n}{h}-\lfloor\frac{n}{\delta}\rfloor(\binom{\delta}{h}-1)$.
And the reduction on sub-packetization mainly benefits from the parallel repair of the $\binom{\delta}{h}$ erasure patterns in $\mathcal{C}_{\mathrm{II}}$ and $\mathcal{T}_2$, compared with \cite{Zhang2020,Li2025}.

\subsection{Organization}
The rest of the paper is organized as follows.
Section \ref{sec0-building-block} designs two types of MDS array code building blocks $\mathcal{C}_{\mathrm{I}}$ and $\mathcal{C}_{\mathrm{II}}$.
Section \ref{general-h} presents the general construction of  optimal-access cooperative MSR codes through parity-check matrix construction.
Then from the transformation perspective,
Section \ref{transform-1} and Section \ref{transform-2} present  the basic transformations $\mathcal{T}_1$ and $\mathcal{T}_2$, respectively.
Based on composition of $\mathcal{T}_1$ and $\mathcal{T}_2$, Section \ref{generic-transform} derives the generic transformation for building optimal-access cooperative MSR codes.
At last,
Section \ref{conclusion} gives some discussion and concludes the paper.

\section{Two types of MDS array code building blocks}\label{sec0-building-block}
\subsection{Notations}
Throughout this paper, we use $[n]$ to denote the set of integers $\{1,2,...,n\}$ for a positive integer $n$, and denote $[i,j]=\{i,i+1,...,j\}$ for two integers $i\leq j$.
Let $F$ denote a finite field.
 Given an $(n,k,\ell)$ MDS array code $\mathcal{C}$ over $F$, for each codeword $\bm{c}\in\mathcal{C}$, we write $\bm{c}=(\bm c_1,...,\bm c_n)$ where $\bm c_i=(c_{i,0},c_{i,1},\ldots,c_{i,\ell-1})\in F^{\ell}$ for $i\in[n]$. Each coordinate $\bm{c}_i$ is called a node.
Note that the bold letters, suc as $\bm{c}, \bm{c}_i$, etc. always denote row vectors. $\bm{c}^{\top}$ denotes the transpose of $\bm{c}$.
Let $I_{\ell}$ represent the identity matrix with order $\ell$.

We will define an $(n,k,\ell)$ MDS array code $\mathcal{C}$ over $F$ by giving its parity-check matrix $H$.
Specifically, write
\begin{equation}\label{def1}
H=\begin{pmatrix}
H_{1,1} &H_{1,2}& \cdots &H_{1,n}\\
H_{2,1} &H_{2,2}& \cdots &H_{2,n}\\
\vdots&\vdots&\ddots&\vdots \\
H_{r,1} &H_{r,2}& \cdots &H_{r,n}
\end{pmatrix}
\end{equation}
where $r=n-k$ and $H_{t,i}$ is an $\ell\times \ell$ matrix over $F$ for $t\in[r]$ and $i\in[n]$.
That is, $\mathcal{C}$ can be defined as $\mathcal{C}=\{(\bm c_1,...,\bm c_n)\in(F^\ell)^n:~H\cdot(\bm c_1,...,\bm c_n)^{\top}=\bm{0}\}$.
The MDS property of $\mathcal{C}$ indicates that any $r$  of the $n$ column blocks of $H$ form an invertible $r\ell\times r\ell$ matrix, equivalently, any
$k$ nodes are able to recover the whole codeword.
In the following, we always denote $r=n-k$.

\subsection{Two types of MDS array code building blocks}\label{base-code}

We present two types of MDS array codes $\mathcal{C}_{\mathrm{I}}$ and $\mathcal{C}_{\mathrm{II}}$ which later serve as building blocks for constructing optimal-access cooperative MSR codes with  $2\leq h\leq n-k$ and $k\leq d\leq n-h$.
The two codes are $(n,k,\ell=d-k+h)$ MDS array codes with cooperative repair of specific $h$ erasures. More precisely,
The first code can cooperatively repair the $h$ nodes $\{1,2,\ldots,h\}$ using $d$ helper nodes, and the second code can cooperatively repair any $h$ nodes in the set $\{1,2,\ldots,d-k+h\}$, i.e., in total $\binom{d-k+h}{h}$ erasure patterns, using $d$ helper nodes.
Next we define the two code constructions in Construction \ref{construction-1} and Construction \ref{construction-2} respectively, by giving their parity-check matrices.

\begin{construction}[Type-I MDS array code  $\mathcal{C}_{\mathrm{I}}$]\label{construction-1}
{\it
Let $F$ be a finite field with $|F|\geq n+d-k$, and $\lambda_1,\lambda_2,\ldots,\lambda_n,$ $\gamma_1,\ldots,\gamma_{d-k}$ be $n+d-k$ distinct elements in $F$.
The Type-I $(n,k,\ell=d-k+h)$ MDS array code $\mathcal{C}_{\mathrm{I}}$ is defined by the parity-check matrix $H$ with the form in (\ref{def1}), where for $t\in[r]$, $H_{t,j}=\lambda_{j}^{t-1}I_{\ell}$ for $j\in[h+1,n]$ and $(H_{t,1},H_{t,2},\ldots,H_{t,h})$ is defined as:
{\setlength{\arraycolsep}{2pt}
\begin{equation}\label{pc-block1}
\resizebox{0.9\hsize}{!}{$
\left(\begin{array}{ccccccc|ccccccc|ccc|ccccccc}
\lambda_1^{t-1} &  & & & \gamma_1^{t-1} &\cdots & \gamma_{d-k}^{t-1} & 
\lambda_2^{t-1} &  & & &   &  &  & 
 & & &  
 \lambda_h^{t-1} &  & & & & & \\ 
&  \lambda_1^{t-1}  & & & &  &   & 
&  \lambda_2^{t-1}  & & &   \gamma_1^{t-1} &  \cdots&  \gamma_{d-k}^{t-1} & 
 & & &  
&  \lambda_h^{t-1}  & & & & & \\  
&  & \ddots& & &  & & 
&  & \ddots & &   &  &  & 
 & & &  
 &  & \ddots & & & & \\ 
 &  & &  \lambda_1^{t-1}& && & 
&  & & \lambda_2^{t-1} &   &  &  & 
 &\cdots& &  
 &  & &  \lambda_h^{t-1}& \gamma_1^{t-1} &\cdots  & \gamma_{d-k}^{t-1} \\ 
  &  & &  & \lambda_1^{t-1}&& & 
&  & & &  \lambda_2^{t-1}  &  &  & 
 & & &  
 &  & & & \lambda_h^{t-1}&  &  \\ 
   &  & &  &  &\ddots& &
&  & & &   &   \ddots&  &
 & & &
 &  & & && \ddots&  \\
    &  & &  &  && \lambda_1^{t-1} &
&  & & &   &  &   \lambda_2^{t-1}&
 & & &
 &  & & &&  &  \lambda_h^{t-1} \\
\end{array}\right),
$}
\end{equation}
}
where the empty positions in (\ref{pc-block1}) represent zeros.
}
\end{construction}

Next, we illustrate the MDS property and cooperative repair property of $\mathcal{C}_{\mathrm{I}}$ defined in Construction \ref{construction-1} as follows.

\begin{itemize}
\item[1)] {\it MDS property.}

The MDS property of the code $\mathcal{C}_{\mathrm{I}}$ is straightforward. First, by the last $d-k$ rows of the matrix $(H_{t,1},\ldots,H_{t,n}), t\in[r]$, one can see that the punctured code of $\mathcal{C}_{\mathrm{I}}$ by deleting the first $h$ symbols $c_{i,0},\ldots, c_{i,h-1}$ from each coordinate $\bm{c}_i$ for $i\in[n]$ (i.e., $\{(c_{i,h},\ldots,c_{i,\ell-1})_{1\leq i\leq n}: (c_{i,0},\ldots,c_{i,\ell-1})_{1\leq i\leq n}\in \mathcal{C}_{\mathrm{I}}\}$) forms an $(n,k,d-k)$ MDS array code. Then, according to the first $h$ rows of $(H_{t,1},\ldots,H_{t,n}), t\in[r]$ and substituting the punctured code into it, one can obtain the MDS property of $\mathcal{C}_{\mathrm{I}}$.

\item[2)]  {\it Cooperative repair of $\{1,2,\ldots,h\}$. }

Suppose the $h$ nodes $\{1,2,\ldots,h\}$ are erased. Let $\mathcal{R}_i$, $i\in[h]$ be the set of $d$ helper nodes connected by node $i$.
For each $i\in[h]$, by the $i$-th row of $(H_{t,1},\ldots,H_{t,n}), t\in[r]$, one can obtain the following parity-check equations
$$
\sum_{j\in[n]}\lambda_{j}^{t-1}c_{j,i-1}+\gamma_1^{t-1}c_{i,h}+\cdots+\gamma_{d-k}^{t-1}c_{i,d-k+h-1}=0, ~~~~t\in[r].
$$
This implies that $(c_{1,i-1}, \ldots,c_{n,i-1}, c_{i,h},\ldots,c_{i,d-k+h-1})$ constitutes an $[n+d-k, d]$ generalized Reed-Solomon code (GRS) codeword.
Then in the download phase, for $i\in[h]$, node $i$  downloads the symbol $c_{p,i-1}$ from each helper node $p\in\mathcal{R}_i$, and thereby obtains the data $\{c_{j,i-1}: j\in[h]\}\cup\{c_{i,h}, \dots, c_{i,d-k+h-1}\}$.

In the collaboration phase, for each $i\in[h]$ and $i'\in[h]\setminus\{i\}$, node $i'$ transmits the symbol $c_{i,i'-1}$ to node $i$.
Thus the $h$ nodes can be repaired.

\end{itemize}

\begin{construction}[Type-II MDS array code  $\mathcal{C}_{\mathrm{II}}$]\label{construction-2}
{\it
Let $F$ be a finite field with $|F|>n$, and $\lambda_1,\lambda_2,\ldots,\lambda_n,\tau\in F$ such that $\lambda_i, i\in[n]$ are all distinct and $\tau\neq 0, 1$.
The Type-II $(n,k,\ell=d-k+h)$ MDS array code $\mathcal{C}_{\mathrm{II}}$ is defined by the parity-check matrix $H$ with the form in (\ref{def1}), where for $t\in[r]$, $H_{t,j}=\lambda_{j}^{t-1}I_{\ell}$ for $j\in[d-k+h+1,n]$ and $(H_{t,1},\dots,H_{t,d-k+h})=$
{\setlength{\arraycolsep}{3pt}
\begin{equation}\label{pc-block2}
\resizebox{0.9\hsize}{!}{$
\begin{aligned}
&\left(\begin{array}{ccccc|ccccc|c|ccccc}
\lambda_1^{t-1} & \lambda_2^{t-1} & \lambda_3^{t-1} &\ldots & \lambda_{d-k+h}^{t-1} & \lambda_2^{t-1} &  &  && & &
\lambda_{d-k+h}^{t-1} &  & & &  \\
 & \lambda_1^{t-1} &  && & \tau\lambda_1^{t-1} & \lambda_2^{t-1} & \lambda_3^{t-1} &\ldots & \lambda_{d-k+h}^{t-1} & &
 & \lambda_{d-k+h}^{t-1} & & &  \\
 & &  \lambda_1^{t-1} && &  &  & \lambda_2^{t-1} &  & & \cdots
 & & & \lambda_{d-k+h}^{t-1} & &  \\
 & &  & \ddots & &  &  &  & \ddots &  &
 & & & & \ddots & \\
 & &  &  & \lambda_1^{t-1} &  &  &  & & \lambda_2^{t-1} &
 & \tau\lambda_1^{t-1} & \tau\lambda_2^{t-1} & \tau\lambda_3^{t-1} & \ldots & \lambda_{d-k+h}^{t-1}\\
\end{array}\right),
\end{aligned}
$}
\end{equation}
}
where the empty positions in (\ref{pc-block2}) represent zeros.
And note that in (\ref{pc-block2}) the element $\tau$ only appears in the $(i,1),\ldots,(i,i-1)$-th entries of $H_{t,i}$ for $i\in[d-k+h]$.
}
\end{construction}
Next we prove the MDS property and cooperative repair property of $\mathcal{C}_{\mathrm{II}}$ defined in Construction \ref{construction-2}.
\begin{itemize}

\item[1)] {\it MDS property.}

It suffices to show that every choice of $r$ column blocks of $H$ forms an invertible matrix, denoted by $H(i_1,i_2,\ldots,i_r)$ for  $1\leq i_1<i_2< \cdots<i_r\leq n$. That is, we  prove for any $\bm{x}\in(F^{d-k+h})^r$, $H(i_1,i_2,\ldots,i_r)\cdot \bm{x}^{\top}=\bm{0}$ always implies $\bm{x}=\bm{0}$.
Denote $\bm{x}=(\bm{x}_1,\ldots,\bm{x}_r)$ and $\bm{x}_i=(x_{i,1},\ldots,x_{i,d-k+h})$ for $i\in[r]$, and  suppose $i_1,\ldots,i_g\in[d-k+h]$ and $i_{g+1},\ldots,i_r\in[d-k+h+1,n]$ for some $0\leq g\leq d-k+h$.

For each $a\in[d-k+h]\setminus\{i_1,\ldots,i_g\}$, according to the $a$-th row of $(H_{t,i_1},\ldots,H_{t,i_r})\bm{x}^{\top}=\bm{0}$, $t\in[r]$, one can obtain that
$
\sum_{j=1}^{r}\lambda_{i_j}^{t-1}x_{j,a}=0, t\in[r].
$
This implies $x_{j,a}=0$ for all $j\in[r], a\in[d-k+h]\setminus\{i_1,\ldots,i_g\}$.

Now consider symbols $x_{j,i_s}$, $j\in[r],s\in[g]$. For each $s\in[g]$, by the $i_s$-th row of $(H_{t,i_1},\ldots,H_{t,i_r})\bm{x}^{\top}=\bm{0}$, $t\in[r]$ and noticing that $x_{s,a}=0$ for all $ a\in[d-k+h]\setminus\{i_1,\ldots,i_g\}$, then one can obtain the parity-check equations in (\ref{eq-1-0}) as follows: for $t\in[r]$,
\begin{equation}\label{eq-1-0}
\begin{aligned}
\sum_{1\leq j<s}\lambda_{i_{j}}^{t-1}(x_{j,i_s}\!+\!\tau x_{s,i_{j}})\!+\!\sum_{s<j\leq g}\lambda_{i_{j}}^{t-1}(x_{j,i_s}\!+\!x_{s,i_{j}})
+\lambda_{i_{s}}^{t-1}x_{s,i_s}+\sum_{j=g+1}^{r}\lambda_{i_{j}}^{t-1}x_{j,i_s} =0.
\end{aligned}
\end{equation}
Thus, the following set of $r$ symbols:
\begin{equation}\label{eq-1-0-0}
\{x_{s,i_s}\}\cup\{x_{j,i_s}+\tau x_{s,i_{j}}\}_{1\leq j<s}\cup\{x_{j,i_s}+x_{s,i_{j}}\}_{s<j\leq g}\cup\{x_{j,i_s}\}_{j\in[g+1,r]}
\end{equation}
are solved to be zeros by (\ref{eq-1-0}).
That is, for each $s\in[g]$, one can directly compute $x_{s,i_s}=0$ and $x_{j,i_s}=0,j\in[g+1,r]$, which implies that ${\bm x}_j={\bf 0}$ for all $j\in[g+1,r]$.
As for the symbol sums, by considering equations labeled by $i_{s}$ and $i_{j}$ where $s,j\in[g]$ and $s\neq j$, one has
\begin{equation}\label{eq-1-1}
\begin{cases}
x_{j,i_s}+\tau x_{s,i_{j}}=0 \\
x_{s,i_{j}}+ x_{j,i_{s}}=0
\end{cases}
\mathrm{if}~j<s;~
\begin{cases}
x_{s,i_j}+ x_{j,i_{s}}=0 \\
x_{s,i_{j}}+\tau x_{j,i_s}=0
\end{cases}
\mathrm{if}~j>s.
\end{equation}
Since $\tau\neq0,1$ and $x_{s,i_s}=0$, it has that $x_{s,i_{j}}=0$ for all $s,j\in [g]$. Hence, ${\bm x}_s={\bf 0}$ for all $s\in[g]$.
This completes the proof.

\item[2)] {\it Cooperative repair of any $h$ nodes in $[d-k+h]$.}

W.L.O.G., suppose the $h$ nodes $\{1,2,\ldots,h\}$ are erased.
For $i\in[h]$, let $\mathcal{R}_i$ with $([d-k+h]\setminus[h])\subseteq\mathcal{R}_i$ be the set of $d$ helper nodes connected by node $i$.
For each $i\in[h]$, according to the $i$-th row of the parity-check equations $(H_{t,1},\ldots,H_{t,n})\bm{c}^{\top}=\bm{0}$ for $t\in[r]$, one obtains that for $t\in[r]$,
\[
\sum_{j=1}^{i-1}\lambda_{j}^{t-1}(c_{j,i-1}+\tau c_{i,j-1})+\lambda_i^{t-1}c_{i,i-1}+\sum_{j=i+1}^{d-k+h}\lambda_{j}^{t-1}(c_{j,i-1}+c_{i,j-1})+\sum_{j=d-k+h+1}^{n}\lambda_j^{t-1}c_{j,i-1}=0.
\]
That is, the following vector
\begin{equation}\label{type-2-repair}
(c_{1,i-1}+\tau c_{i,0},\ldots, c_{i-1,i-1}+\tau c_{i,i-2}, c_{i,i-1}, c_{i+1,i-1}+c_{i,i} \ldots,c_{d-k+h,i-1}+c_{i,d-k+h-1},c_{d-k+h+1,i-1},\ldots,c_{n,i-1})
\end{equation}
forms an $[n,k]$ GRS codeword.
Then, in the download phase, for $i\in[h]$, node $i$ downloads the symbols $c_{j,i-1}$ from each  helper node $j\in\mathcal{R}_i$.
Using the $k$ downloaded symbols $\{c_{j,i-1}: j\in\mathcal{R}_i\setminus([d-k+h]\setminus[h])\}$, node $i$ can reconstruct the whole codeword in \eqref{type-2-repair}.
Moreover, since $([d-k+h]\setminus[h])\subseteq\mathcal{R}_i$ and $\{c_{j,i-1}: j\in[d-k+h]\setminus[h]$ are known, then node $i$ can further recover the following symbols:
\[
\{c_{j,i-1}+\tau c_{i,j-1}\}_{1\leq j\leq i-1}
\cup\{c_{j,i-1}+c_{i,j-1}\}_{i+1\leq j\leq h}
\cup\{c_{i,j-1}\}_{j\in\{i\}\cup[h+1,d-k+h]}.
\]

In the collaboration phase, for each $i\in[h]$ and $i'\in[h]\setminus\{i\}$, node $i'$ transmits the symbol $c_{i,i'-1}+\tau c_{i',i-1}$ to node $i$ if $i<i'$, and transmits the symbol $c_{i,i'-1}+c_{i',i-1}$ to node $i$ if $i>i'$.
Then node $i$ can solve out the unknown  symbols $\{c_{i,j-1}: j\in[h]\setminus{i}\}$ since $\tau\neq 0,1$.
Thus the $h$ nodes can be repaired.
\end{itemize}

In the next section, based on the two codes $\mathcal{C}_{\mathrm{I}}$ and $\mathcal{C}_{\mathrm{II}}$, we construct an optimal-access cooperative MSR code with $h$ erasures and $d$ helper nodes.


\section{Optimal-access cooperative MSR codes with $2\leq h\leq n-k$ and $k\leq d\leq n-h$}\label{general-h}

In this section, let $2\leq h\leq n-k$ and $k\leq d\leq n-h$.
We present an optimal-access cooperative MSR code with $h$ erasures and $d$ helper nodes.
The code is an $(n,k,\ell=\delta^{m})$ MDS array code where $\delta=d-k+h$ and $m=\binom{n}{h}-\lfloor\frac{n}{\delta}\rfloor(\binom{\delta}{h}-1)$, which is constructed by stacking up of the two MDS array code building blocks $\mathcal{C}_{\mathrm{I}}$ and $\mathcal{C}_{\mathrm{II}}$ in Section \ref{sec0-building-block} for several times.
We will define the code by designing its parity-check matrix, and before that, some notations and definitions are needed.

\subsection{Notations and definitions}\label{notation-h}

\begin{itemize}
\item
Let $2\leq h\leq n-k$ and $k\leq d\leq n-h$.
Denote $\delta=d-k+h$ and $m=\binom{n}{h}-\lfloor\frac{n}{\delta}\rfloor(\binom{\delta}{h}-1)$, then $\ell=\delta^m$.
For each $a\in[0,\ell-1]$, $a$ can be uniquely represented as $a=\sum_{i=1}^{m}a_{i}\delta^{i-1}$, where $a_i\in[0,\delta-1]$ for $i\in[m]$, then write $a=(a_1,a_2,\ldots, a_m)$ for simplicity.
For some $i\in[m]$ and $v\in[0,\delta-1]$, denote $a(i,v)=(a_1,\ldots,a_{i-1},v,a_{i+1},\ldots,a_m)$.
\item
The $n$ storage nodes are indexed from $1$ to $n$. We give a partition of the $n$ nodes in $[n]$.
That is, the $n$ nodes are partitioned into $\lceil\frac{n}{\delta}\rceil$ groups, where for $i\in[\lfloor\frac{n}{\delta}\rfloor]$, group $i$ contains $\delta$ consecutive nodes $[(i-1)\delta+1, i\delta]$, and the last group (if have) contains nodes in $[\delta\lfloor\frac{n}{\delta}\rfloor+1,n]$.
Besides, for every node $j\in[\delta\lfloor\frac{n}{\delta}\rfloor]$, we write $j=(u,v)$ to indicate that node $j$ is the $(v+1)$-th node in the $u$-th group, where $u\in[\lfloor\frac{n}{\delta}\rfloor]$ and $v\in[0,\delta-1]$ are the unique integers satisfying $j=(u-1)\delta+v+1$.

\item
Let $\mathcal{P}=\{(j_1,j_2,\ldots,j_h): 1\leq j_1<j_2<\cdots< j_h\leq n\}$ represent the set of all $\binom{n}{h}$ $h$-tuples ($h$-node erasure patterns).
Based on the partition of $[n]$,
for $i\in[\lfloor\frac{n}{\delta}\rfloor]$, define $\mathcal{P}_i$ to be the set of $\binom{\delta}{h}$ $h$-tuples within group $i$, i.e.,
$\mathcal{P}_i=\{(j_1,j_2,\ldots,j_h): (i-1)\delta+1\leq j_1<\cdots<j_h\leq i\delta\}$ and $|\mathcal{P}_i|=\binom{\delta}{h}$.
Moreover, $\mathcal{P}_i$'s, $i\in[\lfloor\frac{n}{\delta}\rfloor]$ are disjoint subsets of $\mathcal{P}$. Define $\mathcal{P}_0=\mathcal{P}\setminus(\cup_{i\in[\lfloor\frac{n}{\delta}\rfloor]}\mathcal{P}_i)$, then $|\mathcal{P}_0|=\binom{n}{h}-\lfloor\frac{n}{\delta}\rfloor\binom{\delta}{h}=m-\lfloor\frac{n}{\delta}\rfloor$.
We call an $h$-node erasure pattern $\mathcal{H}$ to be an intra-group erasure pattern  if $\mathcal{H}\in\mathcal{P}_i$ for  some $i\in[\lfloor\frac{n}{\delta}\rfloor]$, and it is called an inter-group erasure pattern if $\mathcal{H}\in\mathcal{P}_0$.

\item
Define $\pi$ to be a surjective map from $\mathcal{P}$ to the set $[m]$ satisfying the following conditions.
\begin{itemize}
\item[(1)] For $i\in[\lfloor\frac{n}{\delta}\rfloor]$, $\pi$ maps $h$-tuples in $\mathcal{P}_i$ to integer $i$. That is, $\pi(j_1,j_2,\ldots,j_h)=i$ if $(j_1,j_2,\ldots,j_h)\in\mathcal{P}_i$ for $i=1,2,\ldots,\lfloor\frac{n}{\delta}\rfloor$.
\item[(2)] $\pi$ maps $h$-tuples in $\mathcal{P}_0$ to integers in $[\lfloor\frac{n}{\delta}\rfloor+1,m]$ and it is a one to one mapping. This can be done since $|\mathcal{P}_0|=m-\lfloor\frac{n}{\delta}\rfloor$.
\end{itemize}

\item
For each $j\in[n]$, define $\Omega_{j,0},\ldots,\Omega_{j,h-1}$ to be the following $h$ subsets of $[\lfloor\frac{n}{\delta}\rfloor+1,m]$, i.e., for $e\in[0,h-1]$, define\\
$$
\Omega_{j,e}=\{\pi(j_0,\ldots, j_{e}, \ldots, j_{h-1}):~j_{e}=j~ \mathrm{and}~(j_0,\ldots, j_{e}, \ldots, j_{h-1})\in\mathcal{P}_0\}.
$$
\end{itemize}

We give an example to illustrate the notations and definitions.
\begin{example}\label{ex1}
Let $n=7, k=2, d=4, h=3$. Then $\delta=5$, $m=26$ and $\ell=5^{26}$.
Every integer $a\in[0,\ell-1]$ is represented by a vector $(a_1,a_2,\ldots,a_{26})$ with $a_i\in[0,4]$, $i\in[26]$, where $a=\sum_{i=1}^{26}a_i5^{i-1}$.
The $7$ nodes are partitioned into two groups with group $1$ containing $\delta=5$ nodes $\{1,2,\ldots,5\}$ and group $2$ containing the remaining two nodes $\{6,7\}$.
For each node $j\in[5]$, we write $j=(1,j-1)$ for simplicity.

Denote by $\mathcal{P}$ the set of all $\binom{7}{3}=35$ triples (3-node erasure patterns) of nodes in $\{1,\ldots,7\}$.
Denote by $\mathcal{P}_1$ the subset of $\mathcal{P}$ containing all $\binom{5}{3}=10$ triples of nodes in $\{1,\ldots,5\}$.
And denote $\mathcal{P}_0=\mathcal{P}\setminus\mathcal{P}_1$ to be the set of the remaining 25 triples.
We give a surjective map $\pi$ from $\mathcal{P}$ to the set $[m]=[26]$ satisfying (1) and (2), displayed in Table \ref{map-pi}.
 \begin{table}[ht] 
  \renewcommand{\arraystretch}{1.4}
	\begin{center}
\caption{The surjective map $\pi$ on $\mathcal{P}=\mathcal{P}_0\cup\mathcal{P}_1$}
\label{map-pi}
\setlength{\tabcolsep}{1mm}{
\begin{tabular}{c|ccccc}
\hline
\multirow{2}{*}{$\pi$ on $\mathcal{P}_1$}& $\pi(1,2,3)\!=\!1$& $\pi(1,2,4)\!=\!1$ & $\pi(1,2,5)\!=\!1$ & $\pi(1,3,4)\!=\!1$ & $\pi(1,3,5)\!=\!1$
\\
& $\pi(1,4,5)\!=\!1$ & $\pi(2,3,4)\!=\!1$ &$\pi(2,3,5)\!=\!1$ & $\pi(2,4,5)\!=\!1$ & $\pi(3,4,5)\!=\!1$
\\ [0.8ex] \hline
\multirow{5}{*}{$\pi$ on $\mathcal{P}_0$}&$\pi(1,2,6)\!=\!2$& $\pi(1,4,6)\!=\!7$  &$\pi(2,3,6)\!=\!12$ & $\pi(2,5,6)\!=\!17$ & $\pi(3,5,6)\!=\!22$
\\
 &$\pi(1,2,7)\!=\!3$ & $\pi(1,4,7)\!=\!8$ & $\pi(2,3,7)\!=\!13$ & $\pi(2,5,7)\!=\!18$ & $\pi(3,5,7)\!=\!23$
 \\
 &$\pi(1,3,6)\!=\!4$  &$\pi(1,5,6)\!=\!9$ &$\pi(2,4,6)\!=\!14$ &$\pi(3,4,6)\!=\!19$ &$\pi(4,5,6)\!=\!24$
\\
 &$\pi(1,3,7)\!=\!5$  &$\pi(1,5,7)\!=\!10$ &$\pi(2,4,7)\!=\!15$ &$\pi(3,4,7)\!=\!20$ &$\pi(4,5,7)\!=\!25$
 \\
  &$\pi(1,6,7)\!=\!6$  &$\pi(2,6,7)\!=\!11$ &$\pi(3,6,7)\!=\!16$ &$\pi(4,6,7)\!=\!21$ &$\pi(5,6,7)\!=\!26$
\\ \hline
\end{tabular}}
\end{center}
\end{table}

For each node $j\in[7]$, we define $3$ subsets $\Omega_{j,0}$, $\Omega_{j,1}$ and $\Omega_{j,2}$ of the set  $\pi(\mathcal{P}_0)=[2,26]$.
$\Omega_{j,0}$ contains integers in $[2,26]$ whose preimage under $\pi$ has the form $(j,j_1,j_2)\in\mathcal{P}_0$ where $j<j_1<j_2\leq n$.
$\Omega_{j,1}$ contains integers in $[2,26]$ whose preimage under $\pi$ has the form $(j_0,j,j_2)\in\mathcal{P}_0$ where $1\leq j_0<j<j_2\leq n$.
$\Omega_{j,2}$ contains integers in $[2,26]$ whose preimage under $\pi$ has the form $(j_0,j_1,j)\in\mathcal{P}_0$ where $1\leq j_0<j_1<j$.
Take $j=4$ for example, it has
\[\begin{aligned}
&\Omega_{4,0}=\{21, 24, 25\},\\
&\Omega_{4,1}=\{7, 8, 14, 15, 19, 20\},\\
&\Omega_{4,2}=\emptyset.
\end{aligned}\]
\end{example}

In the following subsection, we give the general code construction.

\subsection{Code construction}\label{sec-construction-h}
We give the optima-access cooperative MSR code construction with $2\leq h\leq n-k$ and $k\leq d\leq n-h$.
Actually, recall that in Section \ref{sec0-building-block},
$\mathcal{C}_{\mathrm{I}}$ can cooperatively repair one erasure pattern $\{1,2,\ldots,h\}$, and $\mathcal{C}_{\mathrm{II}}$ can cooperatively repair $\binom{\delta}{h}$ erasure patterns, i.e., any $h$ nodes in $[\delta]$.
Both $\mathcal{C}_{\mathrm{I}}$ and $\mathcal{C}_{\mathrm{II}}$ have sub-packetization $\delta$.
In order to construct a general code  with cooperative repair of any $h$ erasures, i.e., $\binom{n}{h}$ erasure patterns, we extend the sub-packetization to $m$ dimensions, i.e., $\ell=\delta^m=\delta^{\binom{n}{h}-\lfloor\frac{n}{\delta}\rfloor(\binom{\delta}{h}-1)}$.
Then each $a\in[0,\ell-1]$ can be written as $a=(a_1,\ldots,a_m)$.
According to the map $\pi$ in Subsection \ref{notation-h}, each coordinate $u\in[\lfloor \frac{n}{\delta}\rfloor]$ is used to repair the $\binom{\delta}{h}$ intra-group erasure patterns in $\mathcal{P}_u$ (as in $\mathcal{C}_{\mathrm{II}}$), while each of the remaining $m-\lfloor \frac{n}{\delta}\rfloor$ coordinates repairs one inter-group erasure pattern in $\mathcal{P}_0$ (as in $\mathcal{C}_{\mathrm{I}}$).

\begin{construction}\label{construction-h}
{\it
Let $F$ be a finite field with $|F|\geq n+d-k$. Let $\lambda_1,\lambda_2,\ldots,\lambda_n,\gamma_1,\gamma_2,\ldots,\gamma_{d-k}$ be distinct elements in $F$ and $\tau\in F\setminus\{0,1\}$.
Denote $\delta=d-k+h$ and $m=\binom{n}{h}-\lfloor\frac{n}{\delta}\rfloor(\binom{\delta}{h}-1)$.
The $(n,k,\ell=\delta^m)$ cooperative MSR code $\mathcal{C}$ is defined by the parity-check matrix $H$ with the form in (\ref{def1}), where for $t\in[r]$ and $j\in[n]$, $H_{t,j}$'s are defined in Algorithm \ref{alg-0} (we index the rows and columns of $H_{t,j}$ by integers from $0$ to $\ell-1$, and denote by $H_{t,j}(a,b)$ the entry in $a$-th row and $b$-th column).
}

\begin{algorithm}[H]
\caption{Defining $H_{t,j}$, $t\in[r], j\in[n]$.}\label{alg1}
\begin{algorithmic}[1]\label{alg-0}
\REQUIRE The basic parameters $n,k,d,h$ and parameters $t, j$ with $t\in[r]$, $j\in[n]$, and parameter $\ell=\delta^m$. Distinct elements $\lambda_1,\ldots,\lambda_n,\gamma_1,\ldots,\gamma_{d-k}$ in $F$ and an element $\tau\in F\setminus\{0,1\}$.
\ENSURE The desired matrix $H_{t,j}$.
\STATE
Initialize $H_{t,j}$ to be a diagonal matrix: $H_{t,j}=\lambda_{j}^{t-1}I_{\ell}$, where $I_{\ell}$ represents an identity matrix of order $\ell$.
\STATE
Recall the definition of $\Omega_{j,e}$'s, $e\in[0,h-1]$ in Subsection \ref{notation-h}.
\FOR{$a\in[0,\ell-1], e\in[0,h-1], u\in\Omega_{j,e}$}
\IF {$a_{u}=e$}
\STATE
Set $H_{t,j}(a,a(u,w))=\gamma_{w-h+1}^{t-1}$ for all $w\in[h,\delta-1]$.
\ENDIF
\ENDFOR
\IF {$1\leq j\leq \delta\lfloor\frac{n}{\delta}\rfloor$}
\STATE
Denote $j=(u-1)\delta+v+1$ for some $u\in[\lfloor\frac{n}{\delta}\rfloor]$ and $v\in[0,\delta-1]$.
\FOR{$a\in[0,\ell-1]$}
\IF {$a_{u}=v$}
\STATE
Set $H_{t,j}(a,a(u,w))=\tau\lambda_{(u-1)\delta+w+1}^{t-1}$ for all $0\leq w<v$;
\STATE
Set $H_{t,j}(a,a(u,w))=\lambda_{(u-1)\delta+w+1}^{t-1}$ for all $v< w\leq \delta-1$.
\ENDIF
\ENDFOR
\ENDIF
\RETURN
The matrix $H_{t,j}$.
\end{algorithmic}
\end{algorithm}

\end{construction}

\begin{remark}\label{remark}
We construct each $H_{t,j}$, $t\in[r], j\in[n]$ in Algorithm \ref{alg-0} from the diagonal matrix $\lambda_j^{t-1}I_{\ell}$ by successively adding some non-diagonal non-zero entries.
That is, as illustrated in line $3$-$7$ of Algorithm \ref{alg-0}, consider each row indexed by $a\in[0,\ell-1]$.
Write $a=(a_1,a_2,\ldots,a_m)$. We explain how to add non-diagonal non-zero entries in the $a$-th row of $H_{t,j}$.
Let $\Omega_{j,0}, \Omega_{j,1}, \ldots, \Omega_{j,h-1}$ be the $h$ subsets defined in Subsection \ref{notation-h}.
For each $e\in[0,h-1]$ and for each $u\in\Omega_{j,e}$, if $a_u=e$, then we add $\gamma_{w-h+1}^{t-1}$ at the $a$-th row and $a(u,w)$-th column of $H_{t,j}$ for all $w\in[h,\delta-1]$.
These $\gamma_{w-h+1}^{t-1}$'s are added in a similar manner as in matrix (\ref{pc-block1}), which are referred to as {\it type-I non-diagonal entries}.
Besides,
in line $8$-$16$ of Algorithm \ref{alg-0}, we add some non-diagonal non-zero entries in $H_{t,j}$ specially for nodes $j\in[\delta\lfloor\frac{n}{\delta}\rfloor]$.
When $j\in[\delta\lfloor\frac{n}{\delta}\rfloor]$, write $j=(u,v)$ for unique $u\in[\lfloor\frac{n}{\delta}\rfloor]$ and $v\in[0,\delta-1]$.
Consider each row indexed by $a=(a_1,a_2,\ldots,a_m)\in[0,\ell-1]$.
If $a_{u}=v$, we add $\tau\lambda_{(u,w)}^{t-1}$ at the $a$-th row and $a(u,w)$-th column of $H_{t,j}$  for $0\leq w<v$, and add $\lambda_{(u,w)}^{t-1}$ at the $a$-th row and $a(u,w)$-th column of $H_{t,j}$ for $v< w\leq \delta-1$.
These $\tau\lambda_{(u,w)}^{t-1}$'s and $\lambda_{(u,w)}^{t-1}$'s are added following a similar manner as in matrix (\ref{pc-block2}), which are referred to as {\it type-II non-diagonal entries}.

Actually, these non-diagonal non-zero entries are used to execute node repair. 
When repairing an erasure pattern $(j_0,\ldots,j_{h-1})$, for an intra-group pattern (belonging to some $\mathcal{P}_u$, $u\in[\lfloor \frac{n}{\delta}\rfloor]$), we set $j_i=(u-1)\delta+v_i+1$ with $v_i\in[0,\delta-1]$ and use the equations labeled by $\{a:a_u=v_i\}$ to repair node $j_i$  in the download phase; for an inter-group pattern (belonging to $\mathcal{P}_0$), we set $\rho=\pi(j_0,\ldots,j_{h-1})$ and use the equations labeled by $\{a:a_\rho=i\}$ to repair node $j_i$, $i\in[0,h-1]$  in the download phase.

\end{remark}

For a better understanding of the construction, we give an illustrating example.

\begin{example}\label{ex2}
Let $n=7, k=2, d=4, h=3$. Then $\delta=5$, $m=26$ and $\ell=5^{26}$.
Note that in \cite{Zhang2020} for the same parameters, it has $\ell=5^{35}$.
Then the cooperative MSR code $\mathcal{C}$ in Construction \ref{construction-h} is a $(7,2,5^{26})$ MDS array code over $F$, where $F$ is a finite field with $|F|\geq 9$. Choose $\lambda_1,\lambda_2,\ldots,\lambda_7,\gamma_1,\gamma_2$ to be $9$ distinct elements in $F$ and choose $\tau\in F\setminus\{0,1\}$.
Every integer $a\in[0,\ell-1]$ is represented by a vector $(a_1,a_2,\ldots,a_{26})$ with $a_i\in[0,4]$, $i\in[26]$.
The sets $\mathcal{P}$, $\mathcal{P}_i$'s and map $\pi$ defined in Subsection \ref{notation-h} are explicitly given in Example \ref{ex1}, and it has $\mathcal{P}=\mathcal{P}_0\cup\mathcal{P}_1$ with $|\mathcal{P}_0|=10$ and $|\mathcal{P}_1|=25$.
The surjective map $\pi$ is displayed in Table \ref{map-pi}.

Take $j=4$ for example, and we construct $H_{t,4}$ for $t\in[r]$ by Algorithm \ref{alg-0}.
Recall the definition of $\Omega_{4,0}, \Omega_{4,1}, \Omega_{4,2}$ in Subsection \ref{notation-h} which are explicitly given in Example \ref{ex1}.
Note that $j=4\leq \delta\lfloor\frac{n}{\delta}\rfloor=5$ and we can write $4=(u-1)\delta+v+1$ with $u=1$ and $v=3$.
According to Algorithm \ref{alg-0}, we construct $H_{t,4}$ through the following steps:
\begin{itemize}
\item Begin with the diagonal matrix $\lambda_4^{t-1}I_{5^{26}}$;
\item Add type-I non-diagonal entries: \\
Since $\Omega_{4,0}=\{21, 24, 25\}$, for $a\in[0,5^{26}-1]$ with $a_{21}=0$ (resp. $a_{24}=0$, $a_{25}=0$), set the $(a,a(21,3))$-th (resp. $(a,a(24,3))$-th, $(a,a(25,3))$-th) entry to be $\gamma_{1}^{t-1}$;
set the $(a,a(21,4))$-th (resp. $(a,a(24,4))$-th, $(a,a(25,4))$-th) entry to be $\gamma_{2}^{t-1}$.\\
Since $\Omega_{4,1}=\{7, 8, 14, 15, 19, 20\}$, for $a\in[0,5^{26}-1]$ with $a_{7}=1$ (resp. $a_{8}=1$,  $a_{14}=1$, $a_{15}=1$, $a_{19}=1$, $a_{20}=1$), set the $(a,a(7,3))$-th (resp. $(a,a(8,3))$-th, $(a,a(14,3))$-th $(a,a(15,3))$-th $(a,a(19,3))$-th $(a,a(20,3))$-th) entry to be $\gamma_{1}^{t-1}$;
set the $(a,a(7,4))$-th (resp. $(a,a(8,4))$-th, $(a,a(14,4))$-th $(a,a(15,4))$-th $(a,a(19,4))$-th $(a,a(20,4))$-th) entry to be $\gamma_{2}^{t-1}$.
\item Add type-II non-diagonal entries: \\
For $a\in[0,5^{26}-1]$ with $a_1=3$, set the $(a,a(1,0))$-th entry to be $\tau\lambda_{1}^{t-1}$;
set the $(a,a(1,1))$-th entry to be $\tau\lambda_{2}^{t-1}$;
set the $(a,a(1,2))$-th entry to be $\tau\lambda_{3}^{t-1}$;
and set the $(a,a(1,4))$-th entry to be $\lambda_{5}^{t-1}$.
\end{itemize}

Note that for each $a\in[0,\ell-1]$, the $a$-th row of $H_{t,4}$ may have several non-diagonal non-zero entries, and these non-diagonal entries added at each step all lie in different columns.
\end{example}

In the following, we give the MDS property and optimal-access property of the code in Construction \ref{construction-h} in Subsection \ref{mds-h} and Subsection \ref{repair-h}, respectively.

\subsection{MDS property}\label{mds-h}
Before proving the MDS property of the code in Construction \ref{construction-h}, we firstly give some necessary definitions and some explanation about the structure of the parity-check matrix.
At first, in Definition \ref{partition}, we define two partitions of the index set $[0,\ell-1]$, which play important roles in the proof of both MDS property and optimal-access property of $\mathcal{C}$.

\begin{definition}\label{partition}
Define two partitions of $[0,\ell-1]$ as follows:
\begin{itemize}
\item[(1)] {\it The first partition of $[0,\ell-1]$:}
For every integer $a\in[0,\ell-1]$, define
$
w_{\rm{suf}}(a)=|\{u\in[\lfloor\frac{n}{\delta}\rfloor+1,m]: a_u\in[0,h-1]\}|,
$
which indicates the number of digits in $(a_{\lfloor\frac{n}{\delta}\rfloor+1}, \ldots, a_m)$ that belong to $[0,h-1]$.
For $0\leq s\leq m-\lfloor\frac{n}{\delta}\rfloor$, define $\mathcal{L}_s=\{a\in[0,\ell-1]: w_{\rm{suf}}(a)=s\}$.
Then $\mathcal{L}_0, \mathcal{L}_1,\ldots,\mathcal{L}_{m-\lfloor\frac{n}{\delta}\rfloor}$ form a partition of $[0,\ell-1]$.

\item[(2)] {\it The second partition of $[0,\ell-1]$:}
Given $I\subseteq[\delta\lfloor\frac{n}{\delta}\rfloor]$, then for each $i\in I$, $i$ can be written as $i=(u,v)$ for some $u\in[\lfloor\frac{n}{\delta}\rfloor]$ and $v\in[0, \delta-1]$. \\
For every integer $a\in[0,\ell-1]$, define $w_{I}(a)=|\{u\in [\lfloor\frac{n}{\delta}\rfloor]:~(u, a_{u})\in I\}|$, which indicates the number of coordinates in $(a_1,\ldots, a_{\lfloor\frac{n}{\delta}\rfloor})$ such that $(u, a_{u})\in I$.
For each $0\leq b\leq \lfloor\frac{n}{\delta}\rfloor$, define $\Lambda_b=\{a\in[0,\ell-1]: w_{I}(a)=b\}$.
Then $\Lambda_0,\Lambda_1,\ldots,\Lambda_{\lfloor\frac{n}{\delta}\rfloor}$ w.r.t. $I$ form a partition of $[0,\ell-1]$. \\
Moreover, for each $0\leq s\leq m-\lfloor\frac{n}{\delta}\rfloor$, it has  that $\mathcal{L}_s$ defined in (1) has a partition $\mathcal{L}_s\cap \Lambda_b$ for $b=0,1,\ldots, \lfloor\frac{n}{\delta}\rfloor$.
\end{itemize}
\end{definition}

Note that the second partition of $[0,\ell-1]$ in Definition \ref{partition} is related to a given subset $I\subseteq[\delta\lfloor\frac{n}{\delta}\rfloor]$.
Next, we give a lemma to illustrate some relations between $I$ and the partition sets $\Lambda_b$'s.

\begin{lemma}\label{partition-I}
Let $I\subseteq[\delta\lfloor\frac{n}{\delta}\rfloor]$ and recall the node partition in Subsection \ref{notation-h}.
Then $I$ contains exactly $b$ integral node groups for some $b\in[0,\lfloor\frac{n}{\delta}\rfloor]$ if and only if $\Lambda_0=\cdots=\Lambda_{b-1}=\emptyset$ and $\Lambda_{b}\neq\emptyset$.
\end{lemma}

\begin{proof}
Suppose $I$ contains exactly $b$ integral node groups, denoted by group $u_1,\ldots,u_b$, then $\cup_{z\in[b]}[(u_z-1)\delta+1, u_z\delta]\subseteq I$.
That is, $\{(u_z,v): z\in[b], v\in[0,\delta-1]\}\subseteq I$.
This implies for each $a\in[0,\ell-1]$, it has $(u_z,a_{u_z})\in I$ for all $z\in[b]$.
Thus, $\Lambda_0=\cdots=\Lambda_{b-1}=\emptyset$.
Moreover, for each $u\in[\lfloor\frac{n}{\delta}\rfloor]\setminus\{u_1,\ldots,u_b\}$, it has $\{(u,v): v\in[0,\delta-1]\}\nsubseteq I$, i.e., there exists some $v_u\in[0,\delta-1]$ s.t. $(u,v_u)\notin I$.
Choose $a\in[0,\ell-1]$ with $a_u=v_u$ for all $u\in[\lfloor\frac{n}{\delta}\rfloor]\setminus\{u_1,\ldots,u_b\}$.
Then it has $a\in\Lambda_b$, thus $\Lambda_{b}\neq\emptyset$.

Now suppose $\Lambda_0=\cdots=\Lambda_{b-1}=\emptyset$ and $\Lambda_{b}\neq\emptyset$.
Then it has $\forall a\in[0,\ell-1]$, $w_{I}(a)\geq b$, and  there exists  $a'\in[0,\ell-1]$ with $w_{I}(a')=b$.
Denote $\{u\in[\lfloor\frac{n}{\delta}\rfloor]:(u,a'_u)\in I\}=\{u_1,u_2,\cdots,u_{b}\}$. Next we claim that $\{(u_z,v): z\in[b], v\in[0,\delta-1]\}\subseteq I$, i.e., $I$ contains $b$ integral node groups. Otherwise, there exists
$(u_i,v)\notin I$ for some $i\in[b]$ and $v\neq a'_{u_i}$. Let $a''=a'(u_i,v)$, then $w_I(a'')=b-1$, i.e., $a''\in\Lambda_{b-1}\neq \emptyset$,  a contradiction.
Moreover, if $I$ contains $b'>b$ integral node groups, then it has $\Lambda_b=\emptyset$, a contradiction.
Thus $I$ contains exactly $b$ integral node groups.
\end{proof}

According to Lemma \ref{partition-I}, given $I\subseteq[\delta\lfloor\frac{n}{\delta}\rfloor]$, then $[0,\ell-1]$ has a partition: $\Lambda_b,\Lambda_{b+1},\ldots,\Lambda_{\lfloor\frac{n}{\delta}\rfloor}$ w.r.t. $I$, where $b$ is the least number in $[0,\lfloor\frac{n}{\delta}\rfloor]$ s.t. $\Lambda_b\neq\emptyset$.
Next, Let $I\subseteq[\delta\lfloor\frac{n}{\delta}\rfloor]$.
For each $u\in[\lfloor\frac{n}{\delta}\rfloor]$, define $V_u=\{v\in[0,\delta-1]:~(u,v)\in I\}$.
By Lemma \ref{partition-I}, one can directly obtain the following corollary.

\begin{corollary}\label{partition-I-2}
Given $I\subseteq[\delta\lfloor\frac{n}{\delta}\rfloor]$, for each $u\in[\lfloor\frac{n}{\delta}\rfloor]$, define $V_u=\{v\in[0,\delta-1]:~(u,v)\in I\}$. Then $I=\cup_{u\in[\lfloor\frac{n}{\delta}\rfloor]}\{(u,v):~v\in V_u\}$. Moreover,
\begin{itemize}
\item
If $\Lambda_0\neq\emptyset$, then $V_u\subsetneqq[0,\delta-1]$ for all $u\in[\lfloor\frac{n}{\delta}\rfloor]$.
\item If $\Lambda_0=\cdots=\Lambda_{b-1}=\emptyset$ and $\Lambda_{b}\neq\emptyset$, then there exist exactly $b$ integers $\{u_1,\ldots,u_b\}\subseteq[\lfloor\frac{n}{\delta}\rfloor]$ s.t.
$V_u=[0,\delta-1]$ for all $u\in\{u_1,\ldots,u_b\}$, and $V_u\subsetneqq[0,\delta-1]$ for all $u\in[\lfloor\frac{n}{\delta}\rfloor]\setminus\{u_1,\ldots,u_b\}$.
\end{itemize}
\end{corollary}

In the following two lemmas, we characterize the structure of the parity-check matrix of $\mathcal{C}$ defined in Construction \ref{construction-h}.
We give the explicit positions of non-zero elements in the $a$-th row of $(H_{t,1},\ldots,H_{t,n})$ for $t\in[r]$, according to Algorithm \ref{alg-0}. This will be frequently used later in the proof of both MDS property and optimal-access property of $\mathcal{C}$.
Before that,
for simplicity, we define a function $f(v,v')$ where $v\neq v'$ s.t. $f(v,v')=\tau$ if $v>v'$ and $f(v,v')=1$ if $v<v'$.

\begin{lemma}\label{row-s}
Denote by $(H_{t,1},H_{t,2},\ldots,H_{t,n})\in F^{\ell\times n\ell}$ the $t$-th row block of the parity-check matrix $H$ defined in Construction \ref{construction-h}, for $t\in[r]$.
Given some $s\in[0,m-\lfloor\frac{n}{\delta}\rfloor]$, let $a\in\mathcal{L}_s$ and
denote $\{u\in[\lfloor\frac{n}{\delta}\rfloor+1,m]: a_u\in[0,h-1]\}=\{u_1,u_2,\ldots,u_s\}$.
Then, the non-zero entries in the $a$-th row of $(H_{t,1},\ldots,H_{t,n})$ are given as follows:
\begin{itemize}
\item (Diagonal entries):
$\lambda_{j}^{t-1}$ in the $a$-th column of $H_{t,j}$ for $j\in[n]$ (We also write $\lambda_{(u,v)}^{t-1}$ for $j=(u,v)\in[\delta\lfloor\frac{n}{\delta}\rfloor]$).
\item (Type-I non-diagonal entries):
for each $z\in[s]$, the symbols $\gamma_{w-h+1}^{t-1}$, $w\in[h,\delta-1]$ with each  locating in the $a(u_z,w)$-th column of $H_{t,j}$ where $j\in[n]$ satisfying $u_z\in\Omega_{j,a_{u_z}}$.

\item (Type-II non-diagonal entries):
for each $u\in[\lfloor\frac{n}{\delta}\rfloor]$, the symbols $f(a_u,w)\lambda_{(u,w)}^{t-1}$, $w\in[0,\delta-1]\setminus\{a_u\}$ with each locating in the $a(u,w)$-th column of $H_{t,(u,a_u)}$.
\end{itemize}
\end{lemma}
\begin{proof}
Lemma \ref{row-s} can be easily verified according to Algorithm \ref{alg-0}.
\end{proof}

\begin{lemma}\label{row-s-b}
Let $\{i_1,i_2,\ldots,i_r\}\subseteq[n]$ and $I=\{i_1,i_2,\ldots,i_r\}\cap[\delta\lfloor\frac{n}{\delta}\rfloor]$.
Denote by $(H_{t,i_1},H_{t,i_2},\ldots,H_{t,i_r})\in F^{\ell\times r\ell}$ the sub-matrix of $(H_{t,j})_{j\in[n]}$ restricted to the $r$ nodes $\{i_1,i_2,\ldots,i_r\}$.
Given some $s\in[0,m-\lfloor\frac{n}{\delta}\rfloor]$ and $b\in[0,\lfloor\frac{n}{\delta}\rfloor]$, and let $a\in\mathcal{L}_s\cap\Lambda_b$.
Denote $\{u\in[\lfloor\frac{n}{\delta}\rfloor+1,m]: a_u\in[0,h-1]\}=\{\sigma_1,\sigma_2,\ldots,\sigma_s\}$,
and denote $\{u\in[\lfloor\frac{n}{\delta}\rfloor]: (u,a_u)\in I\}=\{u_1,u_2,\ldots, u_b\}$.
Then, the non-zero entries in the $a$-th row of $(H_{t,i_1},H_{t,i_2},\ldots,H_{t,i_r})$ are given as follows:
\begin{itemize}
\item (Diagonal entries):
$\lambda_{i_j}^{t-1}$ in the $a$-th column of $H_{t,i_j}$ for $j\in[r]$.
\item (Type-I non-diagonal entries):
for each $z\in[s]$, the symbols $\gamma_{w-h+1}^{t-1}$, $w\in[h,\delta-1]$ with each  locating in the $a(\sigma_z,w)$-th column of $H_{t,i_j}$ where $j\in[r]$ satisfying $\sigma_z\in\Omega_{i_j,a_{\sigma_z}}$.

\item (Type-II non-diagonal entries):
for each $u\in\{u_1,u_2,\ldots, u_b\}$, the symbols $f(a_u,w)\lambda_{(u,w)}^{t-1}$, $w\in[0,\delta-1]\setminus\{a_u\}$ with each locating in the $a(u,w)$-th column of $H_{t,(u,a_u)}$.
\end{itemize}
\end{lemma}

\begin{proof}
Lemma \ref{row-s-b} can be easily verified according to Algorithm \ref{alg-0}.
\end{proof}

Next, based on Lemma \ref{row-s-b}, in Lemma \ref{equation-s-b} we further characterize the simplified form of the following linear equation system $(H_{t,i_1},H_{t,i_2},\ldots,H_{t,i_r})\cdot\bm{x}^{\top}=\bm{0}$, where $\bm{x}=(\bm{x}_{i_1},\bm{x}_{i_2},\ldots,\bm{x}_{i_r})\in F^{r\ell}$ with $\bm{x}_{i_j}=(x_{i_j,0},\ldots,x_{i_j,\ell-1})\in F^{\ell}$, $j\in[r]$, which will be used later in the proof of MDS property of $\mathcal{C}$.
For convenience, we define a function $\chi(a,u,j)$ where $a\in[0,\ell-1]$, $u\in[m]$ and $j\in[n]$ s.t. $\chi(a,u,j)=1$ if there exists some $e\in[0,h-1]$ s.t. $a_{u}=e$ and $u\in\Omega_{j,e}$, otherwise, $\chi(a,u,j)=0$.

\begin{lemma}\label{equation-s-b}
Let $\{i_1,i_2,\ldots,i_r\}\subseteq[n]$ and $I=\{i_1,i_2,\ldots,i_r\}\cap[\delta\lfloor\frac{n}{\delta}\rfloor]$.
For each $u\in[\lfloor\frac{n}{\delta}\rfloor]$, define $V_u=\{v\in[0,\delta-1]: (u,v)\in I\}$.
Let $(H_{t,i_1},H_{t,i_2},\ldots,H_{t,i_r})\in F^{\ell\times r\ell}$ represent the sub-matrix of $(H_{t,j})_{j\in[n]}$ restricted to the $r$ nodes $\{i_1,i_2,\ldots,i_r\}$ for $t\in[r]$.
Let $s\in[0,m-\lfloor\frac{n}{\delta}\rfloor]$ and $b\in[0,\lfloor\frac{n}{\delta}\rfloor]$. For every $a\in\mathcal{L}_s\cap\Lambda_b$, denote $\{u\in[\lfloor\frac{n}{\delta}\rfloor+1,m]: a_u\in[0,h-1]\}=\{\sigma_1,\sigma_2,\ldots,\sigma_s\}$, and denote $\{u\in[\lfloor\frac{n}{\delta}\rfloor]: (u,a_u)\in I\}=\{u_1,u_2,\ldots, u_b\}$.
Then, the $a$-th row of the linear equations $(H_{t,i_1},H_{t,i_2},\ldots,H_{t,i_r})\cdot\bm{x}^{\top}=\bm{0}$, $t\in[r]$, has the following form:
\begin{equation}\label{equation-form}
\begin{aligned}
&\sum_{u\in\{u_1,\ldots,u_b\}}\bigg\{
\underbrace{\lambda_{(u,a_{u})}^{t-1}x_{(u,a_{u}), a}
+\sum_{ v\in V_u\setminus\{a_{u}\}}\lambda_{(u,v)}^{t-1}\big(f(a_{u},v)x_{(u,a_u), a(u,v)}+x_{(u,v), a}\big)}_{(1)}
+
\underbrace{\sum_{v\in[0,\delta-1]\setminus V_u} \lambda_{(u,v)}^{t-1}f(a_{u},v)x_{(u,a_u), a(u,v)}}_{(2)}
  \bigg\}   \\
+&
\underbrace{\sum_{\substack{ j\in\{i_1,\ldots,i_r\}  \\ j\notin\{(u_z,v): z\in[b], v\in V_{u_z}\}}}\lambda_j^{t-1}x_{j,a} }_{(3)}
 +
\underbrace{
\sum_{w\in[h,\delta-1]}\gamma_{w-h+1}^{t-1}(\sum_{\substack{j\in\{i_1,\ldots,i_r\},\\\sigma\in\{\sigma_1,\ldots,\sigma_s\}}}\chi(a,\sigma,j)x_{j,a(\sigma,w)})}_{(4)}
=0,   ~~~~~~~~~~~~t\in[r],
\end{aligned}
\end{equation}
which can be rewritten as the following matrix form:
\begin{equation}\label{equation-form-matrix}
\begin{aligned}
\left(\begin{array}{ccc|ccc|ccc}
1 & \cdots & 1 & \cdots & 1 &\cdots& 1 & \cdots & 1
\\
\lambda_{i_1} & \cdots & \lambda_{i_r} & \cdots &\lambda_{(u,v)}& \cdots &\gamma_1 & \cdots & \gamma_{d-k}
\\
\vdots & \vdots  & \vdots  &\vdots & \vdots  & \vdots      & \vdots & \vdots & \vdots
       \\
\lambda_{i_1}^{r-1} & \cdots & \lambda_{i_r}^{r-1}
& \cdots &\lambda_{(u,v)}^{r-1}& \cdots
&\gamma_1^{r-1} & \cdots & \gamma_{d-k}^{r-1}
\end{array}\right)
\end{aligned}
\begin{pmatrix}
x_{i_1,a}^*\\
\vdots\\
x_{i_r,a}^*\\  \hline
\vdots\\
f(a_{u},v)x_{(u,a_u), a(u,v)}\\
\vdots\\ \hline
\sum_{\substack{j\in\{i_1,\ldots,i_r\},\\\sigma\in\{\sigma_1,\ldots,\sigma_s\}}}\chi(a,\sigma,j)x_{j,a(\sigma,h)})\\
\vdots\\
\sum_{\substack{j\in\{i_1,\ldots,i_r\},\\\sigma\in\{\sigma_1,\ldots,\sigma_s\}}}\chi(a,\sigma,j)x_{j,a(\sigma,\delta-1)})
\end{pmatrix}=\bm{0},
\end{equation}
where for $i_j\in\{i_1,\ldots,i_r\}\setminus\cup_{z\in[b]}\{(u_z,v): v\in V_{u_z}\setminus\{a_{u_z}\}\}$, $x_{i_j,a}^*=x_{i_j,a}$;
and for $i_j=(u,v)\in\cup_{z\in[b]}\{(u_z,v): v\in V_{u_z}\setminus\{a_{u_z}\}\}$, $x_{i_j,a}^*=f(a_{u},v)x_{(u,a_u), a(u,v)}+x_{(u,v), a}$.
Moreover, it is worth noting that in formula \eqref{equation-form}, the sum of part $(1)$ on $u\in\{u_1,\ldots,u_b\}$ and part $(3)$  contribute to the first block of \eqref{equation-form-matrix}, sum of part $(2)$ on $u\in\{u_1,\ldots,u_b\}$ contributes to the second  block of \eqref{equation-form-matrix}, and part $(4)$ contributes to the third block of \eqref{equation-form-matrix}.
\end{lemma}

\begin{proof}
According to Lemma \ref{row-s-b} and by combining like terms and removing zeros of the $a$-th row of $(H_{t,i_1},\ldots,H_{t,i_r})\bm{x}^{\top}=\bm{0}$, $t\in[r]$, one can easily obtain the simplified equations illustrated  in \eqref{equation-form}.
Here part $(1)$ and $(2)$ of \eqref{equation-form} comes from the type-II non-diagonal non-zero entries and some diagonal non-zero entries, $(3)$ comes from the diagonal non-zero entries, and $(4)$ comes from the type-I non-diagonal non-zero entries.
Moreover, note that $I=\cup_{u\in[\lfloor\frac{n}{\delta}\rfloor]}\{(u,v):~v\in V_u\}$ and $\{(u_z,v): z\in[b], v\in V_{u_z}\}\subseteq I\subseteq\{i_1,\ldots,i_r\}$, then the formula \eqref{equation-form} can be further reformulated to the matrix form illustrated  in \eqref{equation-form-matrix}.
\end{proof}

Based on the above lemmas, now, we are able to give the MDS property of $\mathcal{C}$.

\begin{theorem}\label{thm-mds-h}
The code $\mathcal{C}$ in Construction \ref{construction-h} has MDS property.
\end{theorem}

\begin{proof}
It suffices to prove any $r$ column blocks of $H$, denoted by $H(i_1,\ldots,i_r)$, where $1\leq i_1<\cdots<i_r\leq n$, forms an invertible $r\ell\times r\ell$ matrix. Equivalently, we prove that for any $\bm{x}=(\bm{x}_{i_1},\bm{x}_{i_2},\ldots,\bm{x}_{i_r})\in F^{r\ell}$ with $\bm{x}_{i_j}=(x_{i_j,0},\ldots,x_{i_j,\ell-1})\in F^{\ell}$, $j\in[r]$,  it has that $H(i_1,\ldots,i_r)\cdot\bm{x}^{\top}=\bm{0}$ always implies $\bm{x}=\bm{0}$, i.e., $\{x_{i_j,a}\}_{j\in[r], a\in[0,\ell-1]}$ are all zeros.

Let $I=\{i_1,i_2,\ldots,i_r\}\cap[\delta\lfloor\frac{n}{\delta}\rfloor]$ and $I'=\{i_1,i_2,\ldots,i_r\}\setminus I$, then $I'\subseteq[\delta\lfloor\frac{n}{\delta}\rfloor+1,n]$.
Since $|[\delta\lfloor\frac{n}{\delta}\rfloor+1,n]|\leq\delta-1<r$, then $|I'|<r$ and $I\neq\emptyset$.
Recall the two partitions of $[0,\ell-1]$ in Definition \ref{partition}.
It has that $\mathcal{L}_s$, $s\in[0,m-\lfloor\frac{n}{\delta}\rfloor]$ form a partition of $[0,\ell-1]$. Moreover, for each $s\in[0,m-\lfloor\frac{n}{\delta}\rfloor]$, the sets $\mathcal{L}_s\cap \Lambda_b$, $b=b_0,b_0+1,\ldots, \lfloor\frac{n}{\delta}\rfloor$ w.r.t. $I$ form a partition of $\mathcal{L}_s$, where $b_0$ is the least number in $[0,\lfloor\frac{n}{\delta}\rfloor]$ s.t. $\Lambda_{b_0}\neq\emptyset$.
For simplicity, according to Lemma \ref{partition-I}, denote the $b_0$ integral node groups contained in $I$ as $\{u_1,u_2,\ldots, u_{b_0}\}$, where when $b_0=0$, the set $\{u_1,u_2,\ldots, u_{b_0}\}=\emptyset$.
Next,
we prove by induction on $s$ that for each $s\in[0,m-\lfloor\frac{n}{\delta}\rfloor]$, it has $\{x_{i_j,a}\}_{j\in[r],a\in\mathcal{L}_s}$ are all zeros. More precisely, for each $s$, we prove by induction on $b$ that for each $b\in[b_0,\lfloor\frac{n}{\delta}\rfloor]$, it has $\{x_{i_j,a}\}_{j\in[r],a\in\mathcal{L}_s\cap\Lambda_b}$ are all zeros.

In the following, we begin with the case $s=0$ and prove by induction on $b$ that for each  $b\in[b_0,\lfloor\frac{n}{\delta}\rfloor]$, it has $\{x_{i_j,a}\}_{j\in[r],a\in\mathcal{L}_0\cap\Lambda_{b}}$ are all zeros.
At first, let $b=b_0$, and we prove $\{x_{i_j,a}\}_{j\in[r],a\in\mathcal{L}_0\cap\Lambda_{b_0}}$ are all zeros.

For every integer $a\in\mathcal{L}_0\cap\Lambda_{b_0}$, it has $\{u\in[\lfloor\frac{n}{\delta}\rfloor+1,m]: a_u\in[0,h-1]\}=\emptyset$ by the definition of $\mathcal{L}_{0}$.
Also, $\{u\in[\lfloor\frac{n}{\delta}\rfloor]: (u, a_u)\in I\}=\{u_1,u_2,\ldots, u_{b_0}\}$ if $b_0>0$, and $\{u\in[\lfloor\frac{n}{\delta}\rfloor]: (u, a_u)\in I\}=\emptyset$ if $b_0=0$, by the definition of $\Lambda_{b_0}$.
Consider the $a$-th row of the equations $(H_{t,i_1},\ldots,H_{t,i_r})\bm{x}^{\top}=\bm{0}$, $t\in[r]$. According to Lemma \ref{equation-s-b}, one can obtain $r$ simplified equations with form in \eqref{equation-form}, and the corresponding matrix form in \eqref{equation-form-matrix}.
\begin{itemize}
\item If $b_0=0$, then the $r$ obtained simplified equations in \eqref{equation-form} reduce to the form $\sum_{ j\in\{i_1,\ldots,i_r\}}\lambda_j^{t-1}x_{j,a}=0$, $t\in[r]$, which implies that $x_{i_j,a}$, $j\in[r]$ are all zeros.
Thus, $\{x_{i_j,a}\}_{j\in[r], a\in\mathcal{L}_0\cap\Lambda_{b_0}}$ are all zeros.

\item If $b_0>0$, then the $r$ obtained simplified equations in \eqref{equation-form} do not have part $(4)$ since $a\in \mathcal{L}_0$, and do not have part $(2)$ since for $u\in\{u_1,\ldots,u_{b_0}\}$, it has $V_u=[0,\delta-1]$ by Corollary \ref{partition-I-2}.
Then using the $r$ simplified equations, one can further compute the following $r$ unknowns in \eqref{mds-case1-eq2} to be zeros, where we note that $f(a_{u},v)=\tau$ if $a_{u}>v$ and $f(a_{u},v)=1$ if $a_{u}<v$.
\begin{equation}\label{mds-case1-eq2}
\begin{aligned}
&\{x_{(u,a_{u}), a}: u\in\{u_1,\ldots,u_{b_0}\}\} \\
\cup&\{\tau x_{(u,a_{u}), a(u,v)}+x_{(u,v), a}:~v\in[0,a_{u}-1], u\in\{u_1,\ldots,u_{b_0}\}\}  \\
\cup&\{x_{(u,a_{u}), a(u,v)}+x_{(u,v), a}:~v\in[a_{u}+1, \delta-1], u\in\{u_1,\ldots,u_{b_0}\}\}   \\
\cup&\{x_{j,a}: j\in\{i_1,\ldots,i_r\}\setminus\{(u,v): u\in\{u_1,\ldots,u_{b_0}\}, v\in[0,\delta-1]\}\}.
\end{aligned}
\end{equation}
According to \eqref{mds-case1-eq2}, we can obtain some independent zero symbols and some zero symbol sums, similar to \eqref{eq-1-0-0}.
Next, we continue to solve out the independent symbols from the symbol sums in \eqref{mds-case1-eq2} by considering two different $a$'s, just as in \eqref{eq-1-1}.
To this end, fix some $u\in\{u_1,\ldots,u_{b_0}\}$ and for every $v\in[0,\delta-1]\setminus\{a_{u}\}$, consider the $a$-th and $a'=a(u, v)$-th row of equations $(H_{t,i_1},\ldots,H_{t,i_r})\bm{x}^{\top}=\bm{0}$, $t\in[r]$. Then we can obtain the zero symbol sums as in \eqref{mds-case1-eq2} corresponding to $a$ and $a'=a(u, v)$ respectively, similar to \eqref{eq-1-1}, from which we are able to solve out that $x_{(u,a_{u}), a(u,v)}=x_{(u,v), a}=0$.
When $u$ runs over $\{u_1,\ldots,u_{b_0}\}$, one can finally solve out that $x_{i_j,a}$, $j\in[r]$ are all zeros.
Thus $\{x_{i_j,a}\}_{j\in[r], a\in\mathcal{L}_0\cap\Lambda_{b_0}}$ are all zeros.
\end{itemize}

Suppose for all $b_0\leq b'<b$, we have proved $\{x_{i_j,a}\}_{j\in[r], a\in\mathcal{L}_0\cap\Lambda_{b'}}$ are zeros. Next we prove the case $b'=b$ that $\{x_{i_j,a}\}_{j\in[r], a\in\mathcal{L}_0\cap\Lambda_{b}}$ are zeros.
For every $a\in\mathcal{L}_0\cap\Lambda_{b}$, it has $\{u\in[\lfloor\frac{n}{\delta}\rfloor+1,m]: a_u\in[0,h-1]\}=\emptyset$ by the definition of $\mathcal{L}_{0}$. Besides, denote $\{u\in[\lfloor\frac{n}{\delta}\rfloor]: (u, a_u)\in I\}=\{u_1,u_2,\ldots, u_{b}\}$.
According to Corollary \ref{partition-I-2}, if $b_0=0$, then $V_u\subsetneqq[0,\delta-1]$ for all $u\in\{u_1,\ldots, u_{b}\}$.
If $b_0>0$, then $V_u=[0,\delta-1]$ for all $u\in\{u_1,\ldots,u_{b_0}\}$, and $V_u\subsetneqq[0,\delta-1]$ for all $u\in\{u_{b_0+1},\ldots,u_b\}$.
Now consider the $a$-th row of the equations $(H_{t,i_1},\ldots,H_{t,i_r})\bm{x}^{\top}=\bm{0}$, $t\in[r]$. By Lemma \ref{equation-s-b}, one can obtain $r$ simplified equations with form in \eqref{equation-form}, and the corresponding matrix form in \eqref{equation-form-matrix}, where part $(4)$ of \eqref{equation-form} does not exist since $a\in\mathcal{L}_0$.
Next, we claim that part $(2)$ of the $r$ obtained simplified equations in \eqref{equation-form} are actually zeros.
This is because that for each $u\in\{u_1,\ldots,u_b\}$ and $v\in[0,\delta-1]\setminus V_u$, it has $a(u,v)\in\mathcal{L}_0\cap\Lambda_{b-1}$. By the hypothesis, the data $x_{(u,a_u),a(u,v)}$'s in part $(2)$ of \eqref{equation-form} are all zeros.
Thus part $(2)$ of \eqref{equation-form} can be removed, and one can further solve out the following $r$ unknowns to be zeros:
\begin{equation}\label{mds-case1-eq3}
\begin{aligned}
&\{x_{(u,a_{u}), a}: u\in\{u_1,\ldots,u_{b}\}\} \\
\cup&\{\tau x_{(u,a_{u}), a(u,v)}+x_{(u,v), a}:~u\in\{u_1,\ldots,u_{b}\}, v\in V_u\cap[0,a_{u}-1]\}  \\
\cup&\{x_{(u,a_{u}), a(u,v)}+x_{(u,v), a}:~u\in\{u_1,\ldots,u_{b}\}, v\in V_u\cap[a_{u}+1, \delta-1]\}   \\
\cup&\{x_{j,a}: j\in\{i_1,\ldots,i_r\}\setminus\{(u,v): u\in\{u_1,\ldots,u_{b}\}, v\in V_u\}\}.
\end{aligned}
\end{equation}
As before, fix some $u\in\{u_1,\ldots, u_{b}\}$ and for every $v\in V_u\setminus\{a_{u}\}$, consider the $a$-th and $a'=a(u, v)$-th row of equations $(H_{t,i_1},\ldots,H_{t,i_r})\bm{x}^{\top}=\bm{0}$, $t\in[r]$. One can further compute that $x_{(u,a_{u}), a(u,v)}=x_{(u,v), a}=0$.
When $u$ runs over $\{u_1,\ldots,u_{b}\}$, one can finally solve out that $x_{i_j,a}$, $j\in[r]$ are all zeros.
Thus $\{x_{i_j,a}\}_{j\in[r], a\in\mathcal{L}_0\cap\Lambda_{b}}$ are all zeros.

Therefore, we have proved the base case that $\{x_{i_j,a}\}_{j\in[r], a\in\mathcal{L}_0}$ are all zeros.

Now suppose for all $0\leq s'<s$, we have proved that $\{x_{i_j,a}\}_{j\in[r],a\in\mathcal{L}_{s'}}$ are all zeros. Next we prove $\{x_{i_j,a}\}_{j\in[r],a\in\mathcal{L}_{s}}$ are all zeros.
Note that
$\{x_{i_j,a}\}_{j\in[r],a\in\mathcal{L}_s}=\cup_{b\in[b_0,\lfloor\frac{n}{\delta}\rfloor]}\{x_{i_j,a}\}_{j\in[r],a\in\mathcal{L}_s\cap\Lambda_b}$,
we prove it by induction on $b$ that for each $b\in[b_0,\lfloor\frac{n}{\delta}\rfloor]$, the symbols $\{x_{i_j,a}\}_{j\in[r],a\in\mathcal{L}_s\cap\Lambda_b}$ are zeros.

At first, consider the case $b=b_0$, and we prove $\{x_{i_j,a}\}_{j\in[r],a\in\mathcal{L}_s\cap\Lambda_{b_0}}$ are all zeros.
For each $a\in\mathcal{L}_s\cap\Lambda_{b_0}$, denote $\{u\in[\lfloor\frac{n}{\delta}\rfloor+1,m]: a_u\in[0,h-1]\}=\{\sigma_1,\sigma_2,\ldots,\sigma_s\}$ for simplicity.
Denote $\{u\in[\lfloor\frac{n}{\delta}\rfloor]: (u, a_u)\in I\}=\{u_1,u_2,\ldots, u_{b_0}\}$ if $b_0>0$, and $\{u\in[\lfloor\frac{n}{\delta}\rfloor]: (u, a_u)\in I\}=\emptyset$ if $b_0=0$.
Consider the $a$-th row of the equations $(H_{t,i_1},\ldots,H_{t,i_r})\bm{x}^{\top}=\bm{0}$, $t\in[r]$. According to Lemma \ref{equation-s-b}, one can obtain $r$ simplified equations with form in \eqref{equation-form}, and the corresponding matrix form in \eqref{equation-form-matrix}.
Next, we claim that part $(4)$ in the $r$ obtained simplified equations in \eqref{equation-form} are all zeros. Since for $\sigma\in\{\sigma_1,\ldots,\sigma_s\}$ and $w\in[h,\delta-1]$, it has $a(\sigma,w)\in\mathcal{L}_{s-1}$. Thus by the hypothesis $x_{j,a(\sigma,w)}$ are zeros for all $ j\in\{i_1,\ldots,i_r\}$, $\sigma\in\{\sigma_1,\ldots,\sigma_s\}$ and $w\in[h,\delta-1]$.
Therefore, the $r$ obtained simplified equations in \eqref{equation-form} can be reformulated by removing part $(4)$.
Then, in a similar way as in the case of $s=0$ and $b=b_0$, one can finally compute that $\{x_{i_j,a}\}_{j\in[r],a\in\mathcal{L}_s\cap\Lambda_{b_0}}$ are all zeros.

Suppose for all $b_0\leq b'<b$, we have proved $\{x_{i_j,a}\}_{j\in[r],a\in\mathcal{L}_s\cap\Lambda_{b'}}$ are zeros.
Next we prove $\{x_{i_j,a}\}_{j\in[r],a\in\mathcal{L}_s\cap\Lambda_{b}}$ are all zeros.
For each $a\in\mathcal{L}_s\cap\Lambda_{b}$, denote $\{u\in[\lfloor\frac{n}{\delta}\rfloor+1,m]: a_u\in[0,h-1]\}=\{\sigma_1,\sigma_2,\ldots,\sigma_s\}$.
Besides, denote $\{u\in[\lfloor\frac{n}{\delta}\rfloor]: (u, a_u)\in I\}=\{u_1,u_2,\ldots, u_{b}\}$.
By Corollary \ref{partition-I-2}, if $b_0=0$, then $V_u\subsetneqq[0,\delta-1]$ for all $u\in\{u_1,\ldots, u_{b}\}$.
If $b_0>0$, then $V_u=[0,\delta-1]$ for all $u\in\{u_1,\ldots,u_{b_0}\}$, and $V_u\subsetneqq[0,\delta-1]$ for all $u\in\{u_{b_0+1},\ldots,u_b\}$.
Consider the $a$-th row of the equations $(H_{t,i_1},\ldots,H_{t,i_r})\bm{x}^{\top}=\bm{0}$, $t\in[r]$. According to Lemma \ref{equation-s-b}, one has $r$ simplified equations with form in \eqref{equation-form} for $t\in[r]$, where part $(4)$ are actually zeros by the hypothesis since $a(\sigma,w)\in\mathcal{L}_{s-1}$ for all $\sigma\in\{\sigma_1,\ldots,\sigma_s\}$ and $w\in[h,\delta-1]$.
Thus, in a similar way as in the case of $s=0$ and $b'=b$, one can finally compute that $\{x_{i_j,a}\}_{j\in[r],a\in\mathcal{L}_s\cap\Lambda_{b}}$ are all zeros.

Therefore, $\{x_{i_j,a}: j\in[r], a\in[0,\ell-1]\}$ are all zeros.
Thus the code $\mathcal{C}$ satisfies MDS property.
\end{proof}

\subsection{Optimal-access property}\label{repair-h}

In this subsection,  we show the optimal-access property of $\mathcal{C}$ in Construction \ref{construction-h}.
For each $u\in[m]$ and $v\in[0,\delta-1]$, define $A(u,v)=\{a\in[0,\ell-1]: a_u=v\}$. Recall the definition of $\mathcal{P}_u$'s and map $\pi$ in Subsection \ref{notation-h}. For an intra-group  erasure pattern $\mathcal{H}=(i_0,i_1,\ldots, i_{h-1})\in\mathcal{P}_u$ with some $u\in[\lfloor\frac{n}{\delta}\rfloor]$, then $\pi(\mathcal{H})=u$. Write $i_j=(u-1)\delta+v_j+1$ with $v_j\in[0,\delta-1]$ for $j\in[0,h-1]$. Then we use parity-check equations with rows labeled by $a\in A(u,v_j)$ for repair of node $i_j$, for $j\in[0,h-1]$.
 For an inter-group erasure pattern $\mathcal{H}=(i_0,i_1,\ldots, i_{h-1})\in\mathcal{P}_0$, denote $\rho=\pi(\mathcal{H})$, we use parity-check equations with rows labeled by $a\in A(\rho,j)$ for repair of node $i_j$ for $j\in[0,h-1]$.
In the following, we illustrate the precise repair process of the two kinds of erasure patterns in Theorem \ref{thm1-repair-h} and Theorem \ref{thm2-repair-h}, respectively.

\begin{theorem}\label{thm1-repair-h}
Suppose the erased nodes $(i_0,i_1,\ldots, i_{h-1})\in\mathcal{P}_{u^*}$ for some $u^*\in[\lfloor\frac{n}{\delta}\rfloor]$. Write $i_j=(u^*-1)\delta+v_j+1$ with $v_j\in[0,\delta-1]$ for $j\in[0,h-1]$.
Denote $\mathcal{H}=\{i_0,i_1,\ldots, i_{h-1}\}$ and for $i\in\mathcal{H}$, let $\mathcal{R}_i\subseteq[n]\setminus \mathcal{H}$ with $|\mathcal{R}_i|=d$ satisfying $([(u^*-1)\delta+1,u^*\delta]\setminus \mathcal{H})\subseteq \mathcal{R}_i$ be the set of $d$ helper nodes connected by node $i$.
Then the $h$ nodes can be repaired through the following two phases.

 \begin{itemize}
\item (Download phase)
For $j\in[0,h-1]$, node $i_j$ downloads $\{c_{p,a}: a\in A(u^*,v_j)\}$ from each helper node $p\in \mathcal{R}_{i_j}$.
\item (Collaboration phase)
For each $j\in[0,h-1]$ and $j'\in[0,h-1]\setminus\{j\}$, node $i_j$ recursively computes and transmits data $\{c_{i_j,a(u^*,v_{j'})}+c_{i_{j'},a}: a\in A(u^*,v_j)\}$ to node $i_{j'}$ if $j'>j$, and transmits data $\{\tau c_{i_j,a(u^*,v_{j'})}+c_{i_{j'},a}: a\in A(u^*,v_j)\}$ to node $i_{j'}$ if $j'<j$.
\end{itemize}
\end{theorem}

\begin{theorem}\label{thm2-repair-h}
Suppose the erased nodes $(i_0,i_1,\ldots, i_{h-1})\in\mathcal{P}_0$, denote $\rho=\pi(i_0,i_1,\ldots, i_{h-1})$.
Let $\mathcal{H}=\{i_0,i_1,\ldots, i_{h-1}\}$ and for $i\in\mathcal{H}$, let $\mathcal{R}_i\subseteq[n]\setminus \mathcal{H}$ with $|\mathcal{R}_i|=d$ be the set of $d$ helper nodes connected by node $i$.
Then the $h$ nodes can be repaired through the following two phases.
 \begin{itemize}
\item (Download phase)
For $j\in[0,h-1]$, node $i_j$ downloads $\{c_{p,a}: a\in A(\rho,j)\}$ from each helper node $p\in \mathcal{R}_{i_j}$,
and can recover the data
\[\{c_{i_j,a}: a\in\cup_{v\in[h,\delta-1]\cup\{j\}}A(\rho,v)
\}
\cup\{c_{i_{j'},a}: ~j'\in[0,h-1]\setminus\{j\}, ~a\in A(\rho, j)\}.
\]

\item (Collaboration phase)
For each $j\in[0,h-1]$ and  $j'\in[0,h-1]\setminus\{j\}$, node $i_j$ transmits $\{c_{i_{j'},a}: a\in A(\rho,j)\}$ to node $i_{j'}$.
\end{itemize}
\end{theorem}

Before giving the proofs of the two theorems, note that for $j\in[0,h-1]$, node $i_j$ should be able to recover the erased data $\bm{c}_{i_j}=(c_{i_j,0}, \ldots, c_{i_j,\ell-1})$.
Recall Definition \ref{partition} that $\mathcal{L}_0, \mathcal{L}_1,\ldots,\mathcal{L}_{m-\lfloor\frac{n}{\delta}\rfloor}$ form a partition of $[0,\ell-1]$.
We will prove that for each $s=0,1,\ldots,m-\lfloor\frac{n}{\delta}\rfloor$, node $i_j$ can recover the data $\{c_{i_j,a}: a\in\mathcal{L}_s\}$ by induction on $s$.

In the following, we firstly give a lemma to characterize the simplified form of the following parity-check equations $(H_{t,1},H_{t,2},\ldots,H_{t,n})\cdot\bm{c}^{\top}=\bm{0}$, $t\in[r]$, where the codeword $\bm{c}=(\bm{c}_{1},\bm{c}_{2},\ldots,\bm{c}_{n})\in F^{n\ell}$ with $\bm{c}_{i}=(c_{i,0},\ldots,c_{i,\ell-1})\in F^{\ell}$, $i\in[n]$, which will be used in the proof of the optima-access property of $\mathcal{C}$.

\begin{lemma}\label{equation-s}
Let $(H_{t,1},H_{t,2},\ldots,H_{t,n})\in F^{\ell\times n\ell}$ represents the $t$-th row block of $H$ defined in Construction \ref{construction-h}, for $t\in[r]$.
Let $s\in[0,m-\lfloor\frac{n}{\delta}\rfloor]$. For every $a\in\mathcal{L}_s$, denote $\{u\in[\lfloor\frac{n}{\delta}\rfloor+1,m]: a_u\in[0,h-1]\}=\{\sigma_1,\sigma_2,\ldots,\sigma_s\}$.
Then, the $a$-th row of the parity-check equations $(H_{t,1},H_{t,2},\ldots,H_{t,n})\cdot\bm{c}^{\top}=\bm{0}$, $t\in[r]$, has the following form:
\begin{equation}\label{equation-form-s}
\begin{aligned}
&\underbrace{\sum_{j\in\{(u,a_u): u\in[\lfloor\frac{n}{\delta}\rfloor]\}\cup[\delta\lfloor\frac{n}{\delta}\rfloor+1,n]}\lambda_j^{t-1}c_{j,a}}_{(1)}+
\underbrace{\sum_{u\in[\lfloor\frac{n}{\delta}\rfloor]}\sum_{v\in[0,\delta-1]\setminus\{a_{u}\}}\lambda_{(u,v)}^{t-1} \big(c_{(u,a_u),a(u,v)}f(a_{u},v)+c_{(u,v),a}\big)}_{(2)}         \\
 +&
\underbrace{\sum_{v\in[h,\delta-1]}\gamma_{v-h+1}^{t-1}(\sum_{j\in[n]}\sum_{z\in[s]}\chi(a,\sigma_{z},j)c_{j,a(\sigma_{z},v)})}_{(3)}
 =0, ~~~~~~~~~~~~~~~~~~t\in[r],
\end{aligned}
\end{equation}
which can be rewritten as the following matrix form:
\begin{equation}\label{equation-form-matrix-s}
\begin{aligned}
\left(\begin{array}{cccc|ccccccc}
1 &1& \cdots & 1 & 1  &\cdots & 1
\\
\lambda_{1} &\lambda_{2} & \cdots & \lambda_{n}  &\gamma_1 & \cdots & \gamma_{d-k}
\\
\vdots & \vdots  & \vdots  &\vdots & \vdots   & \vdots & \vdots
       \\
\lambda_{1}^{r-1}&\lambda_{2}^{r-1} & \cdots & \lambda_{n}^{r-1}
&\gamma_1^{r-1} & \cdots & \gamma_{d-k}^{r-1}
\end{array}\right)
\end{aligned}
\begin{pmatrix}
c_{1,a}^*\\
\vdots\\
c_{n,a}^*\\  \hline
\sum_{j\in[n]}\sum_{z\in[s]}\chi(a,\sigma_{z},j)c_{j,a(\sigma_{z},h)})\\
\vdots\\
\sum_{j\in[n]}\sum_{z\in[s]}\chi(a,\sigma_{z},j)c_{j,a(\sigma_{z},\delta-1)})
\end{pmatrix}=\bm{0},
\end{equation}
where for $j\in\{(u,a_u): u\in[\lfloor\frac{n}{\delta}\rfloor]\}\cup[\delta\lfloor\frac{n}{\delta}\rfloor+1,n]$, $c_{j,a}^*=c_{j,a}$;
and for $j=(u,v)$ with $u\in[\lfloor\frac{n}{\delta}\rfloor]$ and  $v\in[0,\delta-1]\setminus\{a_{u}\}$, $c_{j,a}^*=c_{j,a}+f(a_{u},v)c_{(u,a_u), a(u,v)}$.
Moreover, it is worth noting that in formula \eqref{equation-form-s}, part $(1)$ and part $(2)$ contribute to the first block of \eqref{equation-form-matrix-s}, and part $(3)$ contributes to the second block of \eqref{equation-form-matrix-s}.
\end{lemma}

\begin{proof}
Lemma \ref{equation-s} can be easily proved by combining like terms and removing zeros of the $a$-th row of $(H_{t,1},\ldots,H_{t,n})\bm{c}^{\top}=\bm{0}$, $t\in[r]$, according to Lemma \ref{row-s}.\end{proof}

Now we come to the proof of Theorem \ref{thm1-repair-h}. The repair process includes $m-\lfloor\frac{n}{\delta}\rfloor+1$ stages, and in each stage $s=0,1,\ldots,m-\lfloor\frac{n}{\delta}\rfloor$, we prove node $i_j$ can recover the data $\{c_{i_j,a}: a\in\mathcal{L}_s\}$ using the downloaded data and collaborated data in the $s$-th stage. We illustrate this in Lemma \ref{stage-lem}. For the sake of fluency, we put the proof of Lemma \ref{stage-lem} in Appendix \ref{appendix-1}.

\begin{lemma}\label{stage-lem}
In every stage $s\in[0,m-\lfloor\frac{n}{\delta}\rfloor]$, for $j\in[0,h-1]$, node $i_j$ can recover the data $\{c_{i_j,a}: a\in\mathcal{L}_s\}$, by using the downloaded data and the following collaborated data:
\[
\{f(v_{j'},v_j)c_{i_{j'},a(u^*,v_{j})}+c_{i_{j},a}: ~a\in A(u^*,v_{j'})\cap\mathcal{L}_{s}\}
\]
from node $i_{j'}$ for $j'\in[0,h-1]\setminus\{j\}$ at stage $s$.
\end{lemma}

According to Lemma \ref{stage-lem}, till the last stage $s=m-\lfloor\frac{n}{\delta}\rfloor$, all the $h$ failed nodes can be repaired, and Theorem \ref{thm1-repair-h} can be proved.
Moreover, since $[0,\ell-1]=\cup_{s\in[0,m-\lfloor\frac{n}{\delta}\rfloor]}\mathcal{L}_s$, then  the total amount of data communicated among the $h$ failed nodes is $\frac{h(h-1)\ell}{d-k+h}$ symbols. The amount of data downloaded and accessed in the download phase are both $\frac{dh\ell}{d-k+h}$ symbols, achieving the cut-set bound in \eqref{bw-bound}, thus it satisfies the optimal-access property.

As for Theorem \ref{thm2-repair-h}, since $[0,\ell-1]=\cup_{v\in[0,\delta-1]}A(\rho,v)$, then if the download phase is proved, the collaboration phase and node recovery will be straightforward.
Next we prove the download phase of Theorem \ref{thm2-repair-h}, which is illustrated in Lemma \ref{thm2-lem}.

\begin{lemma}\label{thm2-lem}
For every $s\in[m-\lfloor\frac{n}{\delta}\rfloor]$, node $i_j$, $j\in[0,h-1]$ using the downloaded data can recover the following data:
\[
\begin{aligned}
&\{c_{i_j,a}: a\in\cup_{v\in[h,\delta-1]\cup\{j\}}A(\rho,v)
\}
\cup\{c_{i_{j'},a}: ~j'\in[0,h-1]\setminus\{j\}, ~a\in A(\rho, j)\}       \\
=&\cup_{a\in A(\rho,j)}\{c_{i_0,a}, c_{i_1,a}, \ldots, c_{i_{h-1},a}, c_{i_j,a(\rho,h)}, \ldots, c_{i_j,a(\rho,\delta-1)}\}.
\end{aligned}
\]
\end{lemma}

\begin{proof}
The proof is given in Appendix \ref{appendix-2}.
\end{proof}
According to Lemma \ref{thm2-lem}, the proof of Theorem \ref{thm2-repair-h} is straightforward.

In this section, we give a construction of optimal-access cooperative MSR codes from the perspective of designing its parity-check matrix.
In the following sections, from another perspective of transformation, we give a unified framework for building optimal-access cooperative MSR codes.
That is, we present a generic transformation converting an arbitrary $[n+d-k,d]$ MDS scalar code to an optimal-access cooperative MSR code with $2\leq h\leq n-k$ and $k\leq d\leq n-h$.
The generic transformation is composited from two basic transformations $\mathcal{T}_1$ and $\mathcal{T} _2$. Both $\mathcal{T}_1$ and $\mathcal{T} _2$ can convert an MDS scalar/array code to a new MDS array code with cooperative repair of specific $h\geq2$ node erasures.
More specifically,
given any $h$ nodes $\mathcal{H}\subseteq[n]$ with $|\mathcal{H}|=h$,
transformation $\mathcal{T}_1$ can convert an $(n+d-k,d)$ MDS scalar/array code to a new $(n+d-k,d)$ MDS array code with cooperative repair of the $h$ nodes in $\mathcal{H}$.
Given any $d-k+h$ nodes $\mathcal{F}\subseteq[n]$ with $|\mathcal{F}|=d-k+h$,
transformation $\mathcal{T}_2$ can convert an $(n,k)$ MDS scalar/array code to a new $(n,k)$ MDS array code with cooperative repair of any $h$ nodes in $\mathcal{F}$.
Indeed, the two basic transformations $\mathcal{T}_1$ and $\mathcal{T} _2$ can be seen as generalized from the two MDS array code building blocks $\mathcal{C}_{\mathrm{I}}$ and $\mathcal{C}_{\mathrm{II}}$ in Section \ref{sec0-building-block}, respectively.
Hence they serve as building blocks for establishing the generic transformation for constructing cooperative MSR codes.

Next, we present the basic transformations $\mathcal{T}_1$, $\mathcal{T}_2$ and the generic transformation $\mathcal{T}$ in Section \ref{transform-1}, Section \ref{transform-2}, and Section \ref{generic-transform}.
We also show that some previous works, such as \cite{Ye2016sub-,JLi-2018, Zhang2020} are included as special cases of our transformation construction.
This will be explained later in Subsection \ref{comparison}.

\section{The basic transformation $\mathcal{T}_1$}\label{transform-1}
In the following, let $2\leq h\leq n-k$ and $k\leq d\leq n-h$.
Recall the linear cooperative repair model introduced in Section \ref{introduction}, if the repair matrices $S_{i,j}^{(\mathcal{H},\mathcal{R})}$'s w.r.t. the failed node set $\mathcal{H}$ and helper node sets $\mathcal{R}$ are independent of the helper node choice $j$ and $\mathcal{R}$, then we call the repair strategy to be {\it helpers-independent}.
In this work, all the constructions are linear and have helpers-independent repair strategy.
Then for convenience, we will write $S_{i,j}^{(\mathcal{H},\mathcal{R})}=S_{i,\mathcal{H}}$ in the following of this paper.
It is worth noting that almost all of the previous constructions of cooperative MSR codes have helpers-independent repair strategy.

\subsection{The transformation $\mathcal{T}_1$}

Let $F$ be a finite field with $|F|\geq n+d-k$, and $\mathcal{C}_0$ be an $(n+d-k,d,\ell_0\geq1)$ MDS scalar/array code over $F$.
Then for any given $\mathcal{H}\subseteq[n]$ with $|\mathcal{H}|=h$,
the transformation $\mathcal{T}_1$ can convert $\mathcal{C}_0$ to a new $(n+d-k,d,\ell=(d-k+h)\ell_0)$ MDS array code $\mathcal{C}_{\mathcal{T}_1}$ with optimal cooperative repair of $h$ nodes in $\mathcal{H}$ using $2d-k$ helper nodes.
We call the $h$ nodes in $\mathcal{H}$ as goal nodes and suppose $\mathcal{H}=\{1,2,\ldots,h\}$ without loss of generality.

\noindent
\rule{\linewidth}{1pt}
\vspace{-0.3em}
\noindent
\textbf{The transformation $\mathcal{T}_1$}

\vspace{-0.3em}
\noindent
\rule{\linewidth}{0.5pt}

Let $F$ be a finite field with $|F|\geq n+d-k$, and $\mathcal{C}_0$ be an $(n+d-k,d,\ell_0\geq1)$ MDS scalar/array code over $F$.

\begin{itemize}
\item \textbf{Step 1. Generate $d-k$ instance codewords of $\mathcal{C}_0$.} \\
Generate $d-k$ instance codewords of $\mathcal{C}_0$, denoted as $(\bm{f}_{1}^{(i)}, \bm{f}_{2}^{(i)}, \ldots, \bm{f}_{n+d-k}^{(i)})\in\mathcal{C}_0$ for $i\in[h+1, d-k+h]$.

\item \textbf{Step 2. Generate $h$ instance codewords of $\mathcal{C}_0$.} \\
Generate $h$ instance codewords of $\mathcal{C}_0$, denoted as $(\bm{f}_{1}^{(j)}, \bm{f}_{2}^{(j)}, \ldots, \bm{f}_{n}^{(j)}, \bm{f}_{n+1}^{(j)}+\bm{f}_j^{(h+1)},\ldots, \bm{f}_{n+d-k}^{(j)}+\bm{f}_j^{(d-k+h)})\in\mathcal{C}_0$ for $j\in[h]$,
where for each $j\in[h]$ and $i\in[h+1, d-k+h]$, the vector symbols $\bm{f}_j^{(i)}$'s have been generated at the $j$-th coordinate of the $i$-th instance codeword of $\mathcal{C}_0$ generated in Step 1.

\item \textbf{Step 3. Construct a new codeword array from the $d-k+h$ instance codewords of $\mathcal{C}_0$.} \\
Construct a new codeword array using  the $d-k+h$ instance codewords of $\mathcal{C}_0$, which is illustrated in Table \ref{tab2}.
As illustrated in Table \ref{tab2}, each node $j\in[n+d-k]$ stores $d-k+h$ vector symbols $(\bm{f}_{j}^{(1)}, \bm{f}_{j}^{(2)}, \ldots, \bm{f}_{j}^{(d-k+h)})$ with a total length of $\ell=(d-k+h)\ell_0$.
Then, all the new codeword arrays form a new code $\mathcal{C}_{\mathcal{T}_1}$.

\end{itemize}

\vspace{-0.3em}
\noindent
\rule{1\linewidth}{1pt}

 \begin{table}[ht]
 \renewcommand{\arraystretch}{2.2}
	\centering
	\begin{threeparttable}
\caption{Stored data of $n$ storage nodes in the new codeword array of $\mathcal{C}_{\mathcal{T}_1}$}
\label{tab2}
\setlength{\tabcolsep}{1mm}{
\begin{tabular}{c|c|c|c|c|c|c|c|c|c}
\hline
\diagbox{Instances}{Nodes} &\textbf{Goal Node 1} &\textbf{Goal Node 2} & \bm{$\cdots$} &\bf{Goal Node \bm{$h$}} & Node $h+1$ & $\cdots$ & Node $n$ & $\cdots$ & Node $n+d-k$
\\   \hline
1 & \cellcolor{gray!25}$\bm{f}_1^{(1)}$ & \cellcolor{gray!25}$\bm{f}_2^{(1)}$ & \cellcolor{gray!25}$\cdots$ & \cellcolor{gray!25}$\bm{f}_h^{(1)}$ & \cellcolor{gray!25}$\bm{f}_{h+1}^{(1)}$ & \cellcolor{gray!25}$\cdots$ & \cellcolor{gray!25}$\bm{f}_{n}^{(1)}$ & \cellcolor{gray!25}$\cdots$  & \cellcolor{gray!25}$\bm{f}_{n+d-k}^{(1)}$
\\    \hline
2 & $\bm{f}_1^{(2)}$ & $\bm{f}_2^{(2)}$ & $\cdots$ & $\bm{f}_h^{(2)}$ & $\bm{f}_{h+1}^{(2)}$ & $\cdots$ & $\bm{f}_{n}^{(2)}$  & $\cdots$  & $\bm{f}_{n+d-k}^{(2)}$
\\  \hline
$\vdots$ & $\vdots$ & $\vdots$ & $\vdots$ & $\vdots$ & $\vdots$ & $\vdots$ & $\vdots$ & $\vdots$  &  $\vdots$
\\  \hline
$h$ & $\bm{f}_1^{(h)}$ & $\bm{f}_2^{(h)}$ & $\cdots$ & $\bm{f}_h^{(h)}$ & $\bm{f}_{h+1}^{(h)}$ & $\cdots$ & $\bm{f}_{n}^{(h)}$ & $\cdots$  & $\bm{f}_{n+d-k}^{(h)}$
\\  \hline
$h+1$ & \cellcolor{gray!25}$\bm{f}_1^{(h+1)}$ & $\bm{f}_2^{(h+1)}$ & $\cdots$ & $\bm{f}_h^{(h+1)}$ & $\bm{f}_{h+1}^{(h+1)}$ & $\cdots$ & $\bm{f}_{n}^{(h+1)}$  & $\cdots$  & $\bm{f}_{n+d-k}^{(h+1)}$
\\  \hline
$\vdots$ & \cellcolor{gray!25}$\vdots$ & $\vdots$ & $\vdots$ & $\vdots$ & $\vdots$ & $\vdots$ & $\vdots$    & $\vdots$  & $\vdots$
\\  \hline
$d-k+h$ & \cellcolor{gray!25}$\bm{f}_1^{(d-k+h)}$ & $\bm{f}_2^{(d-k+h)}$ & $\cdots$ & $\bm{f}_h^{(d-k+h)}$ & $\bm{f}_{h+1}^{(d-k+h)}$ & $\cdots$ & $\bm{f}_{n}^{(d-k+h)}$   & $\cdots$  & $\bm{f}_{n+d-k}^{(d-k+h)}$
\\  \hline
\end{tabular}}
\begin{tablenotes}
\item[*] Note that the stored data in each node is exhibited as a column array of $d-k+h$ vector symbols each of length $\ell_0$. In the context for simplicity, we also write it as a long row of length $(d-k+h)\ell_0$ if there is no ambiguity.
\end{tablenotes}
\end{threeparttable}
\end{table}

\begin{remark}
Generally, if the $h$ goal nodes $\mathcal{H}=\{i_1,i_2,\ldots,i_h\}\subseteq[n]$, then in Step 2,
the $h$ instance codewords of $\mathcal{C}_0$ are generated as $(\bm{f}_{1}^{(j)}, \bm{f}_{2}^{(j)}, \ldots, \bm{f}_{n}^{(j)}, \bm{f}_{n+1}^{(j)}+\bm{f}_{i_j}^{(h+1)}, \ldots, \bm{f}_{n+d-k}^{(j)}+\bm{f}_{i_j}^{(d-k+h)})\in\mathcal{C}_0$ for $j\in[h]$,
where for each $j\in[h]$ and $s\in[h+1, d-k+h]$, the vector symbols $\bm{f}_{i_j}^{(s)}$'s have been generated at the $i_j$-th coordinate of the $s$-th instance codeword of $\mathcal{C}_0$ in Step 1.
\end{remark}
In the following subsections, we illustrate that the code $\mathcal{C}_{\mathcal{T}_1}$ obtained from transformation $\mathcal{T}_1$  has MDS property and optimal-access property for repairing goal nodes $\mathcal{H}=\{1,2,\ldots,h\}$.

\subsection{MDS property}
\begin{theorem}\label{t1-mds}
The code $\mathcal{C}_{\mathcal{T}_1}$ obtained from transformation $\mathcal{T}_1$ is an $(n+d-k,d,\ell=(d-k+h)\ell_0)$ MDS array code.
\end{theorem}

\begin{proof}
We firstly claim that there are in total $d\ell=(d-k+h)d\ell_0$ information symbols in $\mathcal{C}_{\mathcal{T}_1}$.
According to Step 1, since $(\bm{f}_{1}^{(i)}, \bm{f}_{2}^{(i)}, \ldots, \bm{f}_{n+d-k}^{(i)})$ for $i\in[h+1, d-k+h]$ are $d-k$ codewords in $\mathcal{C}_0$. Then there are in total $(d-k)d\ell_0$ information symbols.
Moreover, since $(\bm{f}_{1}^{(j)}, \bm{f}_{2}^{(j)}, \ldots, \bm{f}_{n}^{(j)}, \bm{f}_{n+1}^{(j)}+\bm{f}_j^{(h+1)},\ldots, \bm{f}_{n+d-k}^{(j)}+\bm{f}_j^{(d-k+h)})$ for $j\in[h]$ are $h$ codewords in $\mathcal{C}_0$, where $\bm{f}_j^{(i)}$ for $j\in[h]$, $i\in[h+1, d-k+h]$ have been known, then there are additionally $hd\ell_0$ information symbols.
Note different choice of $d-k+h$ instances of $\mathcal{C}_0$ gives different codewords of $\mathcal{C}_{\mathcal{T}_1}$.
Thus in $\mathcal{C}_{\mathcal{T}_1}$, there are in total $(d-k)d\ell_0+hd\ell_0=(d-k+h)d\ell_0$ information symbols.

Now we prove the MDS property of $\mathcal{C}_{\mathcal{T}_1}$, that is,  any $d$ nodes in $\mathcal{C}_{\mathcal{T}_1}$ can reconstruct the whole codeword. Suppose the $d$ nodes $\{i_1,i_2,\ldots,i_d\}\subseteq[n+d-k]$ are connected by the data center.
Since for each $s\in[h+1,d-k+h]$, the vector $(\bm{f}_{1}^{(s)}, \bm{f}_{2}^{(s)}, \ldots, \bm{f}_{n+d-k}^{(s)})$ is a codeword of the $(n+d-k,d,\ell_0)$ MDS array code $\mathcal{C}_0$, then $\{\bm{f}_{i_1}^{(s)},\ldots,\bm{f}_{i_d}^{(s)}\}$ are able to recover $\bm{f}_{i}^{(s)}$ for all $i\in[n+d-k]$.
Moreover, for each $a\in[h]$, the vector $(\bm{f}_{1}^{(a)}, \bm{f}_{2}^{(a)}, \ldots, \bm{f}_{n}^{(a)}, \bm{f}_{n+1}^{(a)}+\bm{f}_a^{(h+1)}, \ldots, \bm{f}_{n+d-k}^{(a)}+\bm{f}_a^{(d-k+h)})$ forms an $(n+d-k,d,\ell_0)$ MDS array codeword, where $\bm{f}_a^{(s)}$, $a\in[h]$, $s\in[h+1,d-k+h]$ have been recovered.
Then for each $a\in[h]$, the $d$ vector symbols $\{\bm{f}_{i_1}^{(a)},\ldots,\bm{f}_{i_d}^{(a)}\}$ along with the recovered data are able to compute $\bm{f}_{i}^{(a)}$ for all $i\in[n+d-k]$.
Thus the whole codeword can be reconstructed and $\mathcal{C}_{\mathcal{T}_1}$ is an $(n+d-k,d,\ell=(d-k+h)\ell_0)$ MDS array code.
\end{proof}

\subsection{Optimal-access property of $\mathcal{H}$}\label{sub-t1-repair}
In the code $\mathcal{C}_{\mathcal{T}_1}$, we call a set of $2d-k$ helper nodes $D\subseteq[n+d-k]$ to be {\it constrained helper nodes}, if the last $d-k$ nodes, i.e., node $n+1,\ldots, n+d-k$ are contained in $D$. That is, $D=D'\cup[n+1,n+d-k]$ where $D'$ is a $d$-subset of $[n]$.
The optimal-access property of $\mathcal{H}$ in $\mathcal{C}_{\mathcal{T}_1}$ is ensured by using any $2d-k$ constrained helper nodes.
Next, in Theorem \ref{t1-repair1} we give the optimal-access  property of the $h$ goal nodes in $\mathcal{H}$.
We also characterize the corresponding repair matrices for repairing $\mathcal{H}$ and give an inherent property in Corollary \ref{t1-repair1-cor}.
Based on Corollary \ref{t1-repair1-cor}, in Definition \ref{def-property1} and Theorem \ref{t1-repair2}, we summarize the needed property such that the repair property of any other $h$ nodes $\mathcal{H}'\subseteq[n]$ can be retained after applying transformation $\mathcal{T}_1$.

\begin{theorem}\label{t1-repair1}
The code $\mathcal{C}_{\mathcal{T}_1}$ can cooperatively repair the $h$ goal nodes $\mathcal{H}=\{1,2,\ldots,h\}$ with optimal access using any $2d-k$ constrained helper nodes.
 \end{theorem}

\begin{proof}
Suppose the $h$ goal nodes $\{1,2,\ldots,h\}$ are failed. For each $i\in[h]$, let $\mathcal{R}_i\subseteq[n+d-k]\setminus[h]$ with $|\mathcal{R}_i|=2d-k$ and $[n+1,n+d-k]\subseteq\mathcal{R}_i$ be the set of $2d-k$ constrained helper nodes connected by node $i$.
In the download phase, each node $i\in[h]$ downloads the $i$-th vector symbol $\bm{f}_{j}^{(i)}$ from the helper node $j\in\mathcal{R}_i$.
Observe that for each $i\in[h]$, it has
$(\bm{f}_{1}^{(i)},\ldots,\bm{f}_{n}^{(i)},\bm{f}_{n+1}^{(i)}+\bm{f}_{i}^{(h+1)},\ldots,\bm{f}_{n+d-k}^{(i)}+\bm{f}_{i}^{(d-k+h)})$ forms a codeword of the $(n+d-k,d,\ell_0)$ MDS array code $\mathcal{C}_0$.
Then each node $i\in[h]$ using the $d$ downloaded vector symbols $\{\bm{f}_{j}^{(i)}: j\in\mathcal{R}_i\setminus[n+1,n+d-k]\}$ can recover the whole codeword, hence the data
\begin{equation}\label{t1-repair1-eq1}
\{\bm{f}_{j}^{(i)}: j\in[n+d-k]\}\cup\{\bm{f}_{i}^{(a)}: a\in[h+1,d-k+h]\},
\end{equation}
since $\bm{f}_{j}^{(i)}, j\in[n+1,n+d-k]$ are downloaded data.

In the collaboration phase, for each $i\in[h]$ and $i'\in[h]\setminus\{i\}$, node $i'$ transmits the vector symbol $\bm{f}_{i}^{(i')}$ to node $i$.
Thus for $i\in[h]$, the erased data $\{\bm{f}_{i}^{(a)}: a\in[d-k+h]\}$ at node $i$ can be recovered.
The total repair bandwidth is $(d+h-1)h\ell_0=\frac{(d+h-1)h\ell}{d-k+h}$ symbols in $F$,
and the amount of accessed data at the helper nodes is $dh\ell_0=\frac{dh\ell}{d-k+h}$ symbols in $F$, both achieving the cut-set bound in \eqref{bw-bound}.
Thus it satisfies optimal-access property.
\end{proof}

According to the proof of Theorem \ref{t1-repair1}, one can find the following property:

\begin{corollary}\label{t1-repair1-cor}
In the repair of $\mathcal{H}=[h]$, let $S_{i,\mathcal{H}}$ be the corresponding repair matrix for repairing node $i\in\mathcal{H}$, let $\mathcal{R}_i\subseteq[n+d-k]\setminus\mathcal{H}$ with $|\mathcal{R}_i|=2d-k$ satisfying $[n+1,n+d-k]\subseteq\mathcal{R}_i$ be the set of constrained helper nodes connected by node $i$.
Then for $i\in\mathcal{H}$, it has that
 $\{S_{i,\mathcal{H}}\bm{c}_j^{\top}: j\in \mathcal{R}_i\}$ are able to recover $\{S_{i,\mathcal{H}}\bm{c}_j^{\top}: j\in[n+d-k]\}$, where $(\bm{c}_1,\ldots,\bm{c}_{n+d-k})$ represents a codeword of $\mathcal{C}_{\mathcal{T}_1}$.
 \end{corollary}

Actually, the property in Corollary \ref{t1-repair1-cor} is summarized from the cooperative repair process of the goal nodes $\mathcal{H}$ as illustrated in Theorem \ref{t1-repair1}.
In the following, we take out the two properties in  Theorem \ref{t1-repair1} and Corollary \ref{t1-repair1-cor} and give a definition for the case of repairing a general $h$-subset $\mathcal{\bar{H}}\subseteq[n]$, which will be used later.

\begin{definition}\label{def-property1}
Let $\bar{\mathcal{C}}$ be an $(n+d-k,d)$ MDS scalar/array code, and let $\bar{\mathcal{H}}\subseteq[n]$ with $|\bar{\mathcal{H}}|=h$. Define (P1) and (P2) to be the two properties of $\bar{\mathcal{C}}$ if $\bar{\mathcal{C}}$ has. That is,
\begin{itemize}
\item[(P1)] $\bar{\mathcal{C}}$ enables optimal-access  cooperative repair of the $h$ nodes $\bar{\mathcal{H}}$ using  any $2d-k$ constrained helper nodes.
\item[(P2)]
In the repair of $\bar{\mathcal{H}}$, denote $S_{i,\bar{\mathcal{H}}}$ to be the corresponding repair matrix for repairing node $i\in\bar{\mathcal{H}}$, let $\mathcal{R}_i\subseteq[n]\setminus\bar{\mathcal{H}}$ with $|\mathcal{R}_i|=2d-k$ satisfying $[n+1,n+d-k]\subseteq\mathcal{R}_i$ be the set of constrained helper nodes connected by node $i$.
Then for $i\in\bar{\mathcal{H}}$, it has that
 $\{S_{i,\bar{\mathcal{H}}}\bm{c}_j^{\top}: j\in \mathcal{R}_i\}$ are able to recover $\{S_{i,\bar{\mathcal{H}}}\bm{c}_j^{\top}: j\in[n+d-k]\}$, where $(\bm{c}_1,\ldots,\bm{c}_{n+d-k})$ represents a codeword of $\bar{\mathcal{C}}$.
\end{itemize}
\end{definition}

Next, we are able to give Theorem \ref{t1-repair2}.

\begin{theorem}\label{t1-repair2}
Let $\mathcal{H}'\subseteq[n]$ with $|\mathcal{H}'|=h$ and $\mathcal{H}'\neq\mathcal{H}$, where $\mathcal{H}$ is the goal node set. If $\mathcal{C}_0$ satisfies the two properties (P1) and (P2) in Definition \ref{def-property1} for repairing $\mathcal{H}'$, then $\mathcal{C}_{\mathcal{T}_1}$ retains the two properties (P1) and (P2) for repairing $\mathcal{H}'$.
\end{theorem}

\begin{proof}
Suppose $\mathcal{C}_0$ satisfies the two properties (P1) and (P2) in Definition \ref{def-property1} for repairing $\mathcal{H}'$. Based on the repair property of $\mathcal{H}'$ and the corresponding repair matrices $S_{i,\mathcal{H}'}, i\in\mathcal{H}'$ in $\mathcal{C}_0$, we firstly define the repair matrices $\tilde{S}_{i,\mathcal{H}'}, i\in\mathcal{H}'$ in $\mathcal{C}_{\mathcal{T}_1}$.
Denote $I_{d-k+h}$ to be the identity matrix with order $d-k+h$.
For $i\in\mathcal{H}'$, define
$\tilde{S}_{i,\mathcal{H}'}=I_{d-k+h}\otimes S_{i,\mathcal{H}'}$, where ``$\otimes$'' represents the tensor product.
Then $\tilde{S}_{i,\mathcal{H}'}\in F^{\frac{\ell}{d-k+h}\times\ell}$ where $\ell=(d-k+h)\ell_0$.

Now we prove $\mathcal{C}_{\mathcal{T}_1}$ retains the property (P1) for repairing $\mathcal{H}'$.
Recall the transformation $\mathcal{T}_1$ and Table \ref{tab2}, for each $j\in[n+d-k]$, node $j$ of $\mathcal{C}_{\mathcal{T}_1}$ stores
$\bm{\tilde{c}}_j=(\bm{f}_j^{(1)},\bm{f}_j^{(2)},\ldots,\bm{f}_j^{(d-k+h)})\in F^{\ell}.$
Suppose $h$ nodes $\mathcal{H}'$ in $\mathcal{C}_{\mathcal{T}_1}$ are failed, i.e., $\bm{\tilde{c}}_i=(\bm{f}_i^{(1)},\bm{f}_i^{(2)},\ldots,\bm{f}_i^{(d-k+h)})$, $i\in\mathcal{H}'$ are erased.
For each $i\in\mathcal{H}'$, let $\mathcal{R}_i\subseteq[n+d-k]\setminus\mathcal{H}'$ with $|\mathcal{R}_i|=2d-k$ satisfying $[n+1,n+d-k]\subseteq\mathcal{R}_i$ be a set of $2d-k$ constrained  helper nodes connected by node $i$.

In the download phase, for $i\in\mathcal{H}'$, node $i$ downloads $\tilde{S}_{i,\mathcal{H}'}\bm{\tilde{c}}_j^{\top}$ from each helper node $j\in\mathcal{R}_i$, where
\begin{equation}\label{t1-repair-eq1}
\tilde{S}_{i,\mathcal{H}'}\bm{\tilde{c}}_j^{\top}
=\begin{pmatrix}
S_{i,\mathcal{H}'}\bm{f}_j^{(1)\top} \\
S_{i,\mathcal{H}'}\bm{f}_j^{(2)\top} \\
\vdots \\
S_{i,\mathcal{H}'}\bm{f}_j^{(d-k+h)\top}
\end{pmatrix}.
\end{equation}
Consider each $s\in[h+1,d-k+h]$, it has $(\bm{f}_1^{(s)},\ldots,\bm{f}_{n+d-k}^{(s)})$ forms a codeword of $\mathcal{C}_0$.
Then by property (P1) of $\mathcal{C}_0$ for repairing $\mathcal{H}'$, for each $s\in[h+1,d-k+h]$, node $i\in\mathcal{H}'$ can recover the erased data $\bm{f}_i^{(s)}$ through the cooperative repair process of $\mathcal{H}'$ in $\mathcal{C}_0$.
Meanwhile, by property (P2) of $\mathcal{C}_0$, the downloaded data of node $i$ are able to recover the data
\begin{equation}\label{t1-repair-eq2}
\{S_{i,\mathcal{H}}\bm{f}_j^{(s)\top}: j\in[n+d-k], s\in[h+1,d-k+h]\}.
\end{equation}
Then, consider each $s\in[h]$, it has $(\bm{f}_1^{(s)},\ldots,\bm{f}_n^{(s)},\bm{f}_{n+1}^{(s)}+\bm{f}_{s}^{(h+1)},\ldots, \bm{f}_{n+d-k}^{(s)}+\bm{f}_{s}^{(d-k+h)})$ forms a codeword of $\mathcal{C}_0$.
Note that for each $s\in[h]$, node $i\in\mathcal{H}'$ has recovered $S_{i,\mathcal{H}}\bm{f}_{s}^{(h+1)\top}, \ldots, S_{i,\mathcal{H}}\bm{f}_{s}^{(d-k+h)\top}$ from its downloaded data according to \eqref{t1-repair-eq2}.
Then for each $s\in[h]$, node $i$ using its downloaded data can further compute
\begin{equation}\label{t1-repair-eq3}
\{S_{i,\mathcal{H}'}\bm{f}_j^{(s)\top}: j\in\mathcal{R}_i\setminus[n+1,n+d-k]\}
\cup\{
S_{i,\mathcal{H}'}(\bm{f}_{n+1}^{(s)}+\bm{f}_{s}^{(h+1)})^{\top},\ldots, S_{i,\mathcal{H}'}(\bm{f}_{n+d-k}^{(s)}+\bm{f}_{s}^{(d-k+h)})^{\top}
\},
\end{equation}
where we note that $[n+1,n+d-k]\subseteq\mathcal{R}_i$.
The computed data in \eqref{t1-repair-eq3} can be seen as the helper data of node $i$ in the $s$-th codeword of $\mathcal{C}_0$.
Then, by property (P1) of $\mathcal{C}_0$, for each $s\in[h]$, through a cooperative repair process of $\mathcal{H}'$ in $\mathcal{C}_0$, node $i$ can recover the erased data $\bm{f}_{i}^{(s)}$ for $i\in\mathcal{H}'$.

Therefore, node $i$ can recover $\{\bm{f}_{i}^{(s)}: s\in[d-k+h]\}$ for $i\in\mathcal{H}'$, and the $h$ nodes are cooperatively repaired.
Besides, by property (P2) of $\mathcal{C}_0$, for each $s\in[h]$, node $i\in\mathcal{H}'$ using the known data in \eqref{t1-repair-eq3} can recover the following data
\[
\{S_{i,\mathcal{H}'}\bm{f}_j^{(s)\top}: j\in[n]\}
\cup\{
S_{i,\mathcal{H}'}(\bm{f}_{n+1}^{(s)}+\bm{f}_{s}^{(h+1)})^{\top},\ldots, S_{i,\mathcal{H}'}(\bm{f}_{n+d-k}^{(s)}+\bm{f}_{s}^{(d-k+h)})^{\top}
\}.
\]
Thus, for $i\in\mathcal{H}'$, node $i$ using the downloaded data can recover the data
\begin{equation}\label{t1-repair-eq4}
\{S_{i,\mathcal{H}'}\bm{f}_j^{(s)\top}: j\in[n+d-k], s\in[h]\}
\end{equation}
by observing that $[n+1,n+d-k]\subseteq\mathcal{R}_i$.
Moreover, since $\tilde{S}_{i,\mathcal{H}'}=I_{d-k+h}\otimes S_{i,\mathcal{H}'}$ for $i\in\mathcal{H}'$, then $\mathcal{C}_{\mathcal{T}_1}$ retains the optimal-access property of $\mathcal{C}_0$ for repairing $\mathcal{H}'$.

In the following, we prove $\mathcal{C}_{\mathcal{T}_1}$ satisfies property (P2) for repairing $\mathcal{H}'$.
Recall the transformation $\mathcal{T}_1$ and Table \ref{tab2}, we know that for $i\in\mathcal{H}'$ and $j\in[n+d-k]$, the data $\tilde{S}_{i,\mathcal{H}'}\bm{\tilde{c}}_j^{\top}$ has the form in \eqref{t1-repair-eq1}.
According to the proof of property (P1) of $\mathcal{C}_{\mathcal{T}_1}$, for $i\in\mathcal{H}'$, node $i$ using the downloaded data $\{\tilde{S}_{i,\mathcal{H}'}\bm{\tilde{c}}_j^{\top}: j\in\mathcal{R}_i\}$ can recover the data in \eqref{t1-repair-eq2} and \eqref{t1-repair-eq4}, which are exactly the data $\{\tilde{S}_{i,\mathcal{H}'}\bm{\tilde{c}}_j^{\top}: j\in[n+d-k]\}$.

This completes the proof.
\end{proof}

\subsection{A shortened code}
It is worth noting that the $(n+d-k,d,\ell=(d-k+h)\ell_0)$ MDS array code $\mathcal{C}_{\mathcal{T}_1}$ derived from transformation $\mathcal{T}_1$ can be further shortened to obtain an $(n,k,\ell=(d-k+h)\ell_0)$ MDS array code for cooperatively repairing the $h$ goal nodes $\mathcal{H}$ by using any $d$ helper nodes, as illustrated in the following lemma.

\begin{lemma}\label{shorten-code}
Let $\mathcal{C}_{\mathcal{T}_1}$ be an $(n+d-k,d,\ell=(d-k+h)\ell_0)$ MDS array code derived from transformation $\mathcal{T}_1$.
Let $\mathcal{C}'_{\mathcal{T}_1}$ be the shortened code by shortening $\mathcal{C}_{\mathcal{T}_1}$ at the last $d-k$ nodes, i.e. node $n+1,\ldots,n+d-k$.
Then, $\mathcal{C}'_{\mathcal{T}_1}$ is an $(n,k,\ell=(d-k+h)\ell_0)$ MDS array code with the following repair properties:
(1) $\mathcal{C}'_{\mathcal{T}_1}$ enables optimal-access cooperative repair of the $h$ goal nodes $\mathcal{H}$ by using any $d$ helper nodes.
(2) $\mathcal{C}'_{\mathcal{T}_1}$ retains the optimal-access property of any $h$ nodes $\mathcal{H'}\neq\mathcal{H}$ by using any $d$ helper nodes if the original code $\mathcal{C}_0$ satisfies (P1) and (P2) in Definition \ref{def-property1} for repairing $\mathcal{H'}$.
\end{lemma}

\begin{proof}
We firstly define the code $\mathcal{C}'_{\mathcal{T}_1}$. Consider all the codewords of $\mathcal{C}_{\mathcal{T}_1}$ that have zeros in the last $d-k$ nodes, and then puncture these codewords in the last $d-k$ nodes, it gives the code $\mathcal{C}'_{\mathcal{T}_1}$.
According to the MDS property of $\mathcal{C}_{\mathcal{T}_1}$, the last $d-k$ nodes can be seen as message coordinates, hence $\mathcal{C}'_{\mathcal{T}_1}$ contains $k\ell$ different codewords.
In the data reconstruction, any $k$ nodes in $\mathcal{C}'_{\mathcal{T}_1}$ plus $d-k$ imaginary nodes in $[n+1,n+d-k]$ that store all zeros correspond to $d$ nodes in $\mathcal{C}_{\mathcal{T}_1}$ which uniquely determines the whole codeword.
Therefore, any $k$ nodes in $\mathcal{C}'_{\mathcal{T}_1}$ can reconstruct the whole codeword and $\mathcal{C}'_{\mathcal{T}_1}$ is an $(n,k,\ell=(d-k+h)\ell_0)$ MDS array code.
Moreover, the cooperative repair of $h$ failed nodes in
$\mathcal{C}'_{\mathcal{T}_1}$ with each connecting to $d$ helper nodes can be done as the
cooperative repair of the $h$ failed nodes in $\mathcal{C}_{\mathcal{T}_1}$ with each connecting to $2d-k$ constrained helper nodes including
the $d$ helper nodes and $d-k$ imaginary nodes in $[n+1,n+d-k]$ that store all zeros. Thus the proof is completed.
\end{proof}

\section{The basic transformation $\mathcal{T}_2$}\label{transform-2}
We present the basic transformation $\mathcal{T}_2$ which converts an $(n,k)$ MDS scalar/array code to a new $(n,k)$ MDS array code for cooperatively repairing any $h$ nodes in a predefined set $\mathcal{F}\subseteq[n]$ with $|\mathcal{F}|=d-k+h$.
The transformation $\mathcal{T}_2$ generalizes a previous work in \cite{JLi-2018} for constructing optimal-access MSR codes with single erasures.

\subsection{The transformation $\mathcal{T}_2$}
Let $F$ be a finite field with $|F|>n$, and $\mathcal{C}_0$ be an $(n,k,\ell_0\geq1)$ MDS scalar/array code over $F$.
Then for any given $\mathcal{F}\subseteq[n]$ with $|\mathcal{F}|=d-k+h$,
the transformation $\mathcal{T}_2$ can convert $\mathcal{C}_0$ to an $(n,k,\ell=(d-k+h)\ell_0)$ MDS array code $\mathcal{C}_{\mathcal{T}_2}$ with cooperative repair of any $h$ nodes in $\mathcal{F}$ using $d$ helper nodes.
That is, $\mathcal{C}_{\mathcal{T}_2}$ actually enables cooperative repair of in total $\binom{d-k+h}{h}$ erasure patterns in $\mathcal{F}$.
We call the $d-k+h$ nodes in $\mathcal{F}$ as goal nodes and suppose $\mathcal{F}=\{1,2,\ldots,d-k+h\}$ without loss of generality.

\noindent
\rule{\linewidth}{1pt}
\vspace{-0.3em}
\noindent
\textbf{The transformation $\mathcal{T}_2$}

\vspace{-0.3em}
\noindent
\rule{\linewidth}{0.5pt}

Let $F$ be a finite field with $|F|>n$, and $\mathcal{C}_0$ be an $(n,k,\ell_0\geq1)$ MDS scalar/array code over $F$.
For $1\leq i<j\leq d-k+h$ and $a\in[4]$, choose $\eta_{i,j}^{(a)}$ to be nonzero elements in $F$ satisfying that $\eta_{i,j}^{(1)}\eta_{i,j}^{(4)}\neq \eta_{i,j}^{(2)}\eta_{i,j}^{(3)}$.

\begin{itemize}
\item \textbf{Step 1. Generate $d-k+h$ instance codewords of $\mathcal{C}_0$.} \\
Generate $d-k+h$ instance codewords of $\mathcal{C}_0$, i.e., $(\bm{f}_{1}^{(s)}, \bm{f}_{2}^{(s)}, \ldots, \bm{f}_{n}^{(s)})\in\mathcal{C}_0$ for $s\in[d-k+h]$.

\item \textbf{Step 2. Space sharing the $d-k+h$ instance codewords of $\mathcal{C}_0$ to obtain an intermediate code $\mathcal{C}_1$.} \\
Space sharing the $d-k+h$ instance codewords of $\mathcal{C}_0$ generated in Step 1 to obtain an $(n,k,\ell=(d-k+h)\ell_0)$ MDS array code $\mathcal{C}_1$.
That is, for $j\in[n]$, node $j$ of $\mathcal{C}_1$ stores $(\bm{f}_{j}^{(1)},\bm{f}_{j}^{(2)},\ldots,\bm{f}_{j}^{(d-k+h)})\in F^{\ell}$. For simplicity, every codeword of $\mathcal{C}_1$ can be exhibited as an array of size $(d-k+h)\times n$ with each node column $j\in[n]$ storing $d-k+h$ vector symbols $\bm{f}_{j}^{(s)}$, $s\in[d-k+h]$ each of length $\ell_0$.

\item \textbf{Step 3. New codeword array construction.} \\
Restructure the codeword array of $\mathcal{C}_1$ in Step 2 at the first $d-k+h$ node coordinates.
More precisely,
we construct a new codeword array by the following manner:
\begin{itemize}
\item For $1\leq i<j\leq d-k+h$, set the $(i,j)$-th entry of the array to be $\eta_{i,j}^{(1)}\bm{f}_{j}^{(i)}+\eta_{i,j}^{(2)}\bm{f}_{i}^{(j)}$, and set the $(j,i)$-th entry to be $\eta_{i,j}^{(3)}\bm{f}_{j}^{(i)}+\eta_{i,j}^{(4)}\bm{f}_{i}^{(j)}$.
\item The other entries of the array remain the same as that in $\mathcal{C}_1$.
\end{itemize}
The new codeword array is illustrated in Table \ref{tab3}. Then all the new codeword arrays form the code $\mathcal{C}_{\mathcal{T}_2}$.

\end{itemize}

\vspace{-0.3em}
\noindent
\rule{1\linewidth}{1pt}
 \begin{table}[ht]
 \renewcommand{\arraystretch}{1.8}
	\centering
	\begin{threeparttable}
\caption{Stored data of $n$ storage nodes in the new codeword array of $\mathcal{C}_{\mathcal{T}_2}$}
\label{tab3}
\setlength{\tabcolsep}{0.2mm}{
\begin{tabular}{c|c|c|c|c|c|c|c|c}
\hline
\diagbox{Instances}{Nodes} &\textbf{Goal Node 1} &\textbf{Goal Node 2} &\textbf{Goal Node 3} & \textbf{$\cdots$} &\textbf{Goal Node} \bm{$d-k+h$}
& Node $d-k+h+1$ & $\cdots$ & Node $n$
\\   \hline
1 & $\bm{f}_1^{(1)}$ & \cellcolor{gray!45}$\eta_{1,2}^{(1)}\bm{f}_2^{(1)}+\eta_{1,2}^{(2)}\bm{f}_1^{(2)}$ & \cellcolor{gray!30}$\eta_{1,3}^{(1)}\bm{f}_3^{(1)}+\eta_{1,3}^{(2)}\bm{f}_1^{(3)}$ & $\cdots$ &
\cellcolor{gray!15}\makecell[c]{$\eta_{1,d-k+h}^{(1)}\bm{f}_{d-k+h}^{(1)}$ \\
 $+\eta_{1,d-k+h}^{(2)}\bm{f}_1^{(d-k+h)}$ }
 & $\bm{f}_{d-k+h+1}^{(1)}$ &  $\cdots$ & $\bm{f}_n^{(1)}$     \\    \hline
2 & \cellcolor{gray!45}$\eta_{1,2}^{(3)}\bm{f}_2^{(1)}+\eta_{1,2}^{(4)}\bm{f}_1^{(2)}$ &  $\bm{f}_2^{(2)}$ & $\eta_{2,3}^{(1)}\bm{f}_3^{(2)}+\eta_{2,3}^{(2)}\bm{f}_2^{(3)}$ &  $\cdots$ &
\makecell[c]{$\eta_{2,d-k+h}^{(1)}\bm{f}_{d-k+h}^{(2)}$ \\
$+\eta_{2,d-k+h}^{(2)}\bm{f}_2^{(d-k+h)}$ }
 & $\bm{f}_{d-k+h+1}^{(2)}$ &  $\cdots$&  $\bm{f}_n^{(2)} $
\\  \hline
$3$ & \cellcolor{gray!30}$\eta_{1,3}^{(3)}\bm{f}_3^{(1)}+\eta_{1,3}^{(4)}\bm{f}_1^{(3)}$ &  $\eta_{2,3}^{(3)}\bm{f}_3^{(2)}+\eta_{2,3}^{(4)}\bm{f}_2^{(3)}$  & $\bm{f}_3^{(3)}$  & $\cdots$ &
\makecell[c]{$\eta_{3,d-k+h}^{(1)}\bm{f}_{d-k+h}^{(3)}$ \\
$+\eta_{3,d-k+h}^{(2)}\bm{f}_3^{(d-k+h)}$ }
 & $\bm{f}_{d-k+h+1}^{(3)}$ & $\cdots$&  $\bm{f}_n^{(3)} $
\\  \hline
$\vdots$ &  $\vdots$ & $\vdots$ &  $\vdots$ &  $\ddots$ & $\vdots$
&$\cdots$ &  $\vdots$ & $\cdots$
\\  \hline
$d-k+h$ &
\cellcolor{gray!15}\makecell[c]{$\eta_{1,d-k+h}^{(3)}\bm{f}_{d-k+h}^{(1)}$ \\
$+\eta_{1,d-k+h}^{(4)}\bm{f}_1^{(d-k+h)}$}
 & \makecell[c]{$\eta_{2,d-k+h}^{(3)}\bm{f}_{d-k+h}^{(2)}$ \\
 $+\eta_{2,d-k+h}^{(4)}\bm{f}_2^{(d-k+h)}$}
 & \makecell[c]{$\eta_{3,d-k+h}^{(3)}\bm{f}_{d-k+h}^{(3)}$ \\
 $+\eta_{3,d-k+h}^{(4)}\bm{f}_3^{(d-k+h)}$}
  & $\cdots$ & $\bm{f}_{d-k+h}^{(d-k+h)}$
   & $\bm{f}_{d-k+h+1}^{(d-k+h)}$ &  $\cdots$&  $\bm{f}_n^{(d-k+h)}$
\\  \hline
\end{tabular}}
\begin{tablenotes}
\item[*] Note that in the table, the stored data in each node is exhibited as a column array of $d-k+h$ vector symbols of length $\ell_0$. In the context for simplicity, we write it as a long row of length $(d-k+h)\ell_0$.
\end{tablenotes}
\end{threeparttable}
\end{table}

\subsection{MDS property}

\begin{theorem}\label{t2-mds}
The code $\mathcal{C}_{\mathcal{T}_2}$ obtained from transformation $\mathcal{T}_2$ is an $(n,k,\ell=(d-k+h)\ell_0)$ MDS array code.
\end{theorem}

\begin{proof}
It is easy to see that there are in total $k(d-k+h)\ell_0=k\ell$ information symbols in $\mathcal{C}_{\mathcal{T}_2}$.
Next we prove that any $k$ nodes in $\mathcal{C}_{\mathcal{T}_2}$ can reconstruct the whole codeword.
Suppose the $k$ connected nodes are $U=U_1\cup U_2$ with $U_1\subseteq[d-k+h]$, $U_2\subseteq[d-k+h+1,n]$ and $|U_1\cup U_2|=k$.
Next we illustrate that using the stored data at the $k$ connected nodes, one can recover all the data $\bm{f}_{i}^{(s)}$, $i\in[n]$, $s\in[d-k+h]$.

At first, for any $i,j\in U_1$ with $i<j$, using the $i$-th data symbol in node $j$, i.e., $\eta_{i,j}^{(1)}\bm{f}_{j}^{(i)}+\eta_{i,j}^{(2)}\bm{f}_{i}^{(j)}$, and the $j$-th data symbol in node $i$, i.e., $\eta_{i,j}^{(3)}\bm{f}_{j}^{(i)}+\eta_{i,j}^{(4)}\bm{f}_{i}^{(j)}$, one can compute $\bm{f}_{j}^{(i)}$ and $\bm{f}_{i}^{(j)}$ since $\eta_{i,j}^{(1)}\eta_{i,j}^{(4)}\neq \eta_{i,j}^{(2)}\eta_{i,j}^{(3)}$.
That is, one can recover the data $\bm{f}_{i}^{(j)}$ for $i,j\in U_1$ and $i\neq j$.
Besides, $\bm{f}_{i}^{(i)}$ for $i\in U_1$ are known, thus one can obtain the data $\{\bm{f}_{i}^{(j)}: i,j\in U_1\}$.

Note that for each $s\in[d-k+h]$, $(\bm{f}_{1}^{(s)},\bm{f}_{2}^{(s)},\ldots,\bm{f}_{n}^{(s)})$ forms a codeword of the $(n,k,\ell_0)$ MDS code $\mathcal{C}_0$.
Then for each $s\in U_1\subseteq[d-k+h]$, using the known data $\bm{f}_{i}^{(s)}, i\in U_2$ and the previously recovered data $\bm{f}_{i}^{(s)}, i\in U_1$, in total $k$ vector symbols, one can  recover $\{\bm{f}_{i}^{(s)}: i\in[n]\}$.
Thus, the data $\{\bm{f}_{i}^{(s)}: i\in[n], s\in U_1\}$ can be reconstructed.

Moreover, for each $i\in U_1$ and $s\in [d-k+h]\setminus U_1$, the $s$-th data symbol in the connected node $i$ stores $\eta_{s,i}^{(1)}\bm{f}_{i}^{(s)}+\eta_{s,i}^{(2)}\bm{f}_{s}^{(i)}$ if $s<i$, and stores $\eta_{i,s}^{(3)}\bm{f}_{s}^{(i)}+\eta_{i,s}^{(4)}\bm{f}_{i}^{(s)}$ if $s>i$.
Since $\bm{f}_{s}^{(i)}$ has been previously recovered and the coefficient $\eta_{s,i}^{(1)}$(or $\eta_{i,s}^{(4)}$) is nonzero, then one can compute the data $\bm{f}_{i}^{(s)}$.
That is, all the vector symbols $\{\bm{f}_{i}^{(s)}: i\in U_1, s\in [d-k+h]\setminus U_1\}$ can be obtained.

Also, for each $s\in [d-k+h]\setminus U_1$, note $(\bm{f}_{1}^{(s)},\bm{f}_{2}^{(s)},\ldots,\bm{f}_{n}^{(s)})$ is a codeword of $\mathcal{C}_0$. Then for each $s\in [d-k+h]\setminus U_1$, using the known data $\{\bm{f}_{i}^{(s)}$, $i\in U_2\}$ at the connected nodes and the previously recovered data $\{\bm{f}_{i}^{(s)}: i\in U_1\}$, one can recover $\{\bm{f}_{i}^{(s)}: i\in[n]\}$.
Thus $\{\bm{f}_{i}^{(s)}: i\in[n], s\in [d-k+h]\setminus U_1\}$ can be reconstructed.

Therefore, all the symbols $\{\bm{f}_{i}^{(s)}: i\in[n], s\in [d-k+h]\}$ are recovered and the whole codeword can be reconstructed according to the codeword array in Table \ref{tab3}.
Thus $\mathcal{C}_{\mathcal{T}_2}$ satisfies the MDS property.
\end{proof}

\subsection{Optimal-access property of any $h$ nodes in the goal node set $\mathcal{F}$}
Given the goal node set $\mathcal{F}\subseteq[n]$ with $|\mathcal{F}|=d-k+h$, when repairing any $h$ nodes $\mathcal{H}\subseteq\mathcal{F}$, we call a set of $d$ helper nodes $D\subseteq[n]\setminus\mathcal{H}$ with $|D|=d$ to {\it have local property} if the $d-k$ surviving goal nodes in $\mathcal{F}\setminus\mathcal{H}$ are contained in $D$, i.e., $(\mathcal{F}\setminus\mathcal{H})\subseteq D$.
That is, $D=D'\cup(\mathcal{F}\setminus\mathcal{H})$ where $D'$ is a $k$-subset of $[n]\setminus\mathcal{F}$.
Next, in Theorem \ref{t2-repair1}, we give the optimal-access property of any $h$ goal nodes $\mathcal{H}\subseteq\mathcal{F}$ by using any $d$ helper nodes with local property.
In Corollary \ref{t2-repair1-cor}, we summarize an inherent property relating the repair matrices for repairing $\mathcal{H}$.
In Definition \ref{def-property2} and Theorem \ref{t2-repair2}, we characterize the needed property such that the optimal-access property of any $h$ nodes $\mathcal{H}'\subseteq[n]\setminus\mathcal{F}$ can be retained after applying transformation $\mathcal{T}_2$.

\begin{theorem}\label{t2-repair1}
The code $\mathcal{C}_{\mathcal{T}_2}$ can cooperatively repair any $h$ goal nodes $\mathcal{H}\subseteq\mathcal{F}=\{1,2,\ldots,d-k+h\}$ with optimal access by using any $d$ helper nodes with local property.
\end{theorem}

\begin{proof}
Let $\mathcal{H}=\{i_1,i_2,\ldots,i_h\}\subseteq[d-k+h]$ with $|\mathcal{H}|=h$ be the failed node set.
For each $i\in\mathcal{H}$, let $\mathcal{R}_i=\bar{\mathcal{R}}_i\cup([d-k+h]\setminus\mathcal{H})$ with $\bar{\mathcal{R}}_i\subseteq[n]\setminus[d-k+h]$ and $|\bar{\mathcal{R}}_i|=k$ be the set of $d$ helper nodes connected by node $i$.

In the download phase, for each $j\in[h]$, node $i_j$ downloads the $i_j$-th data symbol from each helper node $p\in\mathcal{R}_{i_j}$.
That is, node $i_j$ downloads the following data:
\begin{equation}\label{t2-repair1-eq1}
\{\bm{f}_p^{(i_j)}: p\in\bar{\mathcal{R}}_{i_j}\}
\cup\{\eta_{i_j,p}^{(1)}\bm{f}_p^{(i_j)}+\eta_{i_j,p}^{(2)}\bm{f}_{i_j}^{(p)}: p\in[d-k+h]\setminus\mathcal{H}, p>i_j\}
\cup\{\eta_{p,i_j}^{(3)}\bm{f}_{i_j}^{(p)}+\eta_{p,i_j}^{(4)}\bm{f}_p^{(i_j)}: p\in[d-k+h]\setminus\mathcal{H}, p<i_j\}.
\end{equation}
For each $j\in[h]$, since $(\bm{f}_{1}^{(i_j)},\bm{f}_{2}^{(i_j)},\ldots,\bm{f}_{n}^{(i_j)})$ forms a codeword of the $(n,k,\ell_0)$ MDS code $\mathcal{C}_0$, then the $k$ downloaded vector symbols
$\{\bm{f}_p^{(i_j)}: p\in\bar{\mathcal{R}}_{i_j}\}$
are able to recover $\{\bm{f}_p^{(i_j)}: p\in[n]\}$.
Furthermore, using the downloaded vector symbol sums in \eqref{t2-repair1-eq1} and the previously recovered data, node $i_j$ can extract and compute the independent vector symbols
$\{\bm{f}_{i_j}^{(p)}: p\in[d-k+h]\setminus\mathcal{H}\}$.
Thus node $i_j$ can compute and recover its $d-k+1$ erased data symbols stored at the $s$-th entry for $s\in\{i_j\}\cup([d-k+h]\setminus\mathcal{H})$, illustrated as follows:
 $$
 \{\bm{f}_{i_j}^{(i_j)}\}
 \cup
 \{\eta_{s,i_j}^{(1)}\bm{f}_{i_j}^{(s)}+\eta_{s,i_j}^{(2)}\bm{f}_{s}^{(i_j)}: s\in[d-k+h]\setminus\mathcal{H}, s<i_j\}
\cup\{\eta_{i_j,s}^{(3)}\bm{f}_{s}^{(i_j)}+\eta_{i_j,s}^{(4)}\bm{f}_{i_j}^{(s)}: s\in[d-k+h]\setminus\mathcal{H}, s>i_j\}.
$$

In the collaboration phase, for $j,j'\in[h]$ with $j\neq j'$, node $i_{j'}$ transmits the recovered data $\bm{f}_{i_j}^{(i_{j'})}$ to node $i_j$.
Then for each $j\in[h]$, node $i_j$ using the received data $\{\bm{f}_{i_j}^{(i_{j'})}: j'\in[h]\setminus\{j\}\}$ and the previously recovered data $\{\bm{f}_{i_{j'}}^{(i_j)}: j'\in[h]\setminus\{j\}\}$, can compute the remaining $h-1$ erased data symbols stored at the $s$-th entry for $s\in\mathcal{H}\setminus\{i_j\}$, illustrated as follows:
$$
 \{\eta_{s,i_j}^{(1)}\bm{f}_{i_j}^{(s)}+\eta_{s,i_j}^{(2)}\bm{f}_{s}^{(i_j)}: s\in\mathcal{H}, s<i_j\}
\cup\{\eta_{i_j,s}^{(3)}\bm{f}_{s}^{(i_j)}+\eta_{i_j,s}^{(4)}\bm{f}_{i_j}^{(s)}: s\in\mathcal{H}, s>i_j\}.
$$
Thus the $h$ nodes in $\mathcal{H}$ can be cooperatively repaired.
The total amount of data downloaded and communicated is $(d+h-1)h\ell_0=\frac{(d+h-1)h\ell}{d-k+h}$ symbols in $F$, and the amount of the accessed data at the helper nodes is $dh\ell_0=\frac{dh\ell}{d-k+h}$ symbols in $F$, both achieving the cut-set bound in \eqref{bw-bound}. Then it satisfies optimal-access property.
\end{proof}

\begin{corollary}\label{t2-repair1-cor}
In the repair of $\mathcal{H}$, let $S_{i,\mathcal{H}}$ be the corresponding repair matrix for repairing node $i\in\mathcal{H}$, let $\mathcal{R}_i\subseteq[n]\setminus\mathcal{H}$ with $|\mathcal{R}_i|=d$ satisfying $(\mathcal{F}\setminus\mathcal{H})\subseteq\mathcal{R}_i$ be the set of helper nodes connected by node $i$.
Denote by $(\bm{c}_1,\ldots,\bm{c}_{n})$ the codewords of $\mathcal{C}_{\mathcal{T}_2}$.
Then for each $i\in\mathcal{H}$, it has that
\begin{itemize}
\item[(1)] the downloaded data $\{S_{i,\mathcal{H}}\bm{c}_j^{\top}: j\in \mathcal{R}_i\}$ are able to recover $\{S_{i,\mathcal{H}}\bm{c}_j^{\top}: j\in[n]\setminus\mathcal{H}\}$.
\item[(2)] the downloaded data and the collaborated data at node $i$ can recover $\{S_{i,\mathcal{H}}\bm{c}_j^{\top}: j\in[n]\}$.
 \end{itemize}
\end{corollary}

\begin{proof}
Let $\mathcal{H}=\{i_1,i_2,\ldots,i_h\}\subseteq\mathcal{F}$ where $\mathcal{F}=[d-k+h]$. According to the proof of Theorem \ref{t2-repair1}, it is easy to see for each $j\in[h]$, the repair matrix of node $i_j$ is
\[
S_{i_j,\mathcal{H}}=\begin{pmatrix}
\bm{0}_{\ell/(d-k+h)} & \cdots & \bm{I}_{\ell/(d-k+h)}  &  \cdots    & \bm{0}_{\ell/(d-k+h)}
\end{pmatrix}\in F^{\frac{\ell}{d-k+h}\times\ell},
\]
where $\bm{0}_{\ell/(d-k+h)}$ and $\bm{I}_{\ell/(d-k+h)}$ represent all-zero matrix and identity matrix of order $\ell/(d-k+h)$.
And for each $j\in[h]$ and $p\in[n]$, it has that
\begin{equation}\label{t2-repair1-eq2}
S_{i_j,\mathcal{H}}\bm{c}_p^{\top}=
\begin{cases}
\bm{f}_{p}^{(i_j)\top}, &{\rm{for}}~p\in([n]\setminus[d-k+h])\cup\{i_j\},\\
\eta_{i_j,p}^{(1)}\bm{f}_p^{(i_j)\top}+\eta_{i_j,p}^{(2)}\bm{f}_{i_j}^{(p)\top}, &{\rm for}~p\in[d-k+h]~{\rm and}~p>i_j,\\
\eta_{p,i_j}^{(3)}\bm{f}_{i_j}^{(p)\top}+\eta_{p,i_j}^{(4)}\bm{f}_{p}^{(i_j)\top}, &{\rm for}~p\in[d-k+h]~{\rm and}~p<i_j.
\end{cases}
\end{equation}

We firstly prove (1). Note that $[n]\setminus\mathcal{H}=([n]\setminus\mathcal{F})\cup(\mathcal{F}\setminus\mathcal{H})$ and $(\mathcal{F}\setminus\mathcal{H})\subseteq\mathcal{R}_{i_j}$ for all $j\in[h]$, then it suffices to prove that $\{S_{i_j,\mathcal{H}}\bm{c}_p^{\top}: p\in \mathcal{R}_{i_j}\}$ are able to recover $\{S_{i_j,\mathcal{H}}\bm{c}_p^{\top}: p\in[n]\setminus\mathcal{F}\}=\{\bm{f}_{p}^{(i_j)\top}: p\in[n]\setminus[d-k+h]\}$.
This is proved in the download phase of the proof of Theorem \ref{t2-repair1}.

Then we prove (2).
According to (1), it remains for us to prove that for each $j\in[h]$, node $i_j$ using the downloaded data and collaborated data can recover $\{S_{i_j,\mathcal{H}}\bm{c}_p^{\top}: p\in\mathcal{H}\}$.
Recall that in the download phase of the proof of Theorem \ref{t2-repair1}, node $i_j$ using the downloaded data can recover $\{\bm{f}_p^{(i_j)}: p\in[n]\}$.
This implies that node $i_j$ can recover $S_{i_j,\mathcal{H}}\bm{c}_{i_j}^{\top}$ since $S_{i_j,\mathcal{H}}\bm{c}_{i_j}^{\top}=\bm{f}_{i_j}^{(i_j)\top}$ by \eqref{t2-repair1-eq2}.
Besides,
in the collaboration phase of the proof of Theorem \ref{t2-repair1}, for each $j\in[h]$, node $i_j$ receives the data $\{\bm{f}_{i_j}^{(i_{j'})}: j'\in[h]\setminus\{j\}\}$ from the remaining $h-1$ failed nodes.
Thus for $j'\in[h]\setminus\{j\}$, node $i_j$ using the recovered data $\bm{f}_{i_{j'}}^{(i_j)}$ in the download phase and received data $\bm{f}_{i_j}^{(i_{j'})}$ in the collaboration phase, can compute $S_{i_j,\mathcal{H}}\bm{c}_{i_{j'}}^{\top}$ according to \eqref{t2-repair1-eq2}.
This completes the proof.
\end{proof}

In the following, as in Subsection \ref{sub-t1-repair}, we summarize the properties in Theorem \ref{t2-repair1} and Corollary \ref{t2-repair1-cor} to give a definition for the case of repairing other $h$-failed node set $\bar{\mathcal{H}}$, which is useful in the following sections.

\begin{definition}\label{def-property2}
Let $\bar{\mathcal{C}}$ be an $(n,k)$ MDS scalar/array code, and let $\bar{\mathcal{F}}\subseteq[n]$ with $|\bar{\mathcal{F}}|=d-k+h$, and $\bar{\mathcal{H}}\subseteq\bar{\mathcal{F}}$ with $|\bar{\mathcal{H}}|=h$. Define (P3), (P4), and (P5) to be the following properties of $\bar{\mathcal{C}}$ if $\bar{\mathcal{C}}$ has. That is,
\begin{itemize}
\item[(P3)] $\bar{\mathcal{C}}$ enables optimal-access  cooperative repair of the $h$ nodes $\bar{\mathcal{H}}$ using  any $d$ helper nodes with local property.
\end{itemize}
If (P3) holds, in the repair of $\bar{\mathcal{H}}$, denote $S_{i,\bar{\mathcal{H}}}$ to be the corresponding repair matrix for repairing node $i\in\bar{\mathcal{H}}$, let $\mathcal{R}_i\subseteq[n]\setminus\bar{\mathcal{H}}$ with $|\mathcal{R}_i|=d$ satisfying $(\bar{\mathcal{F}}\setminus\bar{\mathcal{H}})\subseteq\mathcal{R}_i$ be the set of helper nodes connected by node $i$.
Denote by $(\bm{c}_1,\ldots,\bm{c}_{n})$ the codewords of $\bar{\mathcal{C}}$.
 \begin{itemize}
  \item[(P4)] For $i\in\bar{\mathcal{H}}$, it has $\{S_{i,\bar{\mathcal{H}}}\bm{c}_j^{\top}: j\in \mathcal{R}_i\}$ are able to recover $\{S_{i,\bar{\mathcal{H}}}\bm{c}_j^{\top}: j\in[n]\setminus\bar{\mathcal{H}}\}$.
    \item[(P5)] For $i\in\bar{\mathcal{H}}$, the downloaded data and the collaborated data at node $i$ can recover $\{S_{i,\bar{\mathcal{H}}}\bm{c}_j^{\top}: j\in[n]\}$.
\end{itemize}
\end{definition}

Next, we characterize the repair property of $\mathcal{C}_{\mathcal{T}_2}$ in Theorem \ref{t2-repair2}.

\begin{theorem}\label{t2-repair2}
Given any $\mathcal{F}'\subseteq[n]\setminus\mathcal{F}$ with $|\mathcal{F}'|=d-k+h$, where $\mathcal{F}$ is the goal node set, and given any $h$ nodes $\mathcal{H}'\subseteq\mathcal{F}'$ with $|\mathcal{H}'|=h$, if $\mathcal{C}_0$ satisfies the two properties (P3) and (P4) in Definition \ref{def-property2} for repairing $\mathcal{H}'$, then $\mathcal{C}_{\mathcal{T}_2}$ retains the two properties (P3) and (P4) for repairing $\mathcal{H}'$.
\end{theorem}

\begin{proof}
Suppose $\mathcal{C}_0$ satisfies the two conditions (P3) and (P4) in Definition \ref{def-property2} for repairing $\mathcal{H}'$. Based on the repair property of $\mathcal{H}'$ and the corresponding repair matrices $S_{i,\mathcal{H}'}, i\in\mathcal{H}'$ in $\mathcal{C}_0$, we firstly define the repair matrices $\tilde{S}_{i,\mathcal{H}'}, i\in\mathcal{H}'$ in $\mathcal{C}_{\mathcal{T}_2}$ as follows.
For $i\in\mathcal{H}'$, define
$\tilde{S}_{i,\mathcal{H}'}=I_{d-k+h}\otimes S_{i,\mathcal{H}'}$, where $I_{d-k+h}$ is the identity matrix with order $d-k+h$ and ``$\otimes$'' represents the tensor product.
Then $\tilde{S}_{i,\mathcal{H}'}\in F^{\frac{\ell}{d-k+h}\times\ell}$ where $\ell=(d-k+h)\ell_0$.

Now we prove $\mathcal{C}_{\mathcal{T}_2}$ satisfies property (P3) for repairing $\mathcal{H}'$.
Recall the transformation $\mathcal{T}_2$ and Table \ref{tab3}, for $j\in[n]\setminus[d-k+h]$, node $j$ stores $\bm{\tilde{c}}_j=(\bm{f}_j^{(1)},\bm{f}_j^{(2)},\ldots,\bm{f}_j^{(d-k+h)})\in F^{\ell}$.
Suppose $h$ nodes $\mathcal{H}'$ in $\mathcal{C}_{\mathcal{T}_2}$ are failed, i.e., $\bm{\tilde{c}}_i=(\bm{f}_i^{(1)},\bm{f}_i^{(2)},\ldots,\bm{f}_i^{(d-k+h)})$, $i\in\mathcal{H}'$ are erased, where we note $\mathcal{H}'\subseteq\mathcal{F}'\subseteq[n]\setminus[d-k+h]$.
For $i\in\mathcal{H}'$, let $\mathcal{R}_i\subseteq[n]\setminus\mathcal{H}'$ with $|\mathcal{R}_i|=d$ and $(\mathcal{F}'\setminus\mathcal{H}')\subseteq\mathcal{R}_i$ be the set of $d$ helper nodes connected by node $i$.

In the download phase, for $i\in\mathcal{H}'$, node $i$ downloads $\tilde{S}_{i,\mathcal{H}'}\bm{\tilde{c}}_j^{\top}$ from each helper node $j\in\mathcal{R}_i$, where for $j\in\mathcal{R}_i\cap[d-k+h]$,
\begin{equation}\label{t2-repair-eq2}
\begin{aligned}
\tilde{S}_{i,\mathcal{H}'}\bm{\tilde{c}}_j^{\top}
=\begin{pmatrix}
S_{i,\mathcal{H}'}(\eta_{1,j}^{(1)}\bm{f}_j^{(1)}+ \eta_{1,j}^{(2)}\bm{f}_1^{(j)})^{\top} \\
\vdots \\
S_{i,\mathcal{H}'}\bm{f}_j^{(j)\top} \\
\vdots \\
S_{i,\mathcal{H}'}(\eta_{j,d-k+h}^{(3)}\bm{f}_{d-k+h}^{(j)}+\eta_{j,d-k+h}^{(4)}\bm{f}_{j}^{(d-k+h)})^{\top}
\end{pmatrix}
=\begin{pmatrix}
\eta_{1,j}^{(1)}S_{i,\mathcal{H}'}\bm{f}_j^{(1)\top}+ \eta_{1,j}^{(2)}S_{i,\mathcal{H}'}\bm{f}_1^{(j)\top} \\
\vdots \\
S_{i,\mathcal{H}'}\bm{f}_j^{(j)\top} \\
\vdots \\
\eta_{j,d-k+h}^{(3)}S_{i,\mathcal{H}'}\bm{f}_{d-k+h}^{(j)\top}+\eta_{j,d-k+h}^{(4)}S_{i,\mathcal{H}'}\bm{f}_{j}^{(d-k+h)\top}
\end{pmatrix},
\end{aligned}
\end{equation}
and for $j\in\mathcal{R}_i\cap[d-k+h+1,n]$,
\begin{equation}\label{t2-repair-eq3}
\tilde{S}_{i,\mathcal{H}'}\bm{\tilde{c}}_j^{\top}
=\begin{pmatrix}
S_{i,\mathcal{H}'}\bm{f}_j^{(1)\top} \\
S_{i,\mathcal{H}'}\bm{f}_j^{(2)\top} \\
\vdots \\
S_{i,\mathcal{H}'}\bm{f}_j^{(d-k+h)\top}
\end{pmatrix},
\end{equation}
according to Table \ref{tab3}.
In the following, we firstly illustrate that for each $i\in\mathcal{H}'$, node $i$ using the downloaded data in \eqref{t2-repair-eq2} and \eqref{t2-repair-eq3} can recover the data
\begin{equation}\label{t2-repair-eq3-1}
\{S_{i,\mathcal{H}'}\bm{f}_j^{(s)\top}: j\in[n]\setminus\mathcal{H}', s\in[d-k+h]\}.
\end{equation}

To this end, firstly consider any $j,j'\in\mathcal{R}_i\cap[d-k+h]$ with $j<j'$, then node $i$ using the repair data from the $j$-th symbol in node $j'$, i.e. $S_{i,\mathcal{H}'}(\eta_{j,j'}^{(1)}\bm{f}_{j'}^{(j)}+ \eta_{j,j'}^{(2)}\bm{f}_j^{(j')})^{\top}$, and the repair data from the $j'$-th symbol in node $j$, i.e. $S_{i,\mathcal{H}'}(\eta_{j,j'}^{(3)}\bm{f}_{j'}^{(j)}+ \eta_{j,j'}^{(4)}\bm{f}_j^{(j')})^{\top}$,  can compute $S_{i,\mathcal{H}'}\bm{f}_{j'}^{(j)\top}$ and $S_{i,\mathcal{H}'}\bm{f}_{j}^{(j')\top}$ since $\eta_{j,j'}^{(1)}\eta_{j,j'}^{(4)}\neq \eta_{j,j'}^{(2)}\eta_{j,j'}^{(3)}$.
That is, $\{S_{i,\mathcal{H}'}\bm{f}_{j}^{(j')\top}: j,j'\in\mathcal{R}_i\cap[d-k+h], ~j\neq j'\}$ can be recovered.
Also, for $j\in\mathcal{R}_i\cap[d-k+h]$, $S_{i,\mathcal{H}'}\bm{f}_{j}^{(j)\top}$ is downloaded data. Thus the data $\{S_{i,\mathcal{H}'}\bm{f}_{j}^{(j')\top}: j,j'\in\mathcal{R}_i\cap[d-k+h]\}$ can be recovered.

Besides, for each $s\in\mathcal{R}_i\cap[d-k+h]$, $(\bm{f}_1^{(s)},\ldots,\bm{f}_n^{(s)}$) forms a codeword of the $(n,k,\ell_0)$ MDS code $\mathcal{C}_0$.
Then for each $s\in\mathcal{R}_i\cap[d-k+h]$, by property (P4) of $\mathcal{C}_0$, the previously recovered data $\{S_{i,\mathcal{H}'}\bm{f}_{j}^{(s)\top}: j\in\mathcal{R}_i\cap[d-k+h]\}$ along with the downloaded data $\{S_{i,\mathcal{H}'}\bm{f}_{j}^{(s)\top}: j\in\mathcal{R}_i\cap[d-k+h+1,n]\}$, are able to recover $\{S_{i,\mathcal{H}'}\bm{f}_{j}^{(s)\top}: j\in[n]\setminus\mathcal{H}'\}$.
That is, the data $\{S_{i,\mathcal{H}'}\bm{f}_{j}^{(s)\top}: j\in[n]\setminus\mathcal{H}', s\in\mathcal{R}_i\cap[d-k+h]\}$ can be reconstructed.

Furthermore, for each $s\in[d-k+h]$, $s\notin\mathcal{R}_i$ and each $j\in\mathcal{R}_i\cap[d-k+h]$, note that the repair data from the $s$-th symbol of helper node $j$ is $S_{i,\mathcal{H}'}(\eta_{s,j}^{(1)}\bm{f}_j^{(s)}+ \eta_{s,j}^{(2)}\bm{f}_s^{(j)})^{\top}$ if $s<j$ and $S_{i,\mathcal{H}'}(\eta_{j,s}^{(3)}\bm{f}_s^{(j)}+ \eta_{j,s}^{(4)}\bm{f}_j^{(s)})^{\top}$ if $s>j$, which is known.
And recall the data $S_{i,\mathcal{H}'}\bm{f}_s^{(j)\top}$ for $s\in[d-k+h]$, $s\notin\mathcal{R}_i$ and $j\in\mathcal{R}_i\cap[d-k+h]$ have been previously recovered, then one can solve out the data $S_{i,\mathcal{H}'}\bm{f}_j^{(s)\top}$. That is, the data $\{S_{i,\mathcal{H}'}\bm{f}_j^{(s)\top}: j\in\mathcal{R}_i\cap[d-k+h], s\in[d-k+h], s\notin\mathcal{R}_i\}$ can be recovered.

At last, note for each $s\in[d-k+h]$ and $s\notin\mathcal{R}_i$, $(\bm{f}_1^{(s)},\ldots,\bm{f}_n^{(s)}$) forms a codeword of $\mathcal{C}_0$. Then according to property (P4) of $\mathcal{C}_0$, for each $s\in[d-k+h]$ and $s\notin\mathcal{R}_i$, the recovered data $\{S_{i,\mathcal{H}'}\bm{f}_j^{(s)\top}: j\in\mathcal{R}_i\cap[d-k+h]\}$ along with the downloaded data $\{S_{i,\mathcal{H}'}\bm{f}_j^{(s)\top}: j\in\mathcal{R}_i\cap[d-k+h+1,n]\}$ are able to recover $\{S_{i,\mathcal{H}'}\bm{f}_j^{(s)\top}: j\in[n]\setminus\mathcal{H}'\}$.
That is, the data $\{S_{i,\mathcal{H}'}\bm{f}_j^{(s)\top}: j\in[n]\setminus\mathcal{H}', s\in[d-k+h], s\notin\mathcal{R}_i\}$ can be reconstructed.

Therefore, for each $i\in\mathcal{H}'$, node $i$ using the downloaded data can recover all the data in \eqref{t2-repair-eq3-1}.

Moreover, Note that $(\bm{f}_1^{(s)},\ldots,\bm{f}_n^{(s)}$), $s\in[d-k+h]$, form $d-k+h$ codewords of the $(n,k,\ell_0)$ MDS code $\mathcal{C}_0$.
According to property (P3) of $\mathcal{C}_0$, for each $s\in[d-k+h]$, the failed data $\bm{f}_i^{(s)}$, $i\in\mathcal{H}'$ can be cooperatively recovered through a cooperative repair process of $\mathcal{H}'$ in $\mathcal{C}_0$ with repair matrices $S_{i,\mathcal{H}'}$ and helper node set $\mathcal{R}_i$, $i\in\mathcal{H}'$. Thus the $h$ failed nodes can be cooperatively repaired.
Since $\tilde{S}_{i,\mathcal{H}'}=I_{d-k+h}\otimes S_{i,\mathcal{H}'}$ for $i\in\mathcal{H}'$, then $\mathcal{C}_{\mathcal{T}_2}$ retains the optimal-access property of $\mathcal{C}_0$.

In the following, we prove $\mathcal{C}_{\mathcal{T}_2}$ satisfies (P4) for repairing $\mathcal{H}'$.
Note that for $i\in\mathcal{H}'$, $\tilde{S}_{i,\mathcal{H}'}\bm{\tilde{c}}_j^{\top}$ has the form in \eqref{t2-repair-eq2} for $j\in[d-k+h]$, and has the form in \eqref{t2-repair-eq3} for $j\in[d-k+h+1,n]$.
According to the proof of (P3) in $\mathcal{C}_{\mathcal{T}_2}$, we know that for each $i\in\mathcal{H}'$, node $i$ using its downloaded data can recover the data in \eqref{t2-repair-eq3-1}.
Since $\mathcal{H}'\subseteq[n]\setminus[d-k+h]$, then $[n]\setminus\mathcal{H}'=[d-k+h]\cup([d-k+h+1,n]\setminus\mathcal{H}')$.
That is, the recovered data in \eqref{t2-repair-eq3-1}:
\begin{equation}\label{t2-repair-eq4}
\begin{aligned}
&\{S_{i,\mathcal{H}'}\bm{f}_j^{(s)\top}: j\in[n]\setminus\mathcal{H}', s\in[d-k+h]\} \\
=&\{S_{i,\mathcal{H}'}\bm{f}_j^{(s)\top}: j, s\in[d-k+h]\}
\cup\{S_{i,\mathcal{H}'}\bm{f}_j^{(s)\top}: j\in[d-k+h+1,n]\setminus\mathcal{H}', s\in[d-k+h]\}\\
=&\{S_{i,\mathcal{H}'}\bm{f}_j^{(s)\top}: j, s\in[d-k+h]\}
\cup\{\tilde{S}_{i,\mathcal{H}'}\bm{\tilde{c}}_j^{\top}: j\in[d-k+h+1,n]\setminus\mathcal{H}'\}.
\end{aligned}
\end{equation}
Note that $\tilde{S}_{i,\mathcal{H}'}\bm{\tilde{c}}_j^{\top}$, $j\in[d-k+h]$ can be computed from $\{S_{i,\mathcal{H}'}\bm{f}_j^{(s)\top}: j, s\in[d-k+h]\}$ as illustrated in \eqref{t2-repair-eq4}, according to \eqref{t2-repair-eq2}.
Thus, for each $i\in\mathcal{H}'$, node $i$ using the downloaded data $\{\tilde{S}_{i,\mathcal{H}'}\bm{\tilde{c}}_j^{\top}: j\in\mathcal{R}_i\}$ can recover $\{\tilde{S}_{i,\mathcal{H}'}\bm{\tilde{c}}_j^{\top}: j\in[n]\setminus\mathcal{H}'\}$.

This completes the proof.
\end{proof}

\begin{corollary}\label{t2-repair2-cor}
Given any $\mathcal{F}'\subseteq[n]\setminus\mathcal{F}$ with $|\mathcal{F}'|=d-k+h$, where $\mathcal{F}$ is the goal node set, and given any $h$ nodes $\mathcal{H}'\subseteq\mathcal{F}'$ with $|\mathcal{H}'|=h$, if $\mathcal{C}_0$ satisfies the two properties (P3) and (P5) in Definition \ref{def-property2} for repairing $\mathcal{H}'$, then $\mathcal{C}_{\mathcal{T}_2}$ retains the two properties (P3) and (P5) for repairing $\mathcal{H}'$.
\end{corollary}

\begin{proof}
The Corollary can be similarly proved as in Theorem \ref{t2-repair2}, and we omit it here.
\end{proof}

\section{A generic transformation from scalar MDS codes to optimal-access cooperative MSR codes}\label{generic-transform}

In this section, we present a generic transformation $\mathcal{T}$ that converts an $[n+d-k,d]$ MDS scalar code to an $(n,k,\ell=\delta^m)$ optimal-access cooperative MSR code by using the basic transformations $\mathcal{T}_1$ and $\mathcal{T} _2$ in Section \ref{transform-1} and Section  \ref{transform-2} for several times, where $\delta=d-k+h$ and $m=\binom{n}{h}-\lfloor\frac{n}{\delta}\rfloor(\binom{\delta}{h}-1)$.

\subsection{The generic transformation $\mathcal{T}$}\label{generic-t}
We recall and reuse the notations and definitions in Subsection \ref{notation-h}. Besides, recall the set $\mathcal{P}_0$ in Subsection \ref{notation-h}, we index the $h$-tuples in $\mathcal{P}_0$ as $\mathcal{E}_1,\mathcal{E}_2,\ldots,\mathcal{E}_{m-\lfloor\frac{n}{\delta}\rfloor}$ for simplicity.

Next, we give the generic transformation $\mathcal{T}$.
Let $F$ be a finite field with $|F|\geq n+d-k$, and for $1\leq i<j\leq d-k+h$, $a\in[4]$, choose $\eta_{i,j}^{(a)}$ to be nonzero elements in $F$ satisfying that $\eta_{i,j}^{(1)}\eta_{i,j}^{(4)}\neq \eta_{i,j}^{(2)}\eta_{i,j}^{(3)}$.
Let $\mathcal{C}_0$ be an $(n+d-k, d, \ell_0=1)$ MDS scalar code over $F$.
The transformation $\mathcal{T}$ begins with the scalar code $\mathcal{C}_0$ of length $n+d-k$ and iteratively applies the two basic transformations $\mathcal{T}_1$ and $\mathcal{T} _2$ for repairing different erasure patterns of nodes in $[n]$, and at last remove the last $d-k$ nodes through shortening technique.
Here we note that the basic transformations $\mathcal{T}_1$ and $\mathcal{T} _2$ both keep the code length unchanged.
Next we give the transformation $\mathcal{T}$ in Algorithm \ref{alg1}.

\begin{algorithm}[H]\label{alg1}
\caption{The generic transformation $\mathcal{T}$ for building optimal-access cooperative MSR codes.}\label{alg1}
\begin{algorithmic}[1]
\REQUIRE The followings are required to input:

\begin{itemize}
\item[(1)] Parameters $n,k,d,h$ with $2\leq h\leq n-k$, $k\leq d\leq n-h$. Denote $\delta=d-k+h$ and  $m=\binom{n}{h}-\lfloor\frac{n}{\delta}\rfloor(\binom{\delta}{h}-1)$.
\item[(2)] Nonzero elements $\eta_{i,j}^{(a)}\in F$ for $1\leq i<j\leq d-k+h$ and $a\in[4]$ with $\eta_{i,j}^{(1)}\eta_{i,j}^{(4)}\neq \eta_{i,j}^{(2)}\eta_{i,j}^{(3)}$.
\item[(3)] An $(n+d-k, d, \ell_0=1)$ MDS scalar code $\mathcal{C}_0$ over $F$.
\end{itemize}

\ENSURE An $(n,k,\ell=\delta^m)$ optimal-access cooperative MSR code $\mathcal{C}$ with $h$ erasures and $d$ helper nodes.
\STATE
Set $\mathcal{C}_0$ to be the base code, and initialize two goal node sets $\mathcal{H}=\emptyset, \mathcal{F}=\emptyset$.
\FOR{$t=0; t<\lfloor\frac{n}{\delta}\rfloor; t++$}
\STATE
Set the goal node set $\mathcal{F}=[t\delta+1, (t+1)\delta]$;
\STATE
Set the base code to be the $(n+d-k, d,\ell_t)$ code $\mathcal{C}_t$;
\STATE
Apply the basic transformation $\mathcal{T}_2$ to the base code $\mathcal{C}_t$ w.r.t. the goal node set $\mathcal{F}$, to obtain a new MDS array code $\mathcal{C}_{t+1}$ with parameters $(n+d-k, d, \ell_{t+1}=\delta^{t+1})$.
Then $\mathcal{C}_{t+1}$ can cooperatively repair any $h$ nodes in $\mathcal{F}$ using any $2d-k$ helper nodes with local property since $(2d-k)-d+h=d-k+h=\delta$.
\ENDFOR
\STATE
Set the base code to be the $(n+d-k, d, \ell_{\lfloor\frac{n}{\delta}\rfloor}=\delta^{\lfloor\frac{n}{\delta}\rfloor})$ MDS array code $\mathcal{C}_{\lfloor\frac{n}{\delta}\rfloor}$.
\FOR{$t=\lfloor\frac{n}{\delta}\rfloor; t<m; t++$}
\STATE
Set the goal node set $\mathcal{H}=\mathcal{E}_{t-\lfloor\frac{n}{\delta}\rfloor+1}$, where $\mathcal{E}_i$'s are defined at the beginning of Subsection \ref{generic-t};
\STATE
Set the base code to be the $(n+d-k, d, \ell_t)$ code $\mathcal{C}_{t}$;
\STATE
Apply the basic transformation $\mathcal{T}_1$ to the base code $\mathcal{C}_t$ w.r.t. the goal node set $\mathcal{H}$, to obtain a new MDS array code $\mathcal{C}_{t+1}$ with parameters $(n+d-k, d, \ell_{t+1}=\delta^{t+1})$. Then $\mathcal{C}_{t+1}$ can cooperatively repair the $h$ nodes $\mathcal{H}\subseteq[n]$ using any $2d-k$ constrained helper nodes.
\ENDFOR
\STATE
Shorten the obtained code $\mathcal{C}_m$ at the last $d-k$ nodes to obtain the final code $\mathcal{C}$, as illustrated in Lemma \ref{shorten-code}.
\RETURN
The $(n,k,\ell=\delta^m)$ code $\mathcal{C}$.
\end{algorithmic}
\end{algorithm}

\subsection{MDS property and optimal-access property}\label{generic-repair}

According to the basic transformations $\mathcal{T}_1$, $\mathcal{T}_2$ and Lemma \ref{shorten-code}, it is straightforward that the final code $\mathcal{C}$ derived from Algorithm \ref{alg1} is an $(n,k,\ell=\delta^m)$ MDS array code.
To illustrate the optimal-access property of $\mathcal{C}$, we look into Algorithm \ref{alg1} step by step and explain how to maintain the optimal-access property of nodes in both current goal node set and previous goal node sets simultaneously after each basic transformation.

We begin with the base code $\mathcal{C}_0$ and consider line $2$-$6$ of Algorithm \ref{alg1}.
Denote $n'=n+d-k$, $k'=d$ and $d'=2d-k$, then $\mathcal{C}_0$ is an $(n',k',\ell_0=1)$ MDS scalar code and it satisfies $d'-k'+h=\delta$.
According to Theorem \ref{t2-repair1} and Corollary \ref{t2-repair1-cor}, after applying transformation $\mathcal{T}_2$, the newly obtained code $\mathcal{C}_{t+1}$ satisfies (P3) and (P4) in Definition \ref{def-property2} for repairing any $h$ nodes in the current goal node set $[t\delta+1, (t+1)\delta]$ by using any $d'=2d-k$ helper nodes with local property.
Theorem \ref{t2-repair2} further ensures the successful repair of $h$ nodes in all previous sets $[t'\delta+1, (t'+1)\delta]$, $0\leq t'<t$.
Thus, the code $\mathcal{C}_{\lfloor\frac{n}{\delta}\rfloor}$ can cooperatively repair any $h$ nodes in every node set $[t\delta+1, (t+1)\delta]$ for $t\in[0,\lfloor\frac{n}{\delta}\rfloor-1]$, i.e., all the erasure patterns in $\mathcal{P}_i$, $i\in[\lfloor\frac{n}{\delta}\rfloor]$.

Now we begin with the MDS array code $\mathcal{C}_{\lfloor\frac{n}{\delta}\rfloor}$ and consider line $8$-$12$ of Algorithm \ref{alg1}.
By Theorem \ref{t1-repair1} and Corollary \ref{t1-repair1-cor}, each time we apply  transformation $\mathcal{T}_1$, the obtained code $\mathcal{C}_{t+1}$ satisfies (P1) and (P2) in Definition \ref{def-property1} for repairing the current $h$ goal nodes $\mathcal{H}=\mathcal{E}_{t-\lfloor\frac{n}{\delta}\rfloor+1}\in\mathcal{P}_0$ by using any $2d-k$ constrained helper nodes.
Theorem \ref{t1-repair2} further ensures the repair property of all the previous erasure patterns $\mathcal{E}_{t'-\lfloor\frac{n}{\delta}\rfloor+1}$ for $t'\in[\lfloor\frac{n}{\delta}\rfloor, t-1]$.
Hence, the derived code $\mathcal{C}_{m}$ achieves optimal-access cooperative repair of all the erasure patterns in $\mathcal{P}_0$ using any $2d-k$ constrained helper nodes.

Next, we illustrate that $\mathcal{C}_{m}$ retains the repair property of all the erasure patterns in $\mathcal{P}_i$, $i\in[\lfloor\frac{n}{\delta}\rfloor]$.
We consider the code $\mathcal{C}_{\lfloor\frac{n}{\delta}\rfloor}$ and line $8$-$12$ of Algorithm \ref{alg1}.
Note that $\mathcal{C}_{\lfloor\frac{n}{\delta}\rfloor}$ can repair all erasure patterns in $\mathcal{P}_i$, $i\in[\lfloor\frac{n}{\delta}\rfloor]$ by using any $2d-k$ helper nodes with local property.
Recall the transformation $\mathcal{T}_1$ and according to Theorem \ref{t1-repair2}, the ideal scenario is that $\mathcal{C}_{\lfloor\frac{n}{\delta}\rfloor}$ satisfies property (P1) and (P2) in Definition \ref{def-property1} for repairing the erasure patterns in $\mathcal{P}_i$, $i\in[\lfloor\frac{n}{\delta}\rfloor]$.
However, in (P1) and (P2), it requires the helper nodes to be constrained helper nodes. Besides, in (P2), the downloaded data $\{S_{i,\mathcal{H}'}\bm{c}_j^{\top}: j\in\mathcal{R}_i\}$ should be able to recover $\{S_{i,\mathcal{H}'}\bm{c}_j^{\top}: j\in[n+d-k]\}$.
While we recall Theorem \ref{t2-repair1} and Theorem \ref{t2-repair2}, in $\mathcal{C}_{\lfloor\frac{n}{\delta}\rfloor}$, the downloaded data for repairing an erasure pattern $\mathcal{H}\in\mathcal{P}_i$, $i\in[\lfloor\frac{n}{\delta}\rfloor]$ can only recover part of the data, i.e., $\{S_{i,\mathcal{H}'}\bm{c}_j^{\top}: j\in[n+d-k]\setminus\mathcal{H}\}$.
Fortunately, we observe that $\{S_{i,\mathcal{H}}\bm{c}_j^{\top}: j\in[n+d-k]\}$ can be fully recovered using both downloaded data and collaborated data during node repair, as illustrated in Corollary \ref{t2-repair1-cor} and Corollary \ref{t2-repair2-cor}.
For a better characterization, we summarize these points to define two properties (P1$'$) and (P2$'$) in Definition \ref{def-property3} as follows.

\begin{definition}\label{def-property3}
Let $\bar{\mathcal{C}}$ be an $(n+d-k,d)$ MDS scalar/array code, and let $\bar{\mathcal{H}}\subseteq\bar{\mathcal{F}}\subseteq[n]$ with $|\bar{\mathcal{H}}|=h$ and $|\bar{\mathcal{F}}|=d-k+h$. Define (P1$'$) and (P2$'$) to be the two properties of $\bar{\mathcal{C}}$ if $\bar{\mathcal{C}}$ has. That is,
\begin{itemize}
\item[(P1$'$)] $\bar{\mathcal{C}}$ enables optimal-access  cooperative repair of the $h$ nodes $\bar{\mathcal{H}}$ using  any $2d-k$ constrained helper nodes with local property.
\item[(P2$'$)]
In the repair of $\bar{\mathcal{H}}$, denote $S_{i,\bar{\mathcal{H}}}$ to be the corresponding repair matrix for repairing node $i\in\bar{\mathcal{H}}$, let $\mathcal{R}_i\subseteq[n]\setminus\bar{\mathcal{H}}$ with $|\mathcal{R}_i|=2d-k$ satisfying $[n+1,n+d-k]\subseteq\mathcal{R}_i$ and $(\bar{\mathcal{F}}\setminus\bar{\mathcal{H}})\subseteq\mathcal{R}_i$ be the set of helper nodes connected by node $i$.
Then for $i\in\bar{\mathcal{H}}$, the downloaded data
 $\{S_{i,\bar{\mathcal{H}}}\bm{c}_j^{\top}: j\in \mathcal{R}_i\}$ and collaborated data at node $i$ are able to recover $\{S_{i,\bar{\mathcal{H}}}\bm{c}_j^{\top}: j\in[n+d-k]\}$, where $(\bm{c}_1,\ldots,\bm{c}_{n+d-k})$ represents a codeword of $\bar{\mathcal{C}}$.
\end{itemize}
\end{definition}

According to Definition \ref{def-property3}, Theorem \ref{t2-repair1}, Corollary \ref{t2-repair1-cor} and Corollary \ref{t2-repair2-cor}, it is straightforward that $\mathcal{C}_{\lfloor\frac{n}{\delta}\rfloor}$ satisfies (P1$'$) and (P2$'$) for repairing any erasure pattern in $\mathcal{P}_i$, $i\in[\lfloor\frac{n}{\delta}\rfloor]$.
Then, in order to show that $\mathcal{C}_{m}$ can repair all the erasure patterns in $\mathcal{P}_i$, $i\in[\lfloor\frac{n}{\delta}\rfloor]$, it suffices to prove the two properties (P1$'$) and (P2$'$) can be retained after applying transformation $\mathcal{T}_1$.
This is illustrated in the following Theorem.
\begin{theorem}\label{t2-t1-thm}
For $t\in[\lfloor\frac{n}{\delta}\rfloor,m-1]$, let $\mathcal{C}_{t+1}$ be obtained from $\mathcal{C}_{t}$ by applying transformation $\mathcal{T}_1$ w.r.t. the goal node set $\mathcal{H}=\mathcal{E}_{t-\lfloor\frac{n}{\delta}\rfloor+1}\in\mathcal{P}_0$ as illustrated in Algorithm \ref{alg1}.
For any $\mathcal{H'}\subseteq\mathcal{F}=[z\delta+1, (z+1)\delta]$ with $|\mathcal{H'}|=h$ and $z\in[0,\lfloor\frac{n}{\delta}\rfloor-1]$, if $\mathcal{C}_{t}$ satisfies the two properties (P1$'$) and (P2$'$) in Definition \ref{def-property3} for repairing $\mathcal{H'}$, then $\mathcal{C}_{t+1}$ retains the two properties (P1$'$) and (P2$'$) for repairing $\mathcal{H'}$.
\end{theorem}

\begin{proof}
Suppose $\mathcal{C}_t$ satisfies the two properties (P1$'$) and (P2$'$) for repairing $\mathcal{H'}$. Based on the repair property of $\mathcal{H}'$ and the corresponding repair matrices $S_{i,\mathcal{H}'}, i\in\mathcal{H}'$ in $\mathcal{C}_t$, define the repair matrices $\tilde{S}_{i,\mathcal{H}'}, i\in\mathcal{H}'$ in $\mathcal{C}_{t+1}$ to be $\tilde{S}_{i,\mathcal{H}'}=I_{d-k+h}\otimes S_{i,\mathcal{H}'}$, where $I_{d-k+h}$ is the identity matrix with order $d-k+h$ and "$\otimes$'' represents the tensor product.
Then $\tilde{S}_{i,\mathcal{H}'}\in F^{\frac{\ell}{d-k+h}\times\ell}$ where $\ell=(d-k+h)\ell_0$.

Now we prove $\mathcal{C}_{t+1}$ retains the property (P1$'$) for repairing $\mathcal{H'}$.
Suppose $h$ nodes $\mathcal{H}'$ in $\mathcal{C}_{t+1}$ are failed, i.e., $\bm{\tilde{c}}_i=(\bm{f}_i^{(1)},\bm{f}_i^{(2)},\ldots,\bm{f}_i^{(d-k+h)})$, $i\in\mathcal{H}'$ are erased by transformation $\mathcal{T}_1$ and Table \ref{tab2}, where we substitute $\mathcal{C}_{0}$ with $\mathcal{C}_{t}$ and $\mathcal{C}_{\mathcal{T}_1}$ with $\mathcal{C}_{t+1}$.
For each $i\in\mathcal{H}'$, let $\mathcal{R}_i\subseteq[n+d-k]\setminus\mathcal{H}'$ with $|\mathcal{R}_i|=2d-k$ satisfying $[n+1,n+d-k]\subseteq\mathcal{R}_i$ and $(\mathcal{F}\setminus\mathcal{H'})\subseteq\mathcal{R}_i$ be the set of $2d-k$ helper nodes connected by node $i$.

In the download phase, for $i\in\mathcal{H}'$, node $i$ downloads $\tilde{S}_{i,\mathcal{H}'}\bm{\tilde{c}}_j^{\top}$ as illustrated in \eqref{t1-repair-eq1} from each helper node $j\in\mathcal{R}_i$.
That is, node $i$ downloads the data $\{S_{i,\mathcal{H}'}\bm{f}_j^{(s)\top}: j\in\mathcal{R}_i, s\in[d-k+h]\}$, for $i\in\mathcal{H}'$.
In the following, the node repair is done through two rounds of cooperation.
In the first round, consider each $s\in[h+1,d-k+h]$, it has $(\bm{f}_1^{(s)},\ldots,\bm{f}_{n+d-k}^{(s)})$ forms a codeword of $\mathcal{C}_t$.
Then for each $s\in[h+1,d-k+h]$, by (P1$'$) and using the cooperative repair process of $\mathcal{H}'$ in $\mathcal{C}_t$, node $i\in\mathcal{H}'$ can recover the erased data $\bm{f}_{i}^{(s)}$.
Moreover, by (P2$'$), node $i$ using the downloaded data and collaborated data can recover all the data
\begin{equation}\label{t2-t1-thm-eq1}
\{S_{i,\mathcal{H}'}\bm{f}_j^{(s)\top}: j\in[n+d-k], s\in[h+1,d-k+h]\}.
\end{equation}

In the second round, consider each $s\in[h]$, it has $(\bm{f}_1^{(s)},\ldots,\bm{f}_n^{(s)},\bm{f}_{n+1}^{(s)}+\bm{f}_{s}^{(h+1)},\ldots, \bm{f}_{n+d-k}^{(s)}+\bm{f}_{s}^{(d-k+h)})$ forms a codeword of $\mathcal{C}_0$.
Since $S_{i,\mathcal{H}'}\bm{f}_s^{(h+1)\top},\ldots,S_{i,\mathcal{H}'}\bm{f}_s^{(d-k+h)\top}$ have been recovered, then for each $s\in[h]$, node $i$ using the downloaded data and recovered data, can further compute
\[
\{S_{i,\mathcal{H}'}\bm{f}_j^{(s)\top}: j\in\mathcal{R}_i\setminus[n+1,n+d-k]\}
\cup\{
S_{i,\mathcal{H}'}(\bm{f}_{n+1}^{(s)}+\bm{f}_{s}^{(h+1)})^{\top},\ldots, S_{i,\mathcal{H}'}(\bm{f}_{n+d-k}^{(s)}+\bm{f}_{s}^{(d-k+h)})^{\top}
\}.
\]
This is exactly the helper data of node $i$ in the $s$-th codeword of $\mathcal{C}_t$.
Then, by (P1$'$) and using the cooperative repair process of $\mathcal{H}'$ in $\mathcal{C}_t$, for each $s\in[h+1,d-k+h]$, node $i\in\mathcal{H}'\subseteq[n]$ can recover the erased data $\bm{f}_{i}^{(s)}$.

Therefore, the $h$ nodes are cooperatively repaired.
Moreover, by (P2$'$), node $i\in\mathcal{H}'$ using the downloaded data and collaborated data can recover
\begin{equation}\label{t2-t1-thm-eq2}
\{S_{i,\mathcal{H}'}\bm{f}_j^{(s)\top}: j\in[n+d-k], s\in[h]\}.
\end{equation}
Since $\tilde{S}_{i,\mathcal{H}'}=I_{d-k+h}\otimes S_{i,\mathcal{H}'}$ for $i\in\mathcal{H}'$, then $\mathcal{C}_{t+1}$ retains the optimal-access property of $\mathcal{C}_t$ for repairing $\mathcal{H}'$.

Next we prove $\mathcal{C}_{t+1}$ satisfies (P2$'$) for repairing $\mathcal{H'}$.
according to the proof of (P1$'$), it is straightforward that each node $i\in\mathcal{H}'$ using the downloaded data and collaborated data can recover $\{\tilde{S}_{i,\mathcal{H}'}\bm{\tilde{c}}_j^{\top}: j\in[n+d-k]\}=\{S_{i,\mathcal{H}'}\bm{f}_j^{(s)\top}: j\in[n+d-k], s\in[d-k+h]\}$, as illustrated in \eqref{t2-t1-thm-eq1} and \eqref{t2-t1-thm-eq2}.

This completes the proof.
\end{proof}

By Theorem \ref{t2-t1-thm} and the fact that $\mathcal{C}_{\lfloor\frac{n}{\delta}\rfloor}$ satisfies (P1$'$) and (P2$'$) for repairing any erasure pattern in $\mathcal{P}_i$, $i\in[\lfloor\frac{n}{\delta}\rfloor]$, after executing line $8$-$12$ of Algorithm \ref{alg1}, the obtained code $\mathcal{C}_{m}$ can repair all the erasure patterns in $\mathcal{P}_i$, $i\in[\lfloor\frac{n}{\delta}\rfloor]$ using any $2d-k$ constrained helper nodes with local property.
At last, through a shortening technique as illustrated in Lemma \ref{shorten-code}, one can obtain the final code $\mathcal{C}$ for cooperatively repairing any $h$ node erasures using $d$ helper nodes.
The repair property of $\mathcal{C}$ is given in Corollary \ref{repair-t}.

\begin{corollary}\label{repair-t}
The returned code $\mathcal{C}$ in Algorithm \ref{alg1} enables optimal-access cooperative repair of any $h$ erasures using $d$ helper nodes. More precisely, for any $h$ nodes $\mathcal{H}\subseteq[n]$ with $|\mathcal{H}|=h$, it has
\begin{itemize}
\item if $\mathcal{H}\in\mathcal{P}_i$ for some $i\in[\lfloor\frac{n}{\delta}\rfloor]$, then the $h$ nodes $\mathcal{H}$ can be cooperatively repaired with optimal access using any $d$ helper nodes with local property.

\item if $\mathcal{H}\in\mathcal{P}_0$, then the $h$ nodes $\mathcal{H}$ can be cooperatively repaired with optimal access using any $d$ helper nodes.
\end{itemize}
\end{corollary}

\section{Discussion and Conclusion}\label{conclusion}
\subsection{Discussion on the connections with previous works}\label{comparison}
In this subsection, we give some discussion about the connections of our codes  and previous works given by Zhang, Zhang \& Wang \cite{Zhang2020}, Ye \& Barg \cite{Ye2016sub-}, and Li, Tang \& Tian TIT'2018 \cite{JLi-2018}.
We show that the code structures of \cite{Zhang2020,Ye2016sub-,JLi-2018} are included as special cases of our transformation construction.
We also establish a connection between our transformation method and parity-check matrix construction.
We show that the parity-check matrix code can be derived from the transformation method by using an initial Reed-Solomon-type MDS scalar code.

\begin{itemize}
\item[(1)] {\it Connection with the MSR code by Li, Tang \& Tian TIT'2018 \cite{JLi-2018}:}\\
Recall the MSR code construction by Li, Tang \& Tian TIT'2018 \cite{JLi-2018}, the authors give a transformation for building optimal-access MSR codes with $h=1$ and $d=n-1$.
Actually, the transformation $\mathcal{T}_2$ in Section \ref{transform-2} is a generalization of that in \cite{JLi-2018} for any $1\leq h\leq n-k$ and $k\leq d\leq n-h$.
In the terminology of codes with helpers-independent repair matrices, the  transformation $\mathcal{T}_2$ degenerates to the transformation in \cite{JLi-2018} by setting $h=1$, $d=n-1$ and parameters $\eta_{i,j}^{(1)}=\theta_{j-1,i-1}$, $\eta_{i,j}^{(2)}=\eta_{i-1,j-1}$, $\eta_{i,j}^{(3)}=\eta_{j-1,i-1}$ and $\eta_{i,j}^{(4)}=\theta_{i-1,j-1}$ for all $i,j\in[r]$, $i<j$, where the $\theta_{i,j}$'s, $\eta_{i,j}$'s are parameters in \cite{JLi-2018}.

\item[(2)] {\it Connection with the MSR code by Ye \& Barg \cite{Ye2016sub-}:}\\
Recall the MSR code given by Ye \& Barg \cite{Ye2016sub-}.
The MSR code is defined by its parity-check matrix, which is systematically stacked up from a basic parity-check matrix structure for $\lceil\frac{n}{d-k+1}\rceil$ dimensions.
The basic parity-check matrix structure defines an $(n,k,\ell=d-k+1)$ MDS array code, which can optimally repair any single node within $r$ specific goal nodes.
W.L.O.G., suppose the $r$ goal nodes are nodes $\{1,2,\ldots,r\}$.
We summarize its basic code structure as follows.
\begin{construction}[Basic code structure of \cite{Ye2016sub-}]\label{Ye-basic}
{\it
Let $F$ be a finite field of size $|F|\geq n$. Let $\lambda_i$, $i\in[n]$ be $n$ distinct elements in $F$, and let $\gamma\in F\setminus\{0,1\}$.
The basic code structure of \cite{Ye2016sub-} is an $(n,k,\ell=d-k+1)$ MDS array code, which is defined by the parity-check matrix
$H$ with form in (\ref{def1}), where for $t\in[r]$, $H_{t,j}=\lambda_{j}^{t-1}I_{\ell}$ for $j\in[d-k+2,n]$ and $(H_{t,1},\dots,H_{t,d-k+1})=$
{\setlength{\arraycolsep}{3pt}
\begin{equation}\label{Ye-block}
\resizebox{0.9\hsize}{!}{$
\begin{aligned}
&\left(\begin{array}{ccccc|ccccc|c|ccccc}
\lambda_1^{t-1} & \lambda_2^{t-1} & \lambda_3^{t-1} &\ldots & \lambda_{d-k+1}^{t-1} & \lambda_2^{t-1} &  &  && & &
\lambda_{d-k+1}^{t-1} &  & & &  \\
 &\gamma \lambda_1^{t-1} &  && & \lambda_1^{t-1} & \lambda_2^{t-1} & \lambda_3^{t-1} &\ldots & \lambda_{d-k+1}^{t-1} & &
 & \lambda_{d-k+1}^{t-1} & & &  \\
 & &  \gamma \lambda_1^{t-1} && &  &  & \gamma \lambda_2^{t-1} &  & & \cdots
 & & & \lambda_{d-k+1}^{t-1} & &  \\
 & &  & \ddots & &  &  &  & \ddots &  &
 & & & & \ddots & \\
 & &  &  & \gamma \lambda_1^{t-1} &  &  &  & & \gamma\lambda_2^{t-1} &
 & \lambda_1^{t-1} &  \lambda_2^{t-1} &  \lambda_3^{t-1} & \ldots & \lambda_{d-k+1}^{t-1}\\
\end{array}\right),
\end{aligned}
$}
\end{equation}
}where the empty positions in (\ref{Ye-block}) represent zeros.
And $\gamma$ only appears in the $i$-th diagonal entry of the block matrix $H_{t,j}$ for all $i,j\in[d-k+1]$ and $i>j$.
}
\end{construction}

Next, we illustrate that the basic code structure of \cite{Ye2016sub-} in Construction \ref{Ye-basic} can be derived from our transformation $\mathcal{T}_2$.

For simplicity, we index the rows and columns of the $(d-k+1)\times(d-k+1)$ block matrix $H_{t,j}$ in \eqref{Ye-block} from $1$ to $d-k+1$. And represent the codewords as $\bm{c}=(\bm{c}_1,\ldots,\bm{c}_n)$, where $\bm{c}_i=(c_{i,1},\ldots,c_{i,d-k+1})$ for $i\in[n]$.
Define ${\rm Van}_{r\times n}(\lambda_1,\ldots,\lambda_n)$ to be the $r\times n$ Vandermonde matrix with the $(t,j)$-th entry being $\lambda_j^{t-1}$ for $t\in[r]$, $j\in[n]$.
Denote by $\mathcal{C}_{{\rm GRS}}$ the $[n,k]$ generalized Reed-Solomon code with parity-check matrix ${\rm Van}_{r\times n}(\lambda_1,\ldots,\lambda_n)$.
Actually, according to the parity-check equations $H\cdot\bm{c}^{\top}=\bm{0}$ of Construction \ref{Ye-basic}, it has that the following $d-k+1$ vectors in \eqref{Ye-codeword} form $d-k+1$ codewords of the GRS code $\mathcal{C}_{{\rm GRS}}$:
\begin{equation}\label{Ye-codeword}
 \begin{cases}
 (c_{1,1}, c_{2,1}+c_{1,2},\ldots, c_{d-k+1,1}+c_{1,d-k+1},c_{d-k+2,1},\ldots,c_{n,1}), \\
 (\gamma c_{1,2}+c_{2,1}, c_{2,2},\ldots, c_{d-k+1,2}+c_{2,d-k+1},c_{d-k+2,2},\ldots, c_{n,2}), \\
 ~~\vdots \\
 (\gamma c_{1,d-k+1}+c_{d-k+1,1}, \gamma c_{2,d-k+1}+c_{d-k+1,2},\ldots,c_{d-k+1,d-k+1},c_{d-k+2,d-k+1},\ldots, c_{n,d-k+1}).
 \end{cases}
\end{equation}

For $i\in[d-k+1]$, denote the $i$-th vector in \eqref{Ye-codeword} as $\bm{f}^{(i)}=(f_1^{(i)}, f_2^{(i)}, \ldots, f_n^{(i)})\in\mathcal{C}_{{\rm GRS}}$.
From \eqref{Ye-codeword}, one can solve out the following stored data
\[
 \begin{cases}
c_{i,i}=f_i^{(i)}, & {\rm for}~i\in[d-k+1], \\
 c_{i,j}=\frac{1}{\gamma-1}f_i^{(j)}-\frac{1}{\gamma-1}f_j^{(i)}, & {\rm for}~i,j\in[d-k+1], i<j, \\
 c_{j,i}=\frac{\gamma }{\gamma-1}f_j^{(i)}-\frac{1}{\gamma-1}f_i^{(j)}, & {\rm for}~i,j\in[d-k+1], i<j, \\
 c_{i,j}=f_i^{(j)},  & {\rm for}~i\in[d-k+2,n],j\in[d-k+1]. \end{cases}
\]
Recall Table \ref{tab3} of transformation $\mathcal{T}_2$ that when $h=1$, for all $1\leq i<j\leq d-k+1$, the stored data $c_{i,j}=\eta_{i,j}^{(3)}\bm{f}_{j}^{(i)}+\eta_{i,j}^{(4)}\bm{f}_{i}^{(j)}$ and $c_{j,i}=\eta_{i,j}^{(1)}\bm{f}_{j}^{(i)}+\eta_{i,j}^{(2)}\bm{f}_{i}^{(j)}$.
Thus,
 the transformation $\mathcal{T}_2$ degenerates to the basic code structure in Construction \ref{Ye-basic} by setting $\mathcal{C}_0=\mathcal{C}_{{\rm GRS}}$, $h=1$, $\eta_{i,j}^{(1)}=\frac{\gamma }{\gamma-1}$, $\eta_{i,j}^{(2)}=-\frac{1}{\gamma-1}$, $\eta_{i,j}^{(3)}=-\frac{1}{\gamma-1}$ and $\eta_{i,j}^{(4)}=\frac{1}{\gamma-1}$.

\item[(3)] {\it Connection with the cooperative MSR code by Zhang, Zhang \& Wang \cite{Zhang2020}:}\\
Recall the optimal-access cooperative MSR code given by Zhang, Zhang \& Wang \cite{Zhang2020}.
It is easy to verify that \cite{Zhang2020} is exactly the code by extending only the  basic code $\mathcal{C}_{\mathrm{I}}$ in Section \ref{sec0-building-block} for $\binom{n}{h}$ dimensions.

\item[(4)] {\it Connections of the transformation method and parity-check matrix construction:}\\
Let $\mathcal{C}_{{\rm GRS}}$ denote the $[n,k]$ generalized RS code with the Vandermonde  parity-check matrix ${\rm Van}_{r\times n}(\lambda_1,\ldots,\lambda_n)$.
Let $\mathcal{\tilde{C}}_{{\rm GRS}}$ be the $[n+d-k,d]$ generalized RS code with the Vandermonde  parity-check matrix ${\rm Van}_{r\times (n+d-k)}(\lambda_1,\ldots,\lambda_n,\gamma_1,\ldots,\gamma_{d-k})$.
Then $\mathcal{C}_{{\rm GRS}}$ is the shortened code of $\mathcal{\tilde{C}}_{{\rm GRS}}$ by shortening codewords at the last $d-k$ nodes.

{\it $(a)$ Connection between $\mathcal{C}_{\mathrm{I}}$ and $\mathcal{T}_1$}. Recall the code $\mathcal{C}_{\mathrm{I}}$ in Construction \ref{construction-1}.
According to Construction \ref{construction-1} and by the parity-check equations $H\cdot\bm{c}^{\top}=\bm{0}$ of Construction \ref{construction-1}, it has that
the following $h$ vectors in \eqref{Code1-codeword-1} are codewords of the $[n+d-k,d]$ GRS code $\mathcal{\tilde{C}}_{{\rm GRS}}$, and the $d-k$ vectors in \eqref{Code1-codeword-2} are codewords of the $[n,k]$ GRS code $\mathcal{C}_{{\rm GRS}}$.
\begin{align}
&\begin{cases}
 (c_{1,1}, c_{2,1},\ldots, c_{n,1},c_{1,h+1},\ldots,c_{1,d-k+h}) \\
  (c_{1,2}, c_{2,2},\ldots, c_{n,2},c_{2,h+1},\ldots,c_{2,d-k+h}) \\
 ~~\vdots \\
 (c_{1,h}, c_{2,h},\ldots,c_{n,h},c_{h,h+1},\ldots, c_{h,d-k+h})
 \end{cases}, \label{Code1-codeword-1} \\
&\begin{cases}
(c_{1,h+1}, c_{2,h+1},\ldots, c_{n,h+1}) \\
 ~~\vdots \\
 (c_{1,d-k+h}, c_{2,d-k+h},\ldots,c_{n,d-k+h})
 \end{cases}. \label{Code1-codeword-2}
\end{align}

Recall the transformation $\mathcal{T}_1$, and set the code $\mathcal{C}_0=\mathcal{\tilde{C}}_{{\rm GRS}}$.
By applying transformation $\mathcal{T}_1$ and then shortening the obtained code at the last $d-k$ nodes, it yields the code $\mathcal{C}_{\mathrm{I}}$ of Construction \ref{construction-1} in Section \ref{sec0-building-block}.

{\it $(b)$ Connection between $\mathcal{C}_{\mathrm{II}}$ and $\mathcal{T}_2$}.
Recall the code $\mathcal{C}_{\mathrm{II}}$ in Construction \ref{construction-2}.
Through a similar analysis as in the discussion $(2)$ of Subsection \ref{comparison}, Then the code $\mathcal{C}_{\mathrm{II}}$ in Section \ref{sec0-building-block} can be derived from transformation $\mathcal{T}_2$ by setting $\mathcal{C}_0=\mathcal{C}_{{\rm GRS}}$, $\eta_{i,j}^{(1)}=-\frac{1}{\tau-1}$, $\eta_{i,j}^{(2)}=\frac{1}{\tau-1}$, $\eta_{i,j}^{(3)}=\frac{\tau}{\tau-1}$ and $\eta_{i,j}^{(4)}=-\frac{1}{\tau-1}$, where $\tau\in F\setminus\{0,1\}$ is defined in Construction \ref{construction-2}.
\end{itemize}

\subsection{Conclusion}
We present new constructions of optimal-access cooperative MSR codes with sub-packetization $\ell=\delta^m$, where $\delta=d-k+h$ and $m=\binom{n}{h}-\lfloor\frac{n}{\delta}\rfloor(\binom{\delta}{h}-1)$, which reduces $\ell$ by a fraction of $1/\delta^{\lfloor\frac{n}{\delta}\rfloor(\binom{\delta}{h}-1)}$ compared with the state of the art.
Precisely, the first approach gives a direct explicit construction via the parity-check matrix construction.
The second approach provides a generic transformation for converting an MDS scalar code to an optimal-access cooperative MSR code.
It is worth noting that the transformation can be used to construct codes with systematic nodes repair only, and has more flexibility for code construction.
We also show that our code constructions include some previous works as special cases, such as \cite{Zhang2020,Ye2016sub-,JLi-2018}.
Despite the sub-packetization in our work is still exponential, it is possible to use this code and another large-distance code to build $\epsilon$-cooperative MSR codes with small data access and small sub-packetization as in \cite{Devi2019}.

Moreover, the codes derived in both approaches deploy different repair strategies for intra-group erasure patterns and inter-group erasure patterns.
That is, for intra-group erasure patterns, the repair process requires $d-k$ local helper nodes' participation and multi-round collaboration. While for inter-group erasure patterns, the repair proceeds with any $d$ helper nodes within one round of collaboration. Indeed, we claim that inter-group erasure patterns account for the majority.
For example, set $n=12$, $k=9$, $h=2$ and $d=10$. Then inter-group erasure patterns constitute $82\%$ and intra-group erasure patterns constitute $18\%$.
Nevertheless, a future problem is to construct codes for repairing all erasure patterns with any $d$ helper nodes and one-round collaboration.
Besides, since the sub-packetization is still large, another open problem is to establish a lower bound on the sub-packetization of cooperative MSR codes and give a matching construction.

\appendices

\section{Proof of Lemma \ref{stage-lem}}\label{appendix-1}

In order to prove Lemma \ref{stage-lem}, W.L.O.G., we prove the case $j=0$ that in every stage $s\in[0,m-\lfloor\frac{n}{\delta}\rfloor]$, node $i_0$ can recover the data $\{c_{i_0,a}: a\in\mathcal{L}_s\}$.
To this end, we will firstly prove that in the initial stage $s=0$, node $i_0$ using the downloaded data and collaborated data at stage $0$ can recover the data $\{c_{i_0,a}: a\in\mathcal{L}_0\}$. Then, for any stage $s\in[0,m-\lfloor\frac{n}{\delta}\rfloor]$, the data recovery is proved by induction on $s$.
In the following,
let $\mathcal{R}_{i_0}$ be the set of $d$ helper nodes connected by node $i_0$.
Recall that $\mathcal{R}_{i_0}$ satisfies $([(u^*-1)\delta+1,u^*\delta]\setminus\mathcal{H})\subseteq\mathcal{R}_{i_0}$. Denote $\mathcal{\tilde{R}}_{i_0}=\mathcal{R}_{i_0}\setminus([(u^*-1)\delta+1,u^*\delta]\setminus\mathcal{H})$, then $|\mathcal{\tilde{R}}_{i_0}|=k$.
Next, we need a definition, which is inherited from Definition \ref{partition}.

\begin{definition}\label{def-R}
Recall the second partition of $[0,\ell-1]$ in Definition \ref{partition}.
Define $I=\mathcal{\tilde{R}}_{i_0}\cap[\delta\lfloor\frac{n}{\delta}\rfloor]$.
Define $U_I=\{u\in[\lfloor\frac{n}{\delta}\rfloor]: \exists~v\in[0,\delta-1] ~{\rm{s.t.}}~ (u,v)\in I\}$.
And for each $u\in[\lfloor\frac{n}{\delta}\rfloor]$, define $V_u=\{v\in[0,\delta-1]: (u,v)\in I\}$.
Then $I=\cup_{u\in U_I}\{(u,v): v\in V_u\}$.
By Definition \ref{partition}, $\Lambda_0,\Lambda_1,\ldots,\Lambda_{|U_I|}$ w.r.t. $I$ form a partition of $[0,\ell-1]$.
Moreover, for each $s\in[0,m-\lfloor\frac{n}{\delta}\rfloor]$, the set $\mathcal{L}_s\cap A(u^*,v_0)$ has a partition $\mathcal{L}_s\cap A(u^*,v_0)\cap \Lambda_q$ for $q=0,1,\ldots, |U_I|$.
In particular, when $I=\emptyset$, then $U_I=\emptyset$, and the partition of $[0,\ell-1]$ degenerates to $\Lambda_0=[0,\ell-1]$.
\end{definition}

Here we note that $u^*\notin U_I$ since $\mathcal{\tilde{R}}_{i_0}\cap[(u^*-1)\delta+1,u^*\delta]=\emptyset$.

By Definition \ref{def-R}, we know that for each $s\in[0,m-\lfloor\frac{n}{\delta}\rfloor]$, the set $\mathcal{L}_s\cap A(u^*,v_0)$ has a partition $\mathcal{L}_s\cap A(u^*,v_0)\cap \Lambda_q$ for $q=0,1,\ldots, |U_I|$.
Besides, recall the definition of $c_{i,a}^*$'s in \eqref{equation-form-matrix-s} of Lemma \ref{equation-s}.
We claim that $\{c_{i,a}^*: i\in[n], a\in\mathcal{L}_s\cap A(u^*,v_0)\}$ can be computed by node $i_0$ at stage $s$, which will be proved later in Lemma \ref{lem-recover1-step1}.
Before that, we prove that for each $s\in[0,m-\lfloor\frac{n}{\delta}\rfloor]$ and  $q\in[0,|U_I|]$, if node $i_0$ knows the data $\{c_{i,a}^*: i\in[n], a\in\mathcal{L}_s\cap A(u^*,v_0)\cap\Lambda_q\}$, then $i_0$ can compute the data in \eqref{recover1-step2-eq1} and \eqref{recover1-step2-eq2} as illustrated in Lemma \ref{lem-recover1-step2}.

\begin{lemma}\label{lem-recover1-step2}
For each $s\in[0,m-\lfloor\frac{n}{\delta}\rfloor]$ and  $q\in[0,|U_I|]$, if node $i_0$ knows the data $\{c_{i,a}^*: i\in[n], a\in\mathcal{L}_s\cap A(u^*,v_0)\cap\Lambda_q\}$, then along with the downloaded data
node $i_0$ can compute the following failed data:
\begin{equation}\label{recover1-step2-eq1}
\begin{aligned}
\bigcup_{a\in\mathcal{L}_s\cap A(u^*,v_0)\cap\Lambda_q}\Big(
&
\{c_{i_0,a(u^*,v)}: v\in[0,\delta-1]\setminus\{v_1,v_{2},\ldots,v_{h-1}\}\}        \\
&\cup\{f(v_0,v) c_{i_0,a(u^*,v)}+c_{(u^*,v),a}: v\in\{v_1,v_{2},\ldots,v_{h-1}\}\}
\Big).
\end{aligned}
\end{equation}
and the following intermediate data:
\begin{equation}\label{recover1-step2-eq2}
\cup_{a\in\mathcal{L}_s\cap A(u^*,v_0)\cap\Lambda_q}
\{c_{p,a}:~p\in[n]\setminus\{i_0,i_1,\ldots, i_{h-1}\}\}.
\end{equation}
\end{lemma}

\begin{proof}
Let $s\in[0,m-\lfloor\frac{n}{\delta}\rfloor]$ and  $q\in[0,|U_I|]$.
Recall the definition of $c_{i,a}^*$'s in \eqref{equation-form-matrix-s} of Lemma \ref{equation-s}.
Suppose node $i_0$ knows the data $\{c_{i,a}^*: i\in[n], a\in\mathcal{L}_s\cap A(u^*,v_0)\cap\Lambda_q\}$.
We prove node $i_0$ can recover the data in \eqref{recover1-step2-eq1} and \eqref{recover1-step2-eq2}.

For every $a\in\mathcal{L}_s\cap A(u^*,v_0)\cap\Lambda_q$ , it has $a_{u^*}=v_0$, then $c_{i_0,a}^*=c_{i_0,a}$ and $c_{(u^*,v),a}^*=c_{(u^*,v),a}+c_{i_0,a(u^*,v)}f(v_0,v)$ for $v\in[0,\delta-1]\setminus\{v_0\}$.
Note that  $\{(u^*,v): v\in[0,\delta-1]\setminus\{v_0,v_1,\ldots,v_{h-1}\}\}\subseteq\mathcal{R}_{i_0}$.
Then, $\{c_{(u^*,v),a}:
~v\in[0,\delta-1]\setminus\{v_0,v_1,\ldots,v_{h-1}\}, a\in\mathcal{L}_s\cap A(u^*,v_0)\cap\Lambda_q\}$ are downloaded data.
Thus node $i_0$ can compute $c_{i_0,a(u^*,v)}$ from $c_{(u^*,v),a}^*$ for $v\in[0,\delta-1]\setminus\{v_1,v_{2},\ldots,v_{h-1}\}$, hence can recover the data in \eqref{recover1-step2-eq1}.

As for the data in \eqref{recover1-step2-eq2}, since $\{c_{p,a}: p\in\mathcal{R}_{i_0}, a\in\mathcal{L}_s\cap A(u^*,v_0)\cap\Lambda_q\}$ are downloaded data,
then it remains to prove that node $i_0$ can recover $\cup_{a\in\mathcal{L}_s\cap A(u^*,v_0)\cap\Lambda_q}\{c_{p,a}: p\in[n]\setminus(\mathcal{H}\cup\mathcal{R}_{i_0})\}$.
Note that the set
\[\begin{aligned}
&[n]\setminus(\mathcal{H}\cup\mathcal{R}_{i_0})\\
=&[n]\setminus([(u^*-1)\delta+1,u^*\delta]\cup\mathcal{\tilde{R}}_{i_0})) \\
=&[n]\setminus\big([(u^*-1)\delta+1,u^*\delta]\cup I\cup (\mathcal{\tilde{R}}_{i_0}\cap[\delta\lfloor n/\delta \rfloor+1,n])\big)   \\
=&\{(u,v):~u\in[\lfloor n/\delta\rfloor]\setminus\{u^*\}, v\in[0,\delta-1]\setminus V_u\}
\cup
([\delta\lfloor n/\delta \rfloor+1,n]\setminus\mathcal{\tilde{R}}_{i_0}).
\end{aligned}\]
At first, for each $i\in[\delta\lfloor n/\delta \rfloor+1,n]$, it has $c_{i,a}^*=c_{i,a}$.
Then node $i_0$ using the known data can directly obtain the data
$\cup_{a\in\mathcal{L}_s\cap A(u^*,v_0)\cap\Lambda_q}\{c_{p,a}:~p\in[\delta\lfloor n/\delta \rfloor+1,n]\setminus\mathcal{\tilde{R}}_{i_0}\}$.
Then, for every $u\in[\lfloor n/\delta\rfloor]\setminus\{u^*\}$ and $v\in[0,\delta-1]\setminus V_u$, we consider the recovery of $c_{(u,v),a}$.
For each $a\in\mathcal{L}_s\cap A(u^*,v_0)\cap\Lambda_q$, if $a_u=v$, then it has $c_{(u,v),a}^*=c_{(u,v),a}$, which can be recovered.
If $a_u\neq v$ and $a_u\in V_u$, this implies $(u,a_u)\in I$.
Then $c_{(u,v),a}^*=c_{(u,v),a}+c_{(u,a_u),a(u,v)}f(a_u,v)$, and node $i_0$ can compute $c_{(u,v),a}$ since the symbol $c_{(u,a_u),a(u,v)}$ is downloaded data.
If $a_u\neq v$ and $a_u\notin V_u$, then $(u,a_u)\notin I$.
Combining the known data $c_{(u,v),a}^*=c_{(u,v),a}+c_{(u,a_u),a(u,v)}f(a_u,v)$ and  $c_{(u,a_u),a(u,v)}^*=c_{(u,a_u),a(u,v)}+c_{(u,v),a}f(v,a_u)$, and observing that $\{f(a_u,v),f(v,a_u)\}=\{1,\tau\}$ and the fact that $a(u,v)\in\mathcal{L}_s\cap A(u^*,v_0)\cap\Lambda_q$, then one can solve out the symbol $c_{(u,v),a}$.
Therefore, node $i_0$ can recover the data in \eqref{recover1-step2-eq2}.
\end{proof}

Now, we are left to prove in every stage $s\in[0,m-\lfloor\frac{n}{\delta}\rfloor]$, node $i_0$ using downloaded data and previously collaborated data can compute the data $\{c_{i,a}^*: i\in[n], a\in\mathcal{L}_s\cap A(u^*,v_0)\}$.
For simplicity, we prove the initial stage $s=0$ in Lemma \ref{lem-recover1-step1}, and the stage $s>0$ is illustrated in Lemma \ref{cor-lem}.

\begin{lemma}\label{lem-recover1-step1}
In the stage $s=0$, node $i_0$ using the downloaded data, can recover the following data
\begin{equation}\label{lem-recover1-step1-eq}
\{c_{i,a}^*: i\in[n], a\in\mathcal{L}_0\cap A(u^*,v_0)\},
\end{equation}
where the $c_{i,a}^*$'s are defined in \eqref{equation-form-matrix-s} of Lemma \ref{equation-s}.
\end{lemma}

\begin{proof}
For every $a\in \mathcal{L}_0\cap A(u^*,v_0)$, it has $a_{u^*}=v_0$ and $\{u\in[\lfloor\frac{n}{\delta}\rfloor+1,m]: a_u\in[0,h-1]\}=\emptyset$.
Consider the $a$-th row of parity-check equations $(H_{t,1},H_{t,2},\ldots,H_{t,n})\cdot\bm{c}^{\top}=\bm{0}$, $t\in[r]$.
According to Lemma \ref{equation-s}, one can obtain $r$ simplified equations as illustrated in \eqref{equation-form-s} and the matrix form in \eqref{equation-form-matrix-s}, where we note that part $(3)$ of \eqref{equation-form-s} does not exist.
Thus the vector symbols $(c_{1,a}^*,\ldots,c_{n,a}^*)$ in \eqref{equation-form-matrix-s} forms an $[n,k]$ generalized Reed-Solomon (GRS) codeword.
This implies any $k$ symbols in $\{c_{i,a}^*: i\in[n]\}$ can reconstruct all the $n$ symbols.
Next, we show that node $i_0$ using the downloaded data can recover $\{c_{i,a}^*: i\in\mathcal{\tilde{R}}_{i_0}\}$ where $|\mathcal{\tilde{R}}_{i_0}|=k$, hence can recover $\{c_{i,a}^*: i\in[n]\}$ for all $a\in \mathcal{L}_0\cap A(u^*,v_0)$.

To this end, recall in Definition \ref{def-R} that $\mathcal{L}_0\cap A(u^*,v_0)\cap\Lambda_q$, $q=0,1,\ldots, |U_I|$ form a partition of $\mathcal{L}_0\cap A(u^*,v_0)$.
In the following, we prove by induction on $q$ that for each $q=|U_I|, |U_I|-1, \ldots,0$, the following set of symbols in \eqref{appendix1-eq1} can be recovered from the downloaded data of node $i_0$.
\begin{equation}\label{appendix1-eq1}
\{c_{i,a}^*, i\in\mathcal{\tilde{R}}_{i_0}, a\in \mathcal{L}_0\cap A(u^*,v_0)\cap\Lambda_q\}.
\end{equation}
Firstly, let $q=|U_I|$. Consider every $a\in\mathcal{L}_0\cap A(u^*,v_0)\cap\Lambda_{|U_I|}$, by the definition of $\Lambda_{|U_I|}$, it has for all $u\in U_I$, $(u, a_{u})\in I$.
According to the definition of $c_{i,a}^*$'s,
it has $c_{(u,a_{u}),a}^*=c_{(u,a_{u}),a}$ for
$u\in U_I$ which is the downloaded data. For each  $u\in U_I$ and $v\in V_{u}\setminus\{a_{u}\}$, the data $c_{(u,v),a}^*=c_{(u,v),a}+c_{(u,a_{u}),a(u,v)}f(a_{u},v)$, which can be computed since $c_{(u,v),a}$ and $c_{(u,a_{u}),a(u,v)}$ are both downloaded data by noticing that node $(u,v)\in\mathcal{R}_{i_0}$ and $a(u,v)\in\mathcal{L}_0\cap A(u^*,v_0)$.
That is, node $i_0$ can compute the data $\{c_{i,a}^*, i\in I\}$ where $I=\cup_{u\in U_I}\{(u,v): v\in V_u\}$.
Moreover, since $(\mathcal{\tilde{R}}_{i_0}\setminus I)\subseteq[\delta\lfloor\frac{n}{\delta}\rfloor+1,n]$, then $c_{i,a}^*=c_{i,a}$ for all $i\in \mathcal{\tilde{R}}_{i_0}\setminus I$, which are exactly the downloaded data.
Thus, node $i_0$ can compute all the data $\{c_{i,a}^*, i\in\mathcal{\tilde{R}}_{i_0}, a\in\mathcal{L}_0\cap A(u^*,v_0)\cap\Lambda_{|U_I|}\}$.

Now suppose for all $q=|U_I|, |U_I|-1, \ldots, q_0+1$, node $i_0$ has recovered the data in \eqref{appendix1-eq1}.
Next we prove the case $q=q_0$ that the data in \eqref{appendix1-eq1} can be recovered.

Consider every $a\in\mathcal{L}_0\cap A(u^*,v_0)\cap\Lambda_{q_0}$, then there are exactly $q_0$ integers in $U_I$ s.t. $(u,a_{u})\in I$, denoted by $\{u_1,\ldots,u_{q_0}\}\subseteq U_I$.
Note that $I=\cup_{u\in U_I}\{(u,v): v\in V_u\}$.
We firstly show the data recovery of $\{c_{i,a}^*, i\in I, a\in \mathcal{L}_0\cap A(u^*,v_0)\cap\Lambda_{q_0}\}$.
For each $u\in\{u_1,\ldots,u_{q_0}\}$, it has $c_{(u,a_{u}),a}^*=c_{(u,a_{u}),a}$, and $c_{(u,v),a}^*=c_{(u,v),a}+c_{(u,a_{u}),a(u,v)}f(a_{u},v)$ for $v\in V_{u}\setminus\{a_{u}\}$.
These data $c_{(u,v),a}^*$'s can be computed by node $i_0$ since the nodes $(u,v)$'s are helper nodes in $I$ and $a(u,v)\in\mathcal{L}_0\cap A(u^*,v_0)$ for all $u\in\{u_1,\ldots,u_{q_0}\}$ and $v\in V_{u}$.
Besides, for each $u\in U_I\setminus\{u_1,\ldots,u_{q_0}\}$ and $v\in V_{u}$, it has $(u, a_{u})\notin I$. This implies $a_u\notin V_u$ and  $a_u\neq v$.
Thus, it has $c_{(u,v),a}^*=c_{(u,v),a}+c_{(u,a_{u}),a(u,v)}f(a_{u},v)$.
We claim that these data $c_{(u,v),a}^*$'s can also be computed.
On the one hand, the data $c_{(u,v),a}$'s are all downloaded data.
On the other hand, for $u\in U_I\setminus\{u_1,\ldots,u_{q_0}\}$ and $v\in V_{u}$, it has $a(u,v)\in\mathcal{L}_0\cap A(u^*,v_0)\cap\Lambda_{q_0+1}$. Then by the hypothesis and Lemma \ref{lem-recover1-step2} for the case $s=0$ and $q=q_0+1$, the data $c_{(u,a_{u}),a(u,v)}$'s can be recovered.
Therefore, the data $\{c_{i,a}^*, i\in I, a\in \mathcal{L}_0\cap A(u^*,v_0)\cap\Lambda_{q_0}\}$ can be computed.
Next, we show the data recovery of $\{c_{i,a}^*, i\in\mathcal{\tilde{R}}_{i_0}\setminus I, a\in \mathcal{L}_0\cap A(u^*,v_0)\cap\Lambda_{q_0}\}$.
This is straightforward, since $(\mathcal{\tilde{R}}_{i_0}\setminus I)\subseteq[\delta\lfloor\frac{n}{\delta}\rfloor+1,n]$, and $c_{i,a}^*=c_{i,a}$ for all $i\in \mathcal{\tilde{R}}_{i_0}\setminus I$, which are exactly the downloaded data.

Thus, in the stage $s=0$ node $i_0$ can recover all the data in \eqref{lem-recover1-step1-eq}.
\end{proof}

Next, we give the following corollary to illustrate that node $i_0$ can recover $\{c_{i_0,a}: a\in\mathcal{L}_0\}$ in stage $s=0$.

\begin{corollary}\label{cor-recover1}
In the stage $s=0$, node $i_0$ using the downloaded data and collaborated data can recover the following data:
\begin{equation}\label{cor-recover1-eq}
\{c_{i_0,a}: a\in\mathcal{L}_0\}
\cup\{c_{p,a}: p\in[n]\setminus\{i_0\}, a\in\mathcal{L}_0\cap A(u^*,v_0)\}.
\end{equation}
\end{corollary}

\begin{proof}
In the stage $s=0$, according to Lemma \ref{lem-recover1-step1} and Lemma \ref{lem-recover1-step2} of data recovery of node $i_0$,
in a similar way, for each $j\in[0,h-1]$, node $i_j$ using the downloaded data can recover the following data:
\begin{equation}
\begin{aligned}
\bigcup_{a\in\mathcal{L}_0\cap A(u^*,v_j)}\Big(
&
\{c_{i_j,a(u^*,v)}: v\in[0,\delta-1]\setminus\{v_0,\ldots,v_{j-1},v_{j+1},\ldots,v_{h-1}\}\}        \\
&\cup\{f(v_j,v) c_{i_j,a(u^*,v)}+c_{(u^*,v),a}: v\in\{v_0,\ldots,v_{j-1},v_{j+1},\ldots,v_{h-1}\}\} \\
&\cup\{c_{p,a}:~p\in[n]\setminus\{i_0,i_1,\ldots, i_{h-1}\}\}
\Big),
\end{aligned}
\end{equation}
similar to \eqref{recover1-step2-eq1} and \eqref{recover1-step2-eq2}.
Then in the collaboration of stage $0$, for $j\in[h-1]$, node $i_j$ transmits $\{f(v_j,v_0) c_{i_j,a(u^*,v_0)}+c_{i_0,a}: a\in\mathcal{L}_0\cap A(u^*,v_j)\}$ to node $i_0$.
Then along with the recovered data in \eqref{recover1-step2-eq1} and \eqref{recover1-step2-eq2}, node $i_0$ can compute the data in \eqref{cor-recover1-eq}.
\end{proof}

In Lemma \ref {cor-lem}, we further prove the data recovery of node $i_0$ in
every stage $s\in[0,m-\lfloor\frac{n}{\delta}\rfloor]$ by induction on $s$.

\begin{lemma}\label{cor-lem}
In every stage $s\in[0,m-\lfloor\frac{n}{\delta}\rfloor]$, node $i_0$ can recover the following data by using the downloaded data and collaborated data from other failed nodes.
\begin{equation}\label{cor-recover1-eq-s}
\{c_{i_0,a}: a\in\mathcal{L}_s\}
\cup\{c_{p,a}: p\in[n]\setminus\{i_0\}, a\in\mathcal{L}_s\cap A(u^*,v_0)\}.
\end{equation}
\end{lemma}

\begin{proof}
According to Corollary \ref{cor-recover1}, the case of stage $s=0$ has been proved.
Now suppose in the stages $s-1$, node $i_0$ has recovered the data $\{c_{i_0,a}: a\in\mathcal{L}_{s-1}\}
\cup\{c_{p,a}: p\in[n]\setminus\{i_0\}, a\in\mathcal{L}_{s-1}\cap A(u^*,v_0)\}$.
Next, we prove in the stage $s$, node $i_0$ can recover the data $\{c_{i_0,a}: a\in\mathcal{L}_s\}
\cup\{c_{p,a}: p\in[n]\setminus\{i_0\}, a\in\mathcal{L}_s\cap A(u^*,v_0)\}$.

For every $a\in\mathcal{L}_s\cap A(u^*,v_0)$, it has $a_{u^*}=v_0$ and $|\{u\in[\lfloor\frac{n}{\delta}\rfloor+1,m]: a_u\in[0,h-1]\}|=s$. Denote $\{u\in[\lfloor\frac{n}{\delta}\rfloor+1,m]: a_u\in[0,h-1]\}=\{\sigma_1,\ldots,\sigma_s\}$ for simplicity.
Consider the $a$-th row of parity-check equations $(H_{t,1},H_{t,2},\ldots,H_{t,n})\cdot\bm{c}^{\top}=\bm{0}$, $t\in[r]$.
According to Lemma \ref{equation-s}, one can obtain $r$ simplified equations as illustrated in \eqref{equation-form-s} and the matrix form in \eqref{equation-form-matrix-s}.
Note that for $z\in[s]$ and $v\in[h,\delta-1]$, it has $a(\sigma_z,v)\in\mathcal{L}_{s-1}\cap A(u^*,v_0)$. By the hypothesis, the data $c_{j,a(\sigma_z,v)}$ for $j\in[n]$, $z\in[s]$ and $v\in[h,\delta-1]$ have been recovered.
Therefore, part $(3)$ of \eqref{equation-form-s} is known and can be moved to the right hand side of \eqref{equation-form-s}.
That is, if node $i_0$ knowns the data $\{c_{i,a}^*: i\in\mathcal{\tilde{R}}_{i_0}, a\in\mathcal{L}_s\cap A(u^*,v_0)\}$, then it can recover $\{c_{i,a}^*: i\in[n], a\in\mathcal{L}_s\cap A(u^*,v_0)\}$.
Actually, in a similar way as in the proof of Lemma \ref{lem-recover1-step1} for the case $s=0$, node $i_0$ using the downloaded data and previously recovered data can compute $\{c_{i,a}^*: i\in\mathcal{\tilde{R}}_{i_0}, a\in\mathcal{L}_s\cap A(u^*,v_0)\}$, hence $\{c_{i,a}^*: i\in[n], a\in\mathcal{L}_s\cap A(u^*,v_0)\}$.

Furthermore, by Lemma \ref{lem-recover1-step2}, node $i_0$ using the obtained data can further compute
\begin{equation*}
\begin{aligned}
\bigcup_{a\in\mathcal{L}_s\cap A(u^*,v_0)}\Big(
&
\{c_{i_0,a(u^*,v)}: v\in[0,\delta-1]\setminus\{v_1,\ldots,v_{h-1}\}\}        \\
&\cup\{f(v_0,v) c_{i_0,a(u^*,v)}+c_{(u^*,v),a}: v\in\{v_1,\ldots,v_{h-1}\}\} \\
&\cup\{c_{p,a}:~p\in[n]\setminus\{i_0,i_1,\ldots, i_{h-1}\}\}
\Big).
\end{aligned}
\end{equation*}
Then, in a similar way as in Corollary \ref{cor-recover1}, after collaborating data with the remaining $h-1$ failed nodes at stage $s$, node $i_0$ can finally recover $\{c_{i_0,a}: a\in\mathcal{L}_s\}
\cup\{c_{p,a}: p\in[n]\setminus\{i_0\}, a\in\mathcal{L}_s\cap A(u^*,v_0)\}$, as illustrated in \eqref{cor-recover1-eq-s}.
\end{proof}

Finally, according to Lemma \ref{cor-lem}, Lemma \ref{stage-lem} is straightforward.

\section{Proof of Lemma \ref{thm2-lem}}\label{appendix-2}
In order to prove Lemma \ref{thm2-lem}, W.L.O.G., we prove the case $j=0$ that for every $s\in[0,m-\lfloor\frac{n}{\delta}\rfloor]$, node $i_0$ using the downloaded data can recover the data
$\cup_{a\in A(\rho,0)}\{c_{i_0,a}, c_{i_1,a}, \ldots, c_{i_{h-1},a}, c_{i_0,a(\rho,h)}, \ldots, c_{i_0,a(\rho,\delta-1)}\}$.
Since $\rho\in[\lfloor\frac{n}{\delta}\rfloor+1,m]$,
then $A(\rho,0)\cap\mathcal{L}_s$, for $s=1, 2, \ldots, m-\lfloor\frac{n}{\delta}\rfloor$ form a partition of $A(\rho,0)$.
We will prove by induction on $s$ that for each $s\in[m-\lfloor\frac{n}{\delta}\rfloor]$,
node $i_0$ can recover
\[
\bigcup_{a\in A(\rho,0)\cap\mathcal{L}_s}\{c_{i_0,a}, c_{i_1,a}, \ldots, c_{i_{h-1},a}, c_{i_0,a(\rho,h)}, \ldots, c_{i_0,a(\rho,\delta-1)}\}.
\]
Let $\mathcal{R}_{i_0}$ be the set of $d$ helper nodes connected by node $i_0$.
Recall Definition \ref{def-R}.
We reuse the definitions and notations in Definition \ref{def-R} by replacing $\mathcal{\tilde{R}}_{i_0}$ with $\mathcal{R}_{i_0}$.

In the following, we firstly give two lemmas to illustrate some data recovery of node $i_0$.

\begin{lemma}\label{thm2repair-lem1}
For each $s\in[m-\lfloor\frac{n}{\delta}\rfloor]$ and  $q\in[0,|U_I|]$, suppose node $i_0$ knows the data $\{c_{i,a}^*: i\in[n], a\in\mathcal{L}_s\cap A(\rho,0)\cap\Lambda_q\}$, where $c_{i,a}^*$'s are defined in \eqref{equation-form-matrix-s} of Lemma \ref{equation-s}.
Then along with the downloaded data,
node $i_0$ can compute the following symbols:
\begin{equation}\label{thm2repair-lem1-eq1}
\{c_{i,a}: i\in[n], a\in\mathcal{L}_s\cap A(\rho,0)\cap\Lambda_q\}.
\end{equation}
\end{lemma}

\begin{proof}
Since $\{c_{i,a}: i\in\mathcal{R}_{i_0}, a\in\mathcal{L}_s\cap A(\rho,0)\cap\Lambda_q\}$ are downloaded data, then is remains to prove that $\{c_{i,a}: i\in[n]\setminus\mathcal{R}_{i_0}, a\in\mathcal{L}_s\cap A(\rho,0)\cap\Lambda_q\}$ can be computed.
Note that 
\[
[n]\setminus\mathcal{R}_{i_0}=
\{(u,v):~u\in[\lfloor n/\delta\rfloor], v\in[0,\delta-1]\setminus V_u\}
\cup
([\delta\lfloor n/\delta \rfloor+1,n]\setminus\mathcal{R}_{i_0})
\]
Then by a similar analysis as in the proof of Lemma \ref{lem-recover1-step2}, Lemma \ref{thm2repair-lem1} can be easily proved.
\end{proof}

Now, let $s=1$. We prove that node $i_0$ using the downloaded data can recover $\{c_{i,a}^*: i\in[n], a\in\mathcal{L}_1\cap A(\rho,0)\}$ and some failed data at node $i_0$, as illustrated in Lemma \ref{thm2repair-lem2}.

\begin{lemma}\label{thm2repair-lem2}
Let $s=1$. Node $i_0$ using the downloaded data can recover the following set of symbols:
\begin{equation}\label{thm2repair-lem2-eq1}
\bigcup_{a\in \mathcal{L}_1\cap A(\rho,0)}
\Big(
\{c_{i,a}^*: i\in[n]\}\cup\{c_{i_0,a(\rho,v)}: v\in[h,\delta-1]\}
\Big),
\end{equation}
where $c_{i,a}^*$'s are defined in \eqref{equation-form-matrix-s} of Lemma \ref{equation-s}.
\end{lemma}

\begin{proof}
For every $a\in \mathcal{L}_1\cap A(\rho,0)$, it has $a_{\rho}=0$ and $\{u\in[\lfloor\frac{n}{\delta}\rfloor+1,m]: a_u\in[0,h-1]\}=\{\rho\}$.
Consider the $a$-th row of parity-check equations $(H_{t,1},\ldots,H_{t,n})\bm{c}^{\top}=\bm{0}$, $t\in[r]$.
According to Lemma \ref{equation-s}, one can obtain $r$ simplified equations as illustrated in \eqref{equation-form-s} and the matrix form in \eqref{equation-form-matrix-s}.
Then one obtains that $(c_{1,a}^*,\ldots,c_{n,a}^*,c_{i_0,a(\rho,h)},\ldots,c_{i_0,a(\rho,\delta-1)})$ forms an $[n+d-k,d]$ GRS codeword, and any $d$ coordinates are able to recover the whole codeword.
Next, we prove that node $i_0$ using the downloaded data can recover $\{c_{i,a}^*: i\in\mathcal{R}_{i_0}, a\in \mathcal{L}_1\cap A(\rho,0)\}$ where $|\mathcal{R}_{i_0}|=d$, hence the data in \eqref{thm2repair-lem2-eq1}.

To this end, recall that $\mathcal{L}_1\cap A(\rho,0)\cap\Lambda_q$, $q=0,1,\ldots,|U_I|$ form a partition of $\mathcal{L}_1\cap A(\rho,0)$.
We prove by induction on $q=|U_I|,\ldots,0$ that $\cup_{a\in \mathcal{L}_1\cap A(\rho,0)\cap\Lambda_q}\{c_{i,a}^*: i\in\mathcal{R}_{i_0}\}$
can be recovered. This induction proof is similar as in the proof of Lemma \ref{lem-recover1-step1} and we omit it here.
\end{proof}

Based on Lemma \ref{thm2repair-lem1} and Lemma \ref{thm2repair-lem2}, we give the data recovery of node $i_0$ in the download phase.

\begin{lemma}\label{thm2repair-lem3}
For $s\in[m-\lfloor\frac{n}{\delta}\rfloor]$, node $i_0$ using the downloaded data can recover the following data
\begin{equation}\label{thm2repair-lem3-eq1}
\bigcup_{a\in \mathcal{L}_s\cap A(\rho,0)}
\Big(
\{c_{i,a}: i\in[n]\}\cup\{c_{i_0,a(\rho,v)}: v\in[h,\delta-1]\}
\Big),
\end{equation}
\end{lemma}

\begin{proof}
By Lemma \ref{thm2repair-lem2} and Lemma \ref{thm2repair-lem1}, the case $s=1$ can be proved.
Now suppose for all $s'\leq s-1$, node $i_0$ has recovered the data $\bigcup_{a\in \mathcal{L}_{s'}\cap A(\rho,0)}
\big(
\{c_{i,a}: i\in[n]\}\cup\{c_{i_0,a(\rho,v)}: v\in[h,\delta-1]\}
\big)
$.
Next, we prove the case $s'=s$ that node $i_0$ can recover the data illustrated in \eqref{thm2repair-lem3-eq1}.

For every $a\in \mathcal{L}_{s}\cap A(\rho,0)$, it has $a_{\rho}=0$ and $|\{u\in[\lfloor\frac{n}{\delta}\rfloor+1,m]: a_u\in[0,h-1]\}|=s$. Denote $\{u\in[\lfloor\frac{n}{\delta}\rfloor+1,m]: a_u\in[0,h-1]\}=\{\rho,\sigma_1,\ldots,\sigma_{s-1}\}$.
Consider the $a$-th row of parity-check equations $(H_{t,1},H_{t,2},\ldots,H_{t,n})\cdot\bm{c}^{\top}=\bm{0}$, $t\in[r]$.
According to Lemma \ref{equation-s}, one can obtain $r$ simplified equations as illustrated in \eqref{equation-form-s} and the matrix form in \eqref{equation-form-matrix-s}.
Note that for $z\in[s-1]$ and $v\in[h,\delta-1]$, it has $a(\sigma_z,v)\in\mathcal{L}_{s-1}\cap A(\rho,0)$. By the hypothesis, the data $c_{j,a(\sigma_z,v)}$ for $j\in[n]$, $z\in[s-1]$ and $v\in[h,\delta-1]$ have been recovered.
Therefore, these corresponding symbols in part $(3)$ of \eqref{equation-form-s} is known and can be moved to the right hand side of \eqref{equation-form-s}.
This implies that any $d$ symbols in $\{c_{i,a}^*: i\in[n]\}\cup\{c_{i_0,a(\rho,h)},\ldots,c_{i_0,a(\rho,\delta-1)}\}$ are able to reconstruct all the $n+d-k$ symbols.

Actually, as in the case $s=1$, node $i_0$ using the downloaded data can recover $\bigcup_{a\in \mathcal{L}_{s}\cap A(\rho,0)}\{c_{i,a}^*: i\in\mathcal{R}_{i_0}\}$, which can be similarly proved as in the proof of Lemma \ref{thm2repair-lem2}.
Then, node $i_0$ can further recover $\bigcup_{a\in \mathcal{L}_{s}\cap A(\rho,0)}\big(\{c_{i,a}^*: i\in[n]\}\cup\{c_{i_0,a(\rho,h)},\ldots,c_{i_0,a(\rho,\delta-1)}\}\big)$.
According to Lemma \ref{thm2repair-lem1}, node $i_0$ using the obtained data can recover $\bigcup_{a\in \mathcal{L}_{s}\cap A(\rho,0)}\big(\{c_{i,a}: i\in[n]\}\cup\{c_{i_0,a(\rho,h)},\ldots,c_{i_0,a(\rho,\delta-1)}\}\big)$.
Thus, Lemma \ref{thm2repair-lem3} is proved.
\end{proof}

According to Lemma \ref{thm2repair-lem3} and note the fact that $\mathcal{L}_{s}\cap A(\rho,0)$, $s\in[m-\lfloor\frac{n}{\delta}\rfloor]$ form a partition of $A(\rho,0)$, then Lemma \ref{thm2-lem} is straightforward.





%

\end{document}